\documentclass[12pt]{amsart}  

\usepackage{bbold}
\usepackage{amscd} 
\usepackage{amsmath,amssymb,amsthm}
\usepackage{mathtools}
\usepackage{mathrsfs}
\usepackage[hidelinks]{hyperref}
\usepackage{enumitem}   

\usepackage{lmodern}
\usepackage{microtype}
\usepackage{xcolor}

\hypersetup{colorlinks=true,linkcolor=blue!55!black,citecolor=blue!55!black,urlcolor=blue!55!black}

\newtheorem{theorem}{Theorem}[section]
\newtheorem{lemma}[theorem]{Lemma}
\newtheorem{proposition}[theorem]{Proposition}
\newtheorem{corollary}[theorem]{Corollary} 
\newtheorem{definition}[theorem]{Definition}
 
\newtheorem{remark}[theorem]{Remark} 
\numberwithin{equation}{section}

\def\B{\mathscr B}
\def\M{\mathbb M}

\def\H {\mathfrak H}
\def\F{\mathfrak F}

\def \IS{X}

\def\dom{\text{\rm dom}}
\def\ran{\text{\rm ran}}
\def\supp{\text{\rm supp}}

\def\RE{\mathbb R}
\def\CO{{\mathbb C}}
\def\o{$\bar{\text{\rm o}}$}

\def\NA{\mathbb N}

\def\ph*{\phi_\star}

\def\be{\begin{equation}}
\def\ee{\end{equation}}
\def\min{{\rm min}}
\def\max{{\rm max}}

\def\-{{\rm in}}
\def\+{{\rm ex}}

\def\S{\mathfrak D}

\newcommand{\closure}[2][3]{%
      {}\mkern#1mu\overline{\mkern-#1mu#2}}
\def\cH{\closure{H}}

\newcommand{\lambdastar}{\lambda_{\star}}

\renewcommand{\rho}{\varrho}
\newcommand{\Hni}{H^{free}_{\mu}}
\newcommand{\thr}{T}
\newcommand{\xiq}{\phi}
\newcommand{\new}[1]{{\color{blue} #1}}

\title{A renormalized Lee Model: resolvent and spectral analysis}

\author[]{Claudio Cacciapuoti}
\address{ DiSAT, Sezione di Matematica, Universit\`a dell'Insubria, via Valleggio 11, I-22100 Como, Italy}
\email{claudio.cacciapuoti@uninsubria.it}

\author[]{Federica Muscolino}
\address{Independent researcher}
\email{federica.muscolino@gmail.com}

\author[]{Diego Noja}
\address{Dipartimento di Matematica e Applicazioni, Universit\`a
 di Milano Bicocca, via R. Cozzi 55, 20126, Milano, Italy}
\email{diego.noja@unimib.it}

\author[]{Andrea Posilicano}
\address{DiSAT, Sezione di Matematica, Universit\`a dell'Insubria, via Valleggio 11, I-22100 Como, Italy}
\email{andrea.posilicano@uninsubria.it}

\begin{document}
\begin{abstract}
We construct and analyze a renormalized three-dimensional massive Lee model; its renormalized Hamiltonian is defined intrinsically through an
abstract construction which explicitly provides the operator domain, action and a Kre\u\i n-like
resolvent formula.  Conservation of the total excitation number reduces the
spectral analysis of the full Hamiltonian to a sector-wise study of the spectral profile of an explicit operator pencil, playing the role of an operator-valued Weyl function.  We construct essential spectrum components in a sector generated by lower excitation sectors and obtain the inclusion in the expected
sector-wise HVZ type formula; the full formula is proved in the first
two sectors and, under explicit conditions, in every sector.  Below the free or one-excitation threshold, the
point spectrum in every higher sector is described by a bounded
self-adjoint nonnegative Birman--Schwinger operator.  In every sector with
at least two excitations we prove absence of sub-threshold spectrum for arbitrary coupling when
the renormalized gap is not smaller than the boson mass, and under an explicit
weak-coupling condition when it is smaller.  In the latter case we also derive
a sufficient condition for the existence of a simple bound state in the
two-excitation sector.
\end{abstract}
\maketitle

\tableofcontents


\section{Introduction}

Quantum models in which particles may be emitted and absorbed require a state
space accommodating different particle numbers.  For identical bosons in
three dimensions this is the bosonic Fock space
\begin{equation*}\label{eq:intro-fock}
  \F_b
  =\bigoplus_{n=0}^{\infty}L_b^2(\RE^{3n}),
  \qquad L_b^2(\RE^0)=\CO.
\end{equation*}
The purpose of this paper is to construct and study a particularly explicit
ultraviolet-singular model of this kind: a nonrelativistic version of the Lee model introduced in
\cite{Lee54}.  A fixed source has two internal states and is coupled to a
field of massive nonrelativistic bosons.  A transition from the upper to the
lower source state creates one boson, while the reverse transition annihilates
one.  The Hilbert space is
\begin{equation*}\label{eq:intro-Hilbert}
  \F_b\oplus\F_b\simeq\CO^2\otimes\F_b.
\end{equation*}
The one-boson dispersion relation and the free field Hamiltonian are
\begin{equation*}\label{eq:intro-dispersion}
  \omega(k)=|k|^2+m,
  \qquad m>0,
  \qquad H^\circ=\mathrm d\Gamma(\omega).
\end{equation*}
If $\mu$ denotes the energy-gap parameter of the source, $g\in\RE$ is
the coupling constant, and $v$ is the form factor, the formal Hamiltonian is
\begin{equation}\label{eq:intro-formal-H}
  H_{\mathrm{formal}}
  =
  \begin{pmatrix}
    H^\circ & g\, a^*(v)\\
    g\, a(v) & H^\circ+\mu
  \end{pmatrix},    \qquad \mu>0\,,\quad g\neq 0\,.
\end{equation}
Changing the matrix structure of the interaction leads to generalized
spin--boson models.  The interaction in \eqref{eq:intro-formal-H} is their
rotating-wave, or $2$-nilpotent, specialization; see
\cite{Lonigro2022,Lampart2025,FHM26} and the references therein.

We consider the constant form factor $v\equiv1$.  In position space this
describes an interaction concentrated at the source, whereas in momentum
space it couples all momenta with equal strength.  Since $1\notin L^2(\RE^3)$,
the creation term in \eqref{eq:intro-formal-H} is not an ordinary densely
defined Fock-space operator, and the displayed sum is only formal.

In this paper, the singular Hamiltonian is defined directly, rather than as the
outcome of a cutoff procedure.  We first specialize the abstract resolvent
construction for operators of the formal type $H+A^*+A$ developed in
\cite{MPAG1,MPAG2}.  This yields a self-adjoint, bounded-from-below operator
$H_\mu$ with an explicit domain and action; see Theorem~\ref{teo} and its
specialization in Section~\ref{sec:model-sector-notation}.  Its Kre\u\i n-type
resolvent formula is given in Theorem~\ref{Res-Lee}.  Related direct and cutoff constructions for singular rotating-wave
and spin--boson Hamiltonians have been obtained in
\cite{Lonigro2022,Lampart2025}. Only in
Appendix~\ref{app:reg} do we return to regularized Hamiltonians and prove that
the intrinsically defined $H_\mu$ is also obtained, after the appropriate
energy subtraction, as their norm-resolvent limit; see
Theorem~\ref{th:convergence}. This makes contact with the cited literature.  
In particular, the present abstract construction is arranged so that the reduced
operator needed for the spectral analysis in the subsequent sections appears explicitly from the
outset. A fundamental property of the model is that, although the boson number is not conserved, the sum of the boson number and
the indicator of the upper source state (namely, $0$ or $1$) is conserved.  We call it the total
excitation number.  Consequently,
\begin{equation*}\label{eq:intro-sector-decomp}
  \F_b\oplus\F_b
  =\bigoplus_{n=0}^{\infty}\IS^{[n]},
  \qquad
  H_\mu=\bigoplus_{n=0}^{\infty}H_\mu^{[n]},
\end{equation*}
where
\begin{equation*}\label{eq:intro-sector-space}
  \IS^{[0]}=\CO\oplus0,
  \qquad
  \IS^{[n]}
  =L_b^2(\RE^{3n})\oplus L_b^2(\RE^{3(n-1)}),
  \quad n\geq1.
\end{equation*}
A vector in the $n$-th sector therefore describes either the lower source
state with $n$ bosons or the upper source state with $n-1$ bosons.

For $z$ below the free $n$-boson threshold $nm$, elimination of the first
component gives an explicit reduced operator
\begin{equation*}\label{eq:intro-principal}
  M_z^{(n-1)}
  \quad\text{on}\quad
  L_b^2(\RE^{3(n-1)}).
\end{equation*}
It is the sum, with the sign convention used below, of a multiplication part
containing the finite renormalized self-energy and an integral part describing
boson exchange.  A factorization by boundedly invertible triangular operators allows to
transfer invertibility, kernels and Weyl sequences between
$H_\mu^{[n]}-\lambda$ and $M_\lambda^{(n-1)}$; see
Corollary~\ref{appcert:cor:M-reduction}.  In particular, for $\lambda<nm$,
\begin{equation*}\label{eq:intro-M-equation}
  \lambda\in\sigma_p(H_\mu^{[n]})
  \quad\Longleftrightarrow\quad
  0\in\sigma_p(M_\lambda^{(n-1)}),
\end{equation*}
with equality of multiplicities, and the corresponding equivalence holds for
the Weyl essential spectrum.  Reduced operators with an analogous role in
the Lee model bound state problem were introduced formally in \cite{Raj} and further
studied in \cite{ATT,JTU23}.

We emphasize that the dependence on the bare gap parameter $\mu$ and on a suitable subtraction point $\nu<m$
enters all the final formulae through the renormalized gap
\begin{equation*}\label{eq:intro-mug}
  \mu_g=\mu-2\pi^2g^2\sqrt{m-\nu}.
\end{equation*}
Thus $\nu$ specifies the renormalization convention: changing it while
describing the same renormalized Hamiltonian requires changing $\mu$ so that
$\mu_g$ remains fixed.  If $\mu_g<m$, we further set
\begin{equation}\label{eq:intro-lambda1}
  \lambda_1
  =m-\left(\sqrt{\pi^4g^4+m-\mu_g}-\pi^2g^2\right)^2<m.
\end{equation}

We first summarize the results on the essential spectrum.
Theorem~\ref{thm:free-essential-spectrum} states that, for every $n\geq1$,
\begin{equation*}\label{eq:intro-free-branch}
  [nm,+\infty)\subset\sigma_{\mathrm{ess}}(H_\mu^{[n]}).
\end{equation*}
More generally, Theorem~\ref{thm:lower-sector-inclusion} shows that, if $1\leq r<n$ and
$\zeta\in\sigma(H_\mu^{[r]})$ with $\zeta<rm$, then
\begin{equation}\label{eq:intro-general-branch}
  [\zeta+(n-r)m,+\infty)
  \subset\sigma_{\mathrm{ess}}(H_\mu^{[n]}).
\end{equation}
The proof constructs singular Weyl sequences by adjoining $n-r$ bosons that
escape to spatial infinity to an approximate spectral state of the lower
sector.  Notice that since the source is fixed, no recoil term occurs.  Now, defining
\begin{equation*}
  E_0:=0,
  \qquad E_r:=\inf\sigma(H_\mu^{[r]}),\quad r\geq1,
\end{equation*}
and
\begin{equation*}\label{eq:intro-taun}
  \tau_n:=\min_{0\leq r<n}\{E_r+(n-r)m\},
\end{equation*}
Corollary~\ref{cor:recursive-essential-threshold} gives
\begin{equation*}
  [\tau_n,+\infty)\subset\sigma_{\mathrm{ess}}(H_\mu^{[n]}).
\end{equation*}
This is the inclusion in the expected sector wise HVZ formula adapted to the present model.  The
reverse inclusion, and so the full HVZ theorem, is here proved for $n=1$ and $n=2$ in
Corollaries~\ref{cor:essential-spectrum-sector-one}
and~\ref{cor:essential-spectrum-sector-two}, respectively. For $n\geq3$,
it also follows in the absence regimes of Theorem~\ref{ub-n}, but remains
open in general outside those regimes.

The vacuum is identified as a simple
eigenvector in the sector $n=0$ (see Subsection \ref{subsec:vacuum-sector}).  The
one-excitation sector, which is nothing but a singular Friedrichs--Lee model, is
completely explicit.  Corollary~\ref{cor:essential-spectrum-sector-one} gives
\begin{equation*}
  \sigma_{\mathrm{ess}}(H_\mu^{[1]})=[m,+\infty).
\end{equation*}
Theorem~\ref{thm:n1} shows that, if $\mu_g\geq m$, there is no spectrum
below $m$, whereas, if $\mu_g<m$, the only spectral point below $m$ is the
simple eigenvalue $\lambda_1$ in \eqref{eq:intro-lambda1}.
Proposition~\ref{prop:resonances-sector-one} determines explicitly the poles
of the analytically continued one sector resolvent, including the threshold case.  These
one excitation results are included for completeness and as the first
application of the reduction; the spectral theory of singular
Friedrichs--Lee Hamiltonians is treated, for example, in \cite{FLL21}.

The two excitation sector is the first genuinely nontrivial bound state
problem.  Corollary~\ref{cor:essential-spectrum-sector-two} shows that its
essential threshold is
\begin{equation*}\label{eq:intro-two-threshold1}
  T_2=m+E_1
  =
  \begin{cases}
    2m,&\mu_g\geq m,\\
    m+\lambda_1,&\mu_g<m,
  \end{cases}
\end{equation*}
and
\begin{equation*}\label{eq:intro-two-essential}
  \sigma_{\mathrm{ess}}(H_\mu^{[2]})=[T_2,+\infty).
\end{equation*}

More generally, for $n\geq2$ and below the free or one excitation threshold
\begin{equation}\label{eq:intro-BS-threshold}
  \thr_n=
  \begin{cases}
    nm,&\mu_g\geq m,\\
    (n-1)m+\lambda_1,&\mu_g<m,
  \end{cases}
\end{equation}
Lemma~\ref{sp} shows that the equation $M_\lambda^{(n-1)}v=0$ is equivalent to
\begin{equation*}\label{eq:intro-BS}
  B_{n,\lambda}f=f,
\end{equation*}
where $B_{n,\lambda}$ is a bounded self-adjoint nonnegative operator that
plays the role of a Birman--Schwinger operator; see
Lemmas~\ref{BS} and~\ref{symm}.  Its continuity and monotonicity in
$\lambda$ are proved in Lemma~\ref{mon}, and explicit norm estimates are
given in Theorem~\ref{ub-n}.  For $n=2$, the threshold in
\eqref{eq:intro-BS-threshold} is precisely the bottom
$T_2=m+E_1$ of the essential spectrum by
Corollary~\ref{cor:essential-spectrum-sector-two}; moreover,
$B_{2,\lambda}$ is compact and, if $g\neq0$, its largest eigenvalue is simple
and isolated (see Lemma~\ref{n-sp}).  For $n\geq3$, however, $\thr_n$ need not
be the bottom of the full essential spectrum, because a spectral point of a
proper intermediate sector may generate a lower half-line contribution to the essential spectrum, through
\eqref{eq:intro-general-branch}.

The estimates give the following explicit conclusions for every $n\geq2$.
Theorem~\ref{ub-n} proves that, at arbitrary coupling,
\begin{equation*}\label{eq:intro-n2-no-bound-mug-large}
  \mu_g\geq m
  \quad\Longrightarrow\quad
  \sigma(H_\mu^{[n]})\cap(-\infty,nm)=\varnothing.
\end{equation*}
If $\mu_g<m$, Theorem~\ref{ub-n} gives
\begin{equation}\label{eq:intro-n2-absence}
  \frac{\pi^2g^2}{\sqrt{m-\lambda_1}}\leq1
  \quad\Longrightarrow\quad
  \sigma(H_\mu^{[n]})\cap(-\infty,\thr_n)=\varnothing,
\end{equation}
whereas, for $n=2$, Theorem~\ref{n=2eigenvalues} shows that
\begin{equation}\label{eq:intro-n2-existence}
  \frac{\pi^2g^2}{\sqrt{m-\lambda_1}}
  >\frac{\pi}{4-\pi}
\end{equation}
implies the existence of a simple isolated eigenvalue below
$m+\lambda_1$.  The gap between the sufficient conditions
\eqref{eq:intro-n2-absence} and \eqref{eq:intro-n2-existence}, as well as the
possible existence of further two-sector eigenvalues, are open questions.

When $n\geq3$, $B_{n,\lambda}$ is generally noncompact.  Therefore
$\|B_{n,\lambda}\|=1$ does not by itself imply that $1$ is an eigenvalue,
and the norm crossing the value $1$ alone does not establish binding (see
Remark~\ref{rem:caveat}).
The norm bounds nevertheless exclude spectrum in every sector, despite
the absence of compactness. In either absence regime stated above,
Theorem~\ref{ub-n} gives, for every $n\geq2$,
\begin{equation*}\label{eq:intro-small-coupling-full}
  \sigma(H_\mu^{[n]})
  =\sigma_{\mathrm{ess}}(H_\mu^{[n]})
  =[\thr_n,+\infty).
\end{equation*}

We collect the results on the sector Hamiltonians $H_\mu^{[n]}$ in the following theorem.
\begin{theorem}[Sector spectra]\label{thm:sector-spectra-summary}
\leavevmode
\begin{itemize}
\item[$\bullet\; n=0$]
\begin{equation*}  \sigma(H_\mu^{[0]})=\{0\},\end{equation*}
where $0$ is a simple eigenvalue. 
\item[$\bullet\; n=1$]
\begin{equation*}
  \sigma_{\mathrm{ess}}(H_\mu^{[1]})=[m,+\infty);
\qquad 
  \sigma_d(H_\mu^{[1]})=\left\{\begin{aligned}&\varnothing\quad &  \text{if $\mu_g\geq m$}\\
    &\{\lambda_1\}&\text{if $\mu_g< m$}
    \end{aligned}\right.
\end{equation*}
where $\lambda_1$ is a simple eigenvalue. 
\item[$\bullet\; n=2$]
\begin{equation*} \sigma_{\mathrm{ess}}(H_\mu^{[2]})=[T_2,+\infty);\end{equation*}
\begin{equation*}
  \sigma_d(H_\mu^{[2]})= \varnothing\quad  \text{if $\mu_g\geq m$ or $\mu_g< m$ and $ \frac{\pi^2g^2}{\sqrt{m-\lambda_1}}\leq1$};
\end{equation*}
if 
\begin{equation}\label{condition_for_lambda_2}
\text{$\mu_g< m$ and $\frac{\pi^2g^2}{\sqrt{m-\lambda_1}}
  >\frac{\pi}{4-\pi}$}
  \end{equation}
  then $\sigma_d(H_\mu^{[2]})$ contains at least one simple  eigenvalue $\lambda_2 < T_2 = m+\lambda_1$. 
\item[$\bullet\; n\geq 3$]
\begin{equation}\label{inclusion_ess_ngeq3}
\sigma_{\mathrm{ess}}(H_\mu^{[n]}) \supset [T_n,+\infty);
\end{equation}
 if
\begin{equation}\label{condition_ess_ngeq3}
\text{$
\mu_{g}\geq m$ or 
$\mu_{g}<m$ and $\frac{\pi^{2}g^{2}}{\sqrt{m-\lambda_{1}}}\leq1$}\end{equation}
then
\begin{equation}\label{spectrum_ngeq3}
\sigma(H_\mu^{[n]})=\sigma_{\mathrm{ess}}(H_\mu^{[n]})=[T_n,+\infty);
\end{equation}
furthermore, if the condition in Eq. \eqref{condition_for_lambda_2} holds true then
\begin{equation}\label{spectrum_final_ngeq3}
 \sigma_{\mathrm{ess}}(H_\mu^{[n]}) \supset [\lambda_2 + (n-2)m,+\infty)  \supset [T_n,+\infty) .
\end{equation}
\end{itemize}
\end{theorem}
\begin{proof}
For $n=0$, see Subsection \ref{subsec:vacuum-sector}.

For $n=1$, see Corollary \ref{cor:essential-spectrum-sector-one} for the essential spectrum and Theorem \ref{thm:n1} for the discrete spectrum.

For $n=2$, see Corollary \ref{cor:essential-spectrum-sector-two} for the essential spectrum and Theorem \ref{ub-n} for the absence of spectrum below $T_2$. Theorem \ref{n=2eigenvalues} shows the existence of $\lambda_2$ under the condition in Eq. \eqref{condition_for_lambda_2}. 

For $n\geq3$, the inclusion in Eq. \eqref{inclusion_ess_ngeq3} is given  in Theorem \ref{thm:free-essential-spectrum}  for $\mu_g\geq m$  and in Corollary \ref{cor:essential-spectrum-sector-one}, Eq. \eqref{eq:one-sector-branch}, for $\mu_g< m$. Theorem \ref{ub-n} shows that, under the condition in Eq. \eqref{condition_ess_ngeq3}, $\sigma(H^{[n]}_{\mu})\subset [\thr_n,+\infty)$, hence the opposite inclusion $\sigma_{\mathrm{ess}}(H_\mu^{[n]}) \subset  [T_n,+\infty)$ also holds, giving in turn the identities in Eq. \eqref{spectrum_ngeq3}. Finally, the statement in Eq. \eqref{spectrum_final_ngeq3} follows from the existence of the eigenvalue $\lambda_2$ and from Theorem \ref{thm:lower-sector-inclusion}, with $\zeta = \lambda_2$ and $r=2$.
\end{proof}

We now briefly place these results in the literature.  Low-excitation sectors of
the Lee model were studied through scattering amplitudes, pole equations and
three-body integral equations in the older physics-oriented literature
\cite{Muta65,Pag65,Pag66,Fuda82}; all-sector
algebraic and scattering reductions appear in \cite{Fivel70,LZ70}, while
\cite{Weder74} gives an operator-theoretic analysis with a
dilation-analytic ultraviolet cutoff.  The reduced resolvent viewpoint is proposed in \cite{Raj} and partially
developed in \cite{ATT,JTU23}.  We also recall that there is extensive work on
essential spectra and bound states for other particle--field Hamiltonians;
representative references include
\cite{BFS99,DG99,DJ01,Moller05,MR14,Dam20,DH22,DM20a,DM20b}, with broader
perspectives in \cite{Falconi25,HHS21,Hin22thesis,FHM26}.  In particular,
Dam--M\o ller obtain excited-state results for massive regularized
spin--boson type models, and Ibrogimov proves finiteness of the discrete
spectrum for the regular spin--boson Hamiltonian truncated at two photons
\cite{Ibrogimov2018}.  The latter results and the present spectral results
are more thoroughly compared in Remark~\ref{rem:Ibrogimov-comparison}.  

We end the Introduction by collecting several questions that remain open.
The first is the reverse inclusion in the expected
formula
\begin{equation*}\label{eq:intro-full-HVZ}
  \sigma_{\mathrm{ess}}(H_\mu^{[n]})=[\tau_n,+\infty)
\end{equation*}
for $n\geq3$ outside the absence regimes of Theorem~\ref{ub-n}.  A second is the complete count of the two-sector discrete
spectrum and the determination of the exact critical coupling between
\eqref{eq:intro-n2-absence} and \eqref{eq:intro-n2-existence}.  In higher
sectors, a reduction in the spirit of Faddeev or Amado--Lovelace equations
may isolate the noncompact part and leave a compact remainder, as in
few-body theory; see \cite{Fuda82}.  Such a framework may also be useful for the analysis of the
scattering theory of the model, not touched upon here, and resonances in higher sectors.

The ground state of the full Hamiltonian raises a further question,
requiring a comparison of the spectral bottoms of all excitation sectors.
Remark~\ref{rem:ground-state} identifies it in some explicit parameter regimes
and reduces the remaining problem to finitely many sectors. Its existence
and the minimizing excitation sector remain to be determined for general
parameters.

The paper is organized as follows.  Section~\ref{s:prelude} presents the
abstract resolvent construction and spectral reduction.  In
Section~\ref{sec:model-sector-notation} it is specialized to the singular Lee
Hamiltonian and to fixed excitation sectors.  Section~\ref{sec:essential-spectrum}
constructs the essential-spectrum inclusions, while
Section~\ref{sec:point-discrete-spectrum} treats the point and
discrete spectra for the low and higher sectors.  Appendix~\ref{app:reg} proves norm-resolvent convergence
of the energy-renormalized cutoff Hamiltonians to the intrinsically defined
operator.

\section{Notation and definitions.}
\begin{itemize}
\item $\NA:=\{1,2,\ldots\}$ and
$\NA_0:=\{0,1,2,\ldots\}$;
\item $L_b^2(\RE^{3n})$ denotes the closed subspace of
$L^2(\RE^{3n})$ consisting of boson-symmetric functions, and
$L_b^2(\RE^0):=\CO$;
\item inner products $\langle\cdot,\cdot\rangle_X$ are
conjugate-linear in the first argument and linear in the second, and
$\|\cdot\|_X$ denotes the corresponding norm; subscripts are omitted when
the underlying space is clear, and $u_j\rightharpoonup u$ denotes weak
convergence;
\item $\dom(L)$, $\ker(L)$ and $\ran(L)$ denote,
respectively, the domain, kernel and range of the linear operator $L$;
\item $\varrho(L)$ and $\sigma(L)$ denote the resolvent set and the spectrum of $L$;
\item $\sigma_{p}(L)$, $\sigma_{c}(L)$,
$\sigma_{d}(L)$ and $\sigma_{\mathrm{ess}}(L)$ denote the point,
continuous, discrete and Weyl essential spectra.  For a self-adjoint
operator, $\sigma_d(L)$ consists of its isolated eigenvalues of finite
multiplicity and
$\sigma_{\mathrm{ess}}(L)=\sigma(L)\setminus\sigma_d(L)$;
\item a Weyl sequence for a self-adjoint operator $L$ at
$\lambda\in\RE$ is a sequence $\{u_j\}\subset\dom(L)$ such that
$\|u_j\|=1$ and $(L-\lambda)u_j\to0$; it is called a singular Weyl sequence
if, in addition, $u_j\rightharpoonup0$;
\item $L|_V$ denotes the restriction of $L$ to the subspace
$V\subset\dom(L)$;
\item $\B(X,Y)$ denotes the space of bounded linear operators
from the Banach space $X$ to the Banach space $Y$, and
$\B(X):=\B(X,X)$;
\item $\|\cdot\|_{X,Y}$ denotes the norm in $\B(X,Y)$;
\item $I\Subset J$ means that the closure of $I$ is compact
and contained in $J$;
\item A linear operator $S$ is said to be $H$-small whenever $\dom(S)\supseteq\dom(H)$ and there exist $ c\in[0,1)$ and $b\in\RE$ such that $\|S\psi\|\le c\,\|H\psi\|+b\,\|\psi\|$ for any $\psi\in\dom(H)$;
\item $S_{\Lambda}$ is said to be uniformly $H_{\Lambda}$-small whenever each $S_{\Lambda}$ is $H_{\Lambda}$-small with $\Lambda$-independent constants $c$ and $b$;
\item $C_a$, $C_{a,b}$, and so on, denote generic positive
constants depending only on the indicated parameters; their values may
change from line to line.  
\end{itemize}
\section{Prelude: the abstract framework\label{s:prelude}}
\noindent 
In the first part of this Section we recall a way to build, through resolvents, self-adjoint operators belonging to a class which includes various renormalizable models in Quantum Field Theory. For more details, references and proofs we refer to \cite{MPAG1} and \cite{MPAG2}; in particular, Theorem \ref{real} corresponds to Theorem 3.1 in \cite{MPAG2}. Subsequently, in the second part, we apply the results to a renormalized abstract Lee model and relate its spectral profile 
to whether or not zero belongs to the different spectral components of an operator pencil which plays the role of an operator-valued Weyl function (see Theorem \ref{teo-sp}). 
\subsection{Building a resolvent}\label{building}
Let \begin{equation*}H :\dom(H)\subseteq\F\to\F\end{equation*} be a bounded-from-below self-adjoint operator in the Hilbert space $\F$ equipped with the scalar product $\langle\cdot,\cdot\rangle$ and corresponding norm $\|\cdot\|$. 
\begin{definition}[Scale of Hilbert spaces] \label{d:SHS}
We denote by $\H ^{s}$, $s\in [-1,1]$,
 with $\H^{0}\equiv\F$ and $\H^{1} \equiv\dom(H)$, the decreasing scale of Hilbert spaces given by the completion of $\dom(H)$ endowed with the scalar product
$$
\langle \psi_{1},\psi_{2}\rangle_{s}:=\langle(H^{2}+1)^{s/2}\psi_{1},(H^{2}+1)^{s/2}\psi_{2}\rangle\,.
$$
\end{definition}
One has
\begin{equation*}\label{incl}
\H^{s}\hookrightarrow\F\hookrightarrow\H^{-s}\,,\qquad 0<s\le 1\,, 
\end{equation*}
with dense inclusions. By such dense inclusions, the scalar product on $\F$ induces the dual pairings (conjugate-linear with respect to the first variable)
$\langle\cdot,\cdot\rangle_{\pm s,\mp s}$ between the dual couples
$(\H^{\mp s},\H^{\pm s})$. In the following, adjoints are taken with respect to such dualities and we use the abbreviated notation
$$
\H\equiv\H^{1} \,,\qquad \H^{*}\equiv\H^{-1} \,.
$$
Let
\begin{equation*}\label{ann}  
 A:{\H } \to \F\,,
\end{equation*} 
be a bounded linear map; for any $z\in  \varrho(H)$ we define the bounded operator 
\begin{equation*}\label{Gz}
G_{z} : \F\to\F \,,\qquad G_{z}:=(  A R_{{\bar z} })^{*}\,,
\end{equation*}
where
$$
R_{z}:\F\to{\H }\,,\qquad R_{z}:=(-H +z )^{-1}\,.
$$
Notice that $R_{z}$  extends to a bounded map $R_{z}\!:\!\H^{*} \to\F$ with bounded inverse
$$(-\cH+z): \F\to\H^{*} \,,
$$ where $\cH\in \B(\F,\H^{*})$ denotes the closure of the densely-defined, bounded operator $H:\H \subseteq\F\to\H^{*}$ (we abuse notation and use the same symbol when viewing $H$ as an operator in larger ambient spaces). By introducing the adjoint 
\begin{equation*}\label{A*}
A^{*}:\F\to \H^{*} \,,
\end{equation*}
one has
\begin{equation*}\label{RA*}
G_{z}=R_{z}A^{*}
\end{equation*}
and
\be\label{HG}
(-\cH+z)G_{z}=A^{*}\,.
\ee
By the first resolvent identity, one gets
\begin{equation*}\label{RG}
(z-w)R_{w}G_{z}=G_{w}-G_{z}=(z-w)R_{z}G_{w}\,,
\end{equation*}
and so  
\begin{equation*}\label{RG1}
A(G_{w}-G_{z})=(z-w)G^{*}_{\bar w}G_{z}=(z-w)G^{*}_{\bar z}G_{w}\in\B(\F)\,.
\end{equation*}
Given $\lambdastar<\inf\sigma (H)$, in the following we set $R:=R_{\lambdastar}$ and $ G:=G_{\lambdastar}$.
\begin{theorem}\label{real} Let  $H:\H \subseteq\F\to\F$ be self-adjoint and  bounded-from-below and let the symmetric operator $S$ be $H$-small; let $A$ belong to $\B(\H^ {s},\F)$ for some $s\in(0,1)$ and suppose that 
\be\label{H3}
\ran(G)\cap\H =\{0\}\,.
\ee
Then, 
$$
H_{S}:\S\subseteq\F\to\F\,,\qquad H_{S}:=\cH+A^{*}+A_{S}\,,
$$
$$
\S :=\{\psi\in\F:\psi^{\circ}:=\psi-G\psi\in\H \}\,,
$$
$$
A_{S}:\S \subseteq\F\to\F\,,\qquad A_{S}
:=\big(A-(1-G^{*})S\big)(1-G)\,,
$$
is a bounded-from-below self-adjoint operator with resolvent given by\be\label{krf}
(-H_{S}+z)^{-1}=R_{z}+\begin{bmatrix}G_{z}&R_{z}\end{bmatrix}
\begin{bmatrix} T_{S}+A(G-G_{z})&1-G^{*}_{\bar z}\\
1-G_{z}&-R_{z}
\end{bmatrix}^{-1}\begin{bmatrix}G_{\bar z}^{*}\\R_{z}\end{bmatrix}, \quad z\in \varrho(H)\cap\varrho(H_{S})\,,
\ee
where $T_{S}:=(1-G)^{*}{S}(1-G)$ and the block operator matrix inverse exists in $\B(\F\oplus\H,\F\oplus\H^{*})$.
\end{theorem}
\begin{remark} Here, a few words of explanation of the results given in the previous Theorem. \par The right hand side of \eqref{krf} defines a pseudo-resolvent $\widehat R_{z}$ which, by hypothesis \eqref{H3}, is injective; hence, it is the resolvent of a closed operator. Such an operator is self-adjoint by $\widehat R_{z}^{*}=\widehat R_{\bar z}$. Then, one shows that $-\widehat R_{z}^{-1}+z$ coincides with $H_{S}$ (for the details, see the proof of \cite[Theorem 3.1]{MPAG2}).\par
By \eqref{HG}, one gets $\cH+A^{*}=(\cH-\lambdastar)(1-G)+\lambdastar$, which implies that the operator $(\cH+A^{*})|\S$ is $\F$-valued. Hence, $H_{S}$ is $\F$-valued as well.\par
$H_{S}$ is bounded-from-below since $S$ is $H$-small, $\|S\psi\|\le c\|H\psi\|+b\|\psi\|$. Indeed, by \eqref{krf} evaluated at $z=\lambda$ and
by \cite[relations (3.10) and (3.11)]{MPAG2}, $\lambda\in \RE\cap\varrho(H_{S})\cap\varrho(H)$ whenever 
$$
{\mathbb T}_{\lambda}:=\begin{bmatrix} T_{S_{\lambda}}&1-G_{\lambda}^{*}\\
1-G_{\lambda}&-R_{\lambda}
\end{bmatrix}= \begin{bmatrix} 1-G_{\lambda}^{*}&0\\
0&1
\end{bmatrix}
{\mathbb S}_{\lambda}
\begin{bmatrix} 1-G_{\lambda}&0\\
0&1
\end{bmatrix},\qquad {\mathbb S}_{\lambda}:=\begin{bmatrix} S_{\lambda}&1\\
1&-R_{\lambda}
\end{bmatrix},$$
has a bounded inverse; here, $S_{\lambda}$ is an explicit symmetric operator such that $S- S_{\lambda}\in\B(\F)$. By \cite[Lemma 2.4]{MPAG2}, $1-G_{\lambda}$ and $1-G_{\lambda}^{*}$ have bounded inverses whenever $\lambda$ is sufficiently below $\sigma(H)$. 
Since the second Schur complement of  ${\mathbb S}_{\lambda}$ is $-(H-S_{\lambda})+\lambda$,
${\mathbb S}_{\lambda}$ has a bounded inverse whenever $(-(H-S_{\lambda})+\lambda)^{-1}\in\B(\F)$. Since  $S_{\lambda}$ is $H$-small, $\|S_{\lambda}\psi\|\le c\|H\psi\|+b_{\lambda}\|\psi\|$, with $b_{\lambda}$ growing sub-linearly as $|\lambda|\nearrow+\infty$ (see the proof of \cite[Theorem 3.1]{MPAG2}), 
${\mathbb T}_{\lambda}$ has a bounded inverse whenever $\lambda$ is sufficiently below $\sigma(H)$ (see \cite[Remark 2.12]{MPAG2}).
\end{remark}
\begin{remark} Since $R\in\B(\H^{s-1},\H^ {s})$, one has $AR\in \B(\H^ {s-1},\F)$ whenever $A\in\B(\H^ {s},\F)$. Then, by duality, $(AR)^{*}=G\in\B(\F,\H^ {1-s})$ and so
$$
\S\subseteq\H^ {1-s}\,.
$$   
Moreover, by \eqref{H3},
$$
\S\cap\H =\ker(G)=\ran(AR)^{\perp}\,.
$$
Hence, since $R:\F\to\H $ is a continuous bijection,
$$
\S\cap\H =\{0\}\quad\Leftrightarrow\quad \text{$\ran(A)$ is dense.}
$$ 
\end{remark}
\subsection{The abstract Lee model}\label{s:abstractlee}
Here, we apply the results of the previous subsection to the case 
\begin{equation*} \label{HAALambda}
\F=\F_{1}^{\circ}\oplus \F_{2}^{\circ}\,,\quad H=H_{1}^{\circ}\oplus H_{2}^{\circ}
\equiv\begin{bmatrix}H_{1}^{\circ}&0\\
0&H_{2}^{\circ}\end{bmatrix}\,,\quad A:=\begin{bmatrix}0&0\\
a&0\end{bmatrix}.
\end{equation*}
Here,  both $H_{1}^{\circ}$ and $ H_{2}^{\circ}$ are bounded from below self-adjoint operators, $\lambda_{k,\star}<\inf\sigma(H_{k}^{\circ})$, and $a\in\B(\H^{\circ}_{1},\F^{\circ}_{2})$. The superscript ${\,}^{\circ}$ indicates objects defined as in the previous section but regarding $H_{k}^{\circ}$ and $a$; in particular, we define 
\begin{equation*}
G_{z}^{\circ}:=(a(-H_{1}^{\circ}+\bar z)^{-1})^{*}=R_{1,z}^{\circ}a^*\in\B(\F_{2}^{\circ},\F^{\circ}_{1})\,,\qquad 
z\in\varrho(H^{\circ}_{1})\,,
\end{equation*}  
and use the shorthand notation $G^{\circ}\equiv G_{\lambda_{1,\star}}^{\circ}$.

By the first resolvent identity, there follows
\begin{equation}\label{utile}
a(G_{w}^{\circ}-G_{z}^{\circ}) = (z-w)aR_{1,w}^{\circ}G_{z}^{\circ}=
 (z-w)aR_{1,z}^{\circ}G_{w}^{\circ}\in\B(\F_{2}^{\circ})\,,\qquad 
z,w\in\varrho(H^{\circ}_{1})\,.
\end{equation}

\par 
By the above definitions, setting $R^{\circ}_{k,z}:=(-H^{\circ}_{k}+z)^{-1}$ with $z\in\varrho(H^{\circ}_{k})$, one gets, for any $z\in\varrho(H^{\circ}_{1})$,
\begin{equation*}
AR_{z}=\begin{bmatrix}0&0\\
aR_{1,z}^{\circ}&0\end{bmatrix}\,,\qquad G_{z}=\begin{bmatrix}0&G_{z}^{\circ}\\
0&0\end{bmatrix}\qquad
1-G_{z}=\begin{bmatrix}1&-G_{z}^{\circ}\\
0&1\end{bmatrix}.
\end{equation*}
Notice that, even if $AR_{z}$ should a priori be defined only for any $z\in \varrho(H_{1}^{\circ}\oplus H_{2}^{\circ})=\varrho(H_{1}^{\circ})\cap\varrho( H_{2}^{\circ})$ it is indeed defined for any $z\in \varrho(H_{1}^{\circ})$. Consequently, in contrast with the definition in Subsection \ref{building}, here we define $G:=G_{\lambda_{1,\star}}$.
One has 
\begin{equation*}
\ran(G)\cap(\H_{1}^{\circ}\oplus \H_{2}^{\circ})=\{0\}\quad\Leftrightarrow\quad\ran(G^{\circ})\cap\H_{1}^{\circ}=\{0\}
\end{equation*}
and, by \cite[Lemma 2.5]{MPAG1}, 
\be\label{dense}
\text{$\ker(a)$ is dense in $\F_{1}^{\circ}$}\quad\Rightarrow\quad\ran(G^{\circ})\cap\H_{1}^{\circ} =\{0\}\,.
\ee
Notice that,  for any $z\in\varrho(H^{\circ}_{1})$,
\be\label{inv}
(1-G_{z})^{-1}=\begin{bmatrix}1&-G_{z}^{\circ}\\
0&1\end{bmatrix}^{-1}=\begin{bmatrix}1&G_{z}^{\circ}\\
0&1\end{bmatrix},\qquad 
(1-G_{z}^{*})^{-1}=\begin{bmatrix}1&0\\-{G_{z}^{\circ}}^{*}&1\end{bmatrix}^{-1}=
\begin{bmatrix}1&0\\{G_{z}^{\circ}}^{*}&1\end{bmatrix}.
\ee
Then, Theorem \ref{real} gives
\begin{theorem}\label{teo} Let $S_{2}^{\circ}$ be a symmetric, $H_{2}^{\circ}$-small operator
and let $a\in\B(\H_{1}^{\circ,s},\F_{2}^{\circ})$ for some $0<s<1$ such that  $\ker(a)$ is dense in $\F_{1}^{\circ}$. Then,   
\begin{equation}\label{HS1}
H_{S_{2}^{\circ}}:\S^{\circ}\subseteq\F_{1}^{\circ}\oplus\F_{2}^{\circ}\to\F_{1}^{\circ}\oplus\F_{2}^{\circ}\,,
\end{equation}
\begin{align}\S^{\circ}:=&\{\psi_{1}\oplus\psi_{2}\in\F_{1}^{\circ}\oplus\H_{2}^{\circ}:\psi_{1}-G^{\circ}\psi_{2}\in\H_{1}^{\circ}\} \label{HS2} \\
=&\{\psi_{1}\oplus\psi_{2}\in\F_{1}^{\circ}\oplus\H_{2}^{\circ}:\overline {H_{1}^{\circ}}\psi_{1}+a^{*}\psi_{2}\in\F_{1}^{\circ}\}
\,,\nonumber
\end{align}
\begin{equation}\label{HS4}
H_{S_{2}^{\circ}}(\psi_{1}\oplus\psi_{2}):=(\,\overline{H_{1}^{\circ}}\psi_{1}+a^{*}\psi_{2})\oplus((H_{2}^{\circ}-S_{2}^{\circ})\psi_{2}+ a(\psi_{1}-G^{\circ}\psi_{2}))
\end{equation}
is a bounded-from-below self-adjoint operator.
\end{theorem}
\begin{proof} The only apparent discrepancy with Theorem \ref{real} is our representation of $\S^{\circ}$ 
in terms of $G^{\circ}=G^{\circ}_{\lambda_{1,\star}}$ while, according to the setting in Subsection \ref{building}, we should use $G_{\star}:=G_{\lambdastar}$, where $\lambdastar=\min\{\lambda_{1,\star},\lambda_{2,\star}\}$. However, by the resolvent identity, $G^{\circ}-G^{\circ}_{\star}=(\lambdastar-\lambda_{1,\star})R^{\circ}_{1,\lambdastar}G^{\circ} \in \B(\F^{\circ}_{2},\H^{\circ}_{1})$ and so this change does not modify $\S^{\circ}$. Furthermore, by
$$
(H_{2}^{\circ}-S_{2}^{\circ})\psi_{2}+ a(\psi_{1}-G^{\circ}\psi_{2})=
(H_{2}^{\circ}-S_{2}^{\circ}-a(G^{\circ}-G^{\circ}_{\star}))\psi_{2}+ a(\psi_{1}-G_{\star}^{\circ}\psi_{2})\,,
$$ 
then it suffices to use Theorem \ref{real} with 
\begin{equation*}\label{tilde-S}
S:=\begin{bmatrix}0&0\\
0&\widetilde S_{2}^{\circ}\end{bmatrix},\qquad \widetilde S_{2}^{\circ}:=S_{2}^{\circ}+a(G^{\circ}-G^{\circ}_{\star})=S_{2}^{\circ}+(\lambdastar-\lambda_{1,\star})G_{\star}^{\circ*}G^{\circ}\,. 
\end{equation*}
\end{proof}
\begin{remark} By $\ran(G^{\circ})\cap\H_{1}^{\circ}=\{0\}$, which is a consequence of the denseness of $\ker (a)$ (see \eqref{dense}), one has $\S^{\circ}\cap (\H_{1}^{\circ}\oplus\H_{2}^{\circ})=\S^{\circ}\cap\dom\bigl(H_{1}^{\circ}\oplus H_{2}^{\circ}\bigr)=\H_{1}^{\circ}\oplus(\H_{2}^{\circ}\cap\ker(G^\circ))$.  In particular, the last factor is $\{0\}$ whenever $G^\circ$ is injective.
\end{remark}
\begin{remark}\label{rem} By 
$$
\S^{\circ}=\{\psi_{1}\oplus\psi_{2}\in\F_{1}^{\circ}\oplus\H_{2}^{\circ}:\psi_{1}=\varphi_{1}+G^{\circ}\psi_{2}\,,\ \varphi_{1}\in\H_{1}^{\circ}\}\,,
$$
and by $$\overline{H_{1}^{\circ}}G^{\circ}=
-(-\overline{H_{1}^{\circ}}+\lambda_{1,\star})G^{\circ}+\lambda_{1,\star} G^{\circ}=-a^{*}+
\lambda_{1,\star} G^{\circ}\,,
$$
one obtains the alternative representation
\begin{align*}
H_{S_{2}^{\circ}}(\psi_{1}\oplus\psi_{2})=&(\,H_{1}^{\circ}\varphi_{1}+\overline{H_{1}^{\circ}}G^{\circ}\psi_{2}+a^{*}\psi_{2})\oplus((H_{2}^{\circ}-S_{2}^{\circ})\psi_{2}+ a\varphi_{1})\\
=&(\,H_{1}^{\circ}\varphi_{1}+\lambda_{1,\star} G^{\circ}\psi_{2})\oplus((H_{2}^{\circ}-S_{2}^{\circ})\psi_{2}+ a\,\varphi_{1})\,.
\end{align*}
\end{remark}
\subsection{The resolvent of the Lee model} By the Schur complement, elaborating on the general Kre\u\i n-type resolvent formula \eqref{krf}, one gets the resolvent for the Lee Hamiltonian $H_{S_{2}^{\circ}}$ defined in \eqref{HS1}-\eqref{HS4} (compare, in the case $\F_{1}^{\circ}=\F_{2}^{\circ}$ and $H_{1}^{\circ}=H_{2}^{\circ}$, with \cite[equation (85)]{Raj}). Here, we give a more direct derivation.
\begin{theorem}\label{Res-Lee} For any $z\in\varrho(H_{1}^{\circ})\cap\varrho(H_{S_{2}^{\circ}})$, there holds
\be\label{Raj}
(-H_{S_{2}^{\circ}}+z)^{-1}=\begin{bmatrix}R_{1,z}^{\circ}+G^{\circ}_{z}M_{z}^{-1}{G^{\circ}_{\bar z}}^{*}&G^{\circ}_{z}M_{z}^{-1}\\M_{z}^{-1}{G^{\circ}_{\bar z}}^{*}&M_{z}^{-1}\end{bmatrix}\,,
\ee
where
\be\label{Mz}
M_{z}\equiv M^\circ_{z} :=-H_{2}^{\circ}+S_{2}^{\circ}+a(G^{\circ}-G^{\circ}_{z})+z=-H_{2}^{\circ}+S_{2}^{\circ}+(z-\lambda_{1,\star})\,{G^{\circ}}^{*}G^{\circ}_{z}+z\,,
\ee

\begin{equation*}
\dom(M_{z}) := \H_{2}^{\circ}.
\end{equation*}

\end{theorem}
\begin{proof}
Fix $z\in\varrho(H_{1}^{\circ})$ and set

\begin{equation*}
 \mathcal T_z:=
 \begin{pmatrix}1&G_z^\circ\\0&1\end{pmatrix}.
\end{equation*}
$\mathcal T_z$ is bounded and boundedly invertible on
$\F_1^\circ\oplus\F_2^\circ$. Furthermore, since resolvent identity implies that $G_z^\circ-G^\circ$ maps
$\F_2^\circ$ into $\H_1^\circ$, $\mathcal T_z$ maps
$\H_1^\circ\oplus\H_2^\circ$ bijectively onto $\S^\circ$.

 For
$\varphi\in\H_1^\circ$ and $\xiq
\in\H_2^\circ$, the domain and action
in \eqref{HS2}--\eqref{HS4} give
\begin{equation*}
 (-H_{S_2^\circ}+z)
 \bigl((\varphi+G_z^\circ \xiq
 )\oplus \xiq
 \bigr)
 =
 (-H_1^\circ+z)\varphi
 \oplus\bigl(-a\varphi+M_z\xiq
 \bigr).
\end{equation*}
Since $
G_{\bar z}^{\circ*}(-H_1^\circ+z)\varphi =
aR_{1,z}^\circ(-H_1^\circ+z)\varphi=a\varphi$, this yields the
factorization
\begin{equation}\label{eq:abstract-Lee-factorization}
\mathcal T_{\bar z}^*
 (-H_{S_2^\circ}+z)\mathcal T_z
 =
 \begin{pmatrix}
  -H_1^\circ+z&0\\
  0&M_z
 \end{pmatrix}
\end{equation}
on $\H_1^\circ\oplus\H_2^\circ$.  Because the first diagonal entry
is boundedly invertible, \eqref{eq:abstract-Lee-factorization} shows that
$z\in\varrho(H_{S_2^\circ})$ if and only if $0\in\varrho(M_z)$.
Under the hypotheses of the theorem, taking inverses gives
\begin{equation}\label{eq:abstract-spectral-factorization_NEW}
 (-H_{S_2^\circ}+z)^{-1}
 =
 \mathcal T_z
 \begin{pmatrix}R_{1,z}^\circ&0\\0&M_z^{-1}\end{pmatrix}
 \mathcal T_{\bar z}^*
\end{equation}
Finally, $aR_{1,z}^\circ=(G_{\bar z}^\circ)^*$, and multiplication
of the three block matrices gives \eqref{Raj}.
\end{proof}
\begin{remark}\label{r:3.9} Since $S_{2}^{\circ}$ is $H_{2}^{\circ}$-small and by \eqref{utile}, $M_z$ is well defined on $\H_{2}^{\circ}$. Furthermore, by  
\eqref{Mz}, there follows, for any $z,w\in \varrho(H_{1}^{\circ})$, 
\begin{equation*}
M_{z}-M_{w}=z-w+a(G^{\circ}_{w}-G^{\circ}_{z})=(z-w)\big(1+G^{\circ*}_{\bar w}G^{\circ}_{z}\big)
\,.
\end{equation*}
Since $\|G_{\bar\zeta}^{\circ}\|=\|G_{\zeta}^{\circ}\|$ by the spectral calculus and $\varrho(H_{1}^{\circ})\ni\zeta\mapsto G^{\circ}_{\zeta}$ is analytic (hence continuous) $\B(\F_{2}^{\circ},\F_{1}^{\circ})$-valued, for any compact $K\subset \varrho(H_{1}^{\circ})$ and for any $z,w\in K$ one has 
\begin{equation*}
M_{z}-M_{w}\in\B(\F_{2}^{\circ})\qquad\text{and}\qquad\|M_{z}-M_{w}\|\le |z-w|\,\left(1+\sup_{\zeta\in K}\|G^{\circ}_{\zeta}\|^{2}\right)\,.
\end{equation*}
\end{remark}
\subsection{Spectral reduction\label{ss:ee}} The next theorem shows how the spectrum of the abstract Lee model  relates to the  spectrum of $M_{z}$.
\begin{theorem}\label{teo-sp}
Let $H_{S_{2}^{\circ}}$ be the bounded-from-below self-adjoint operator in
Theorem \ref{teo}, and let $M_z$ be defined by \eqref{Mz}.  For every
$\lambda\in\RE\cap\varrho(H_1^\circ)$,
\begin{equation}\label{iff-rho}
 \lambda\in\varrho(H_{S_2^\circ})
 \quad\Longleftrightarrow\quad
 0\in\varrho(M_\lambda),
\end{equation}
and
\begin{equation}\label{iff-sp}
 \lambda\in\sigma_{\natural}(H_{S_2^\circ})
 \quad\Longleftrightarrow\quad
 0\in\sigma_{\natural}(M_\lambda),
 \qquad \natural=p,c,d,\mathrm{ess}.
\end{equation}
Both
residual spectra are empty, since the two operators are self-adjoint.

Moreover,
\begin{equation}\label{Ml}
 \ker(H_{S_2^\circ}-\lambda)
 =
 \left\{G_\lambda^\circ \xiq
 \oplus \xiq
 :\xiq
 \in\ker(M_\lambda)\right\},
\end{equation}
and the map $\xiq
\mapsto G_\lambda^\circ \xiq
\oplus \xiq
$ is a linear isomorphism
between the two kernels.  In particular,
\begin{equation}\label{dimkerid}
 \dim\ker(H_{S_2^\circ}-\lambda)=\dim\ker(M_\lambda).
\end{equation}

Ordinary and singular Weyl sequences can be transferred in both directions.
More precisely, if $\{\Psi_j\}$ is such a Weyl sequence for
$H_{S_2^\circ}$ at $\lambda$ and
\begin{equation*}
 (\varphi_j,\xiq
 _j):=( \mathcal T_\lambda
 )^{-1}\Psi_j,
\end{equation*}
then $\varphi_j\to0$, $\liminf_j\|\xiq
_j\|>0$, and
$\{\xiq
_j/\|\xiq
_j\|\}$ is a Weyl sequence of the same kind for $M_\lambda$ at
zero.  Conversely, if $\{\xiq
_j\}$ is an ordinary or singular Weyl sequence for
$M_\lambda$ at zero, then
\begin{equation*}
 \left\{
 \frac{G_\lambda^\circ \xiq
 _j\oplus \xiq
 _j}
 {\|G_\lambda^\circ \xiq
 _j\oplus \xiq
 _j\|}
 \right\}_j
\end{equation*}
is a Weyl sequence of the same kind for $H_{S_2^\circ}$ at $\lambda$.
\end{theorem}
\begin{proof}
For fixed $\lambda\in\RE\cap\varrho(H_1^\circ)$, $M_\lambda$ is
self-adjoint on $\H_2^\circ$: indeed, $S_2^\circ$ is $H_2^\circ$-small,
whereas $a(G^\circ-G_\lambda^\circ)$ is bounded and symmetric.  For every
$\lambda\in\RE\cap\varrho(H_1^\circ)$ the bounded invertibility of $\mathcal T_\lambda$ and \eqref{eq:abstract-Lee-factorization} with $z=\lambda$ give 
  \eqref{iff-rho},
\eqref{Ml}, and the equivalence of injectivity, density of the range, and
closedness of the range.  These properties yield equivalence between the point, continuous,
and full spectrum.  The kernel isomorphism gives
\eqref{dimkerid}.

It remains to justify explicitly the assertions concerning the
essential spectrum and Weyl sequences.  If $\{\xiq
_j\}$ is a normalized
singular Weyl sequence for $M_\lambda$ at zero, then
$G_\lambda^\circ \xiq
_j\rightharpoonup0$, and
\begin{equation*}
 1\leq\|G_\lambda^\circ \xiq
 _j\oplus \xiq
 _j\|
 \leq\sqrt{1+\|G_\lambda^\circ\|^2}.
\end{equation*}
The action formula above with $\varphi=0$ therefore proves that the normalized
lift is a singular Weyl sequence for $H_{S_2^\circ}$.

Conversely, let $\{\Psi_j\}$ be a normalized singular Weyl sequence for
$H_{S_2^\circ}$ and write
$(\varphi_j,\xiq
_j)=(\mathcal T_\lambda
)^{-1}\Psi_j$.
The boundedness of $\mathcal T_\lambda$ and, again, \eqref{eq:abstract-Lee-factorization} with $z=\lambda$ 
imply
\begin{equation*}
 (H_1^\circ-\lambda)\varphi_j\longrightarrow0,
 \qquad
 M_\lambda \xiq
 _j\longrightarrow0.
\end{equation*}
Hence $\varphi_j\to0$.  Moreover,
$(\varphi_j,\xiq
_j)\rightharpoonup0$ because
$(\mathcal T_\lambda
)^{-1}$ is bounded.  If a subsequence of $\xiq
_j$
converged to zero in norm, then
$\Psi_j=\mathcal T_\lambda
(\varphi_j,\xiq
_j)$ would converge to zero along
that subsequence, contradicting $\|\Psi_j\|=1$.  Thus
$\liminf_j\|\xiq
_j\|>0$, and normalization gives a singular Weyl sequence for
$M_\lambda$ at zero.  The same argument without weak convergence allows to treat
ordinary Weyl sequences.  Weyl's criterion now gives the essential-spectrum
equivalence in \eqref{iff-sp}; together with the point-spectrum and
multiplicity statements, it also gives the discrete-spectrum equivalence.
\end{proof}
\begin{remark}\label{sss:vacuum} Suppose there exists $\Omega_{1}^{\circ}\in \H_{1}^{\circ}$, $\|\Omega_{1}^{\circ}\|=1$,   such that $a\,\Omega_{1}^{\circ}=H_{1}^{\circ}\Omega_{1}^{\circ}=0$. $\Omega_{1}^{\circ}$ plays the role of the vacuum in $\F_{1}^{\circ}$.  Then, $\Omega_{1}^{\circ}\oplus 0\in\S^{\circ}$ and $$H_{S^{\circ}_{2}}(\Omega_{1}^{\circ}\oplus 0)=0\oplus0\,.
$$ 
Hence, $0\in \sigma(H_{S^{\circ}_{2}})\cap \sigma(H_{1}^{\circ})$ and this result can not be obtained by solving \eqref{Ml}, i.e, is not contained in Theorem \ref{teo-sp}. \par 
\end{remark}
\begin{remark}
Suppose there exists $\psi_{2}^{\circ}\in \H_{2}^{\circ}\backslash\{0\}$   such that $H_{2}^{\circ}\psi_{2}^{\circ}=0$. Further suppose that 
$$
S_{2}^{\circ}\psi_{2}^{\circ}=s^{\circ}\psi_{2}^{\circ}\,,\qquad 
a(G^{\circ}-G^{\circ}_{\lambda})\psi_{2}^{\circ}=f^{\circ}(\lambda)\psi_{2}^{\circ}\,,
$$
where $s^{\circ}\in\RE$ and $f^{\circ}:\RE\cap\varrho(H_{1}^{\circ})\to\RE$. Then, $\lambda\in\RE\cap\varrho(H_{1}^{\circ})$ is an eigenvalue of $H_{S^{\circ}_{2}}$, with corresponding eigenvector $G_{\lambda}^{\circ}\psi_{2}^{\circ}\oplus\psi_{2}^{\circ}$, if and only if $\lambda$ solves the equation
\begin{equation*}\label{evEq}
\lambda+s^{\circ}+f^{\circ}(\lambda)=0\,.
\end{equation*}
\end{remark}

\section{The renormalized Lee model} 
\label{sec:model-sector-notation}

Here we specialize the results in Subsection \ref{s:abstractlee} to the case in which $\F^{\circ}_{1}=\F^{\circ}_{2}=\F^{\circ}$, so that $\F=\F^{\circ}\oplus \F^{\circ}\simeq\CO^{2}\otimes\F^{\circ}$, where $\F^{\circ}$ is the bosonic Fock space for the single-particle Hilbert space $L^{2}(\RE^{3})$, $H^{\circ}_{1}=H^{\circ}_{2}=H^{\circ}$, where $H^{\circ}$ is the free field Hamiltonian and $a$ is the standard annihilation operator.
\par
We use lower case letters for points in $\RE^3$, so that $k,k_1, \dots , k_r\in \RE^3$; while capital letters denote points in $\RE^{3r}$, $r\geq 2$, e.g.,  $K= (k_1,\dots,k_r)\in \RE^{3r}$.  For $\psi^{(r)}: \RE^{3r}\to \CO$  a measurable function, we use the notation $\psi^{(r)}(K)  = \psi^{(r)}(k_1,\dots, k_r)$. When no confusion can arise, we omit the superscript $(r)$ and simply write $\psi(K)$. 

For $K=(k_1,\ldots,k_r)$ and $j\in\{1,\ldots,r\}$, we write
\begin{equation*}
  \hat K_j
  :=(k_1,\ldots,k_{j-1},k_{j+1},\ldots,k_r)\in \RE^{3(r-1)}
\end{equation*}
for the tuple obtained by deleting the $j$-th variable, and, for $p\in\RE^3$,
\begin{equation*}
  \hat K_{j,p}
  :=(k_1,\ldots,k_{j-1},p,k_{j+1},\ldots,k_r)\in \RE^{3r}
\end{equation*}
for the tuple obtained by replacing the $j$-th variable with $p$. $\hat K_{j,p} = p$ for $r=1$.

We denote by $L_b^2(\RE^{3r})\subseteq L^2(\RE^{3r})$ the subspace of functions which enjoy  bosonic symmetry. For $r=0$,  we set $L^2_b(\RE^{0}) =L^2(\RE^{0}) = \CO$. 

\subsection{Fock space and free Hamiltonian}
We set $\F^{\circ} = \F_b$, where $\F_b$ denotes the bosonic Fock space
\begin{equation*}
  \F_b:=\bigoplus_{n=0}^{\infty}L_b^2(\RE^{3n}).
\end{equation*}  Given $\psi\in \F_b$ we denote by $\psi^{(n)} = \psi^{(n)}(K)$ its  component corresponding to the $n$-boson state, with the understanding that $\psi^{(0)} \in \CO$. We
write $\Omega^\circ= 1_\CO\oplus_{n=1}^\infty 0_{L_b^2(\RE^{3n})}$ for the Fock vacuum.  The one-boson
dispersion relation is
\begin{equation*}
  \omega(k):=|k|^2+m,
  \qquad m>0,
\end{equation*}
and the free boson Hamiltonian is
\begin{equation*}
  H^\circ:=\mathrm d\Gamma(\omega):\dom(H^\circ)\subseteq\F_{b}\to\F_{b},
\end{equation*}
\begin{equation*}
\dom(H^\circ):= \bigoplus_{n=0}^\infty L_b^2(\RE^{3n}, ((|K|^2 + n m)^2+1) dK),
\end{equation*}
The operator $H^\circ$ is self-adjoint, non-negative, and its spectrum is given by 
\begin{equation*}
\sigma(H^\circ) = \{0\} \cup  [m,+\infty), \; \sigma_{ess}(H^\circ) =\sigma_{ac}(H^\circ) =  [m,+\infty).  
\end{equation*}
$0$ is a simple eigenvalue with eigenvector $\Omega^\circ$.

The  restriction of $H^\circ$ to the $n$-boson subspace will be denoted by
$H^{\circ,(n)}$.  Thus, for $n\geq1$,
\begin{equation}\label{eq:free-sector-action}
  (H^{\circ,(n)}\varphi)(K)
  =\bigl(|K|^2+nm\bigr)\varphi(K),
  \qquad K=(k_1,\ldots,k_n)\in\RE^{3n},
\end{equation}
with maximal domain
\begin{equation}\label{eq:free-sector-domain}
  \dom(H^{\circ,(n)})
  =L_b^2(\RE^{3n}, (|K|^2 + n m)^2 dK).
\end{equation}
On $L_b^2(\RE^0)=\CO$ we set $H^{\circ,(0)}=0$.  

Correspondingly, we have a scale of Hilbert spaces $\H^{s}_b$, $s\in [-1,1]$, defined according to Definition \ref{d:SHS}; these are given by 
\begin{equation*}\label{Hsb}
\H^{s}_b:= \bigoplus_{n=0}^\infty  \H^{(n),s}_b\qquad \H^{(n),s}_b:=  L_b^2(\RE^{3n}, ((|K|^2 + n m)^2+1)^s dK)\,.
\end{equation*}
In what follows we set $\H^{(n)}_b := \H^{(n),1}_b$.

For all  $z\in \rho(H^{\circ})$, we denote by  $R_z^\circ$ the resolvent of $H^\circ$: $R_z^\circ = (-H^\circ +z)^{-1}$. We point out the explicit formula    
\begin{equation*}
(R_z^\circ \psi)^{(n)}(K)  = \frac{\psi^{(n)}(K)}{-|K|^2 - nm + z }\,, \qquad n\geq0 \,,
\end{equation*}
where it is understood that for $n=0$ the variable $K$ is absent.  There holds 
\begin{equation}\label{RzBound}
R_z^\circ \in \B(\H^s_b,\H^{s+1}_b) \qquad s\in[-1,0].
\end{equation}
\subsection{The non-interacting Hamiltonian}
For a given parameter $\mu>0$, representing the energy gap between the ``spin up state'' and the ``spin down state'', we define the non-interacting Hamiltonian  by 
\begin{equation*}
\Hni:=H^{\circ}\oplus (H^{\circ}+\mu) \qquad \dom(\Hni) := \dom(H^{\circ}) \oplus \dom(H^{\circ}).
\end{equation*}
$\Hni$ is self-adjoint, lower bounded and its spectrum is given by 
\begin{equation*}
\sigma(\Hni) =\left \{ \begin{aligned}
& \{0\} \cup \{\mu\} \cup  [m,+\infty) \qquad && 0<\mu <m \\
& \{0\} \cup [m,+\infty)  && \mu\geq m.
 \end{aligned} \right.
\end{equation*}
We remark that $\sigma_{ess}(\Hni) =   [m,+\infty)$, and, if $0<\mu <m$,  $\sigma_{ess}(\Hni) = \sigma_{ac}(\Hni)$;  $0$ and $\mu$ are eigenvalues of $\Hni$ with eigenvectors $\Omega^{\circ}\oplus 0$ and $0\oplus \Omega^{\circ}$ respectively; if $\mu \geq m$ the eigenvalue $\mu$, with eigenvector $0\oplus \Omega^{\circ}$, is embedded in the essential spectrum. 
\subsection{The Renormalized Lee Hamiltonian $H_{\mu}$}
Given the coupling constant $g\in \RE$, we introduce the annihilator  
\begin{equation}\label{an0}
a: \H_b \to \F_b   
\end{equation}
defined by
\begin{equation}\label{an}
(a \psi)^{(n)}(K) := g \sqrt{n+1}  \int_{\RE^3} \psi^{(n+1)}(K,p) \, dp \qquad K\in \RE^{3n}, \; n\geq 0 \,;
\end{equation}

\begin{remark}\label{r:dense-kernel} For any $\phi\in L^{2}(\RE^{3})$, let us denote by $e(\phi)$ the exponential (or coherent) vector in $\F_{b}$ defined by 
\begin{equation*}
e(\phi)=\{ e(\phi)^{(n)}\}_{n=0}^{\infty}\,,\qquad e(\phi)^{(0)}:=1\,,\quad e(\phi)^{(n)}:=\frac1{\sqrt{n!}}\,\underbrace{\phi\otimes\dots\otimes\phi}_{n\text{-times}}\,,\ n\in\NA\,.
\end{equation*}
For any $\phi\in L^{2}(\RE^{3},(|k|^{2}+m)^{2}dk)$, one has
\begin{equation*}
a(e(\phi))=g \int_{\RE^{3}}\phi(p)\,dp\ e(\phi)
\end{equation*}
and so $\ker(a)$ contains the subspace $\mathscr E(V)\subset\F_{b}$ generated by $\{e(\phi)\,,\,\phi\in V\}$, where 
\begin{equation*}
V:=\left\{\phi\in L^{2}(\RE^{3},(|k|^{2}+m)^{2}dk): \int_{\RE^{3}}\phi(p)\,dp=0\right\}\,.
\end{equation*}
Notice that $V$ is dense in $L^{2}(\RE^{3})$ since it contains the set of Fourier transforms of the Schwartz functions vanishing at the origin; therefore, by \cite[Theorem 5.34(ii)]{Arai}, $\mathscr E(V)$ is  dense in $\F_b$ and  $\ker(a)$ is dense as well. \par
\end{remark}
\begin{remark}\label{r:lonigro}
For every $s\in(3/4,1]$, the annihilation map satisfies
\begin{equation*}
 a\in\B(\H_b^s,\F_b),
 \qquad
 a^*\in\B(\F_b,\H_b^{-s}).
\end{equation*}
The first estimate is \cite[Proposition~3.4]{Lonigro2022}, applied to the
constant form factor with the exponent $2s>3/2$; the second follows by
duality.  Together with Remark~\ref{r:dense-kernel}, this verifies the
assumptions on $a$ in Theorem~\ref{teo}.
\end{remark}
By Eq.~\eqref{RzBound} and Remark~\ref{r:lonigro},
$aR_z^\circ\in\B(\H^s_b,\F_b)$ for every
$s\in(-1/4,0]$ and $z\in \rho(H^{\circ})$. The operator
$G^\circ_z := R^\circ_z a^*:\F_b  \to \F_b$, $z \in \rho(H^{\circ})$,
belongs to $\B(\F_b, \H^s_b)$ for every $s \in [0,1/4)$.

We fix a subtraction point
\begin{equation*}
  \nu\in\RE\cap\rho(H^\circ),
\end{equation*}
where necessarily $\nu<m$, and an auxiliary point $\lambda_\star<0$.  The parameter
$\lambda_\star$ is used only in the description of the domain of the renormalized Hamiltonian; the resulting
Hamiltonian is independent of its choice. 

To proceed,  we introduce  the operator $S^{\circ}$ defined by
\begin{equation}\label{Scirc}
S^{\circ} :\H_b\subseteq\F_{b} \to \F_b \,,\qquad
S^{\circ} := S^{\circ,d} + S^{\circ,od}\,,
\end{equation}
where
\begin{align}\label{Scircdn}
    (S^{\circ,d}\psi)^{(n)}(K)  =&
    \left( -\mu +g^2 \int_{\RE^3} \frac{(-|K|^2-nm+\lambdastar -\nu)}{(-|p|^2 - m + \nu)(-|p|^2-|K|^2 - (n+1)m + \lambdastar ) }\ dp 
\right)\psi^{(n)}(K)\nonumber
\\
=&\left(-\mu + 2\pi^2 g^2(\sqrt{m-\nu}-\sqrt{|K|^2+(n+1)m-\lambdastar }\,)\right) \psi^{(n)}(K) \,, \quad n\geq0\,,
    \end{align}
where it is understood that for $n=0$ the variable $K$ is absent; and
\begin{equation}\label{Scircodn}
(S^{\circ,od}\psi)^{(n)}(K) 
= -g^2 \sum_{j=1}^n   \int_{\RE^3} \frac{\psi^{(n)}(\hat K_j,p) }{ -|p|^2-|K|^2 - (n+1)m + \lambdastar  }
  \, dp\qquad n\geq 1,
\end{equation} 
with $(S^{\circ,od}\psi)^{(0)} = 0$. 
\begin{lemma}\label{lem:S-construction-estimates}
Let $s\in(3/4,1]$ and set $\alpha_s:=s-1/4\in(1/2,3/4]$.
The operators in \eqref{Scircdn} and \eqref{Scircodn}, initially
defined on finite particle vectors with smooth and compactly supported 
components, extend to
\begin{equation*}
 S^{\circ,d}\in\B(\H_b^{1/2},\F_b),
 \qquad
 S^{\circ,od}\in\B(\H_b^{\alpha_s},\F_b).
\end{equation*}
Their restrictions to $\H_b=\dom(H^{\circ})$ are symmetric and consequently,
$S^\circ=S^{\circ,d}+S^{\circ,od}$ is well defined and symmetric on
$\H_b$.\\ Moreover $S^\circ$ is infinitesimally $H^\circ$-small: for every
$\varepsilon\in(0,1)$ there is $b_\varepsilon\in\RE$ such that
\begin{equation}\label{inf-small}
 \|S^\circ\psi\|
 \leq \varepsilon\|H^\circ\psi\|+b_\varepsilon\|\psi\|,
 \qquad \psi\in\H_b.
\end{equation}
\end{lemma}

\begin{remark}We postpone the proof of Lemma \ref{lem:S-construction-estimates} to comment on the definition of the renormalized Lee Hamiltonian $H_\mu$.

$S^{\circ}$ depends on the parameters $\nu\in\RE\cap\rho(H^{\circ})$, $\mu>0$ and $\lambdastar <0$, even though we omit such dependence from the notation. 

Remarks~\ref{r:dense-kernel} and~\ref{r:lonigro}, together with Lemma~\ref{lem:S-construction-estimates}, verify the hypotheses of Theorem~\ref{teo}.  Hence, posing $\F^{\circ}_{1}=\F^{\circ}_{2}= \F_b$, $H^{\circ}_{1}=H^{\circ}_{2}=H^{\circ}$, $a$ as in Eqs.~\eqref{an0}--\eqref{an}, and $S_{2}^{\circ}=S^{\circ}$, the operator 
\begin{equation}
H_\mu\equiv H_{S^{\circ}}
\end{equation}
defined according to Eqs.~\eqref{HS1}--\eqref{HS4} is self-adjoint and bounded from below; its resolvent is provided in Theorem~\ref{Res-Lee}.  As made explicit in \eqref{eq:mu-g-definition}, the dependence on $\mu$ and on the subtraction point $\nu$ enters through the renormalized gap $\mu_g$; changing $\nu$ while keeping the same Hamiltonian therefore requires the corresponding change of $\mu$.  The dependence on $g$ is suppressed in the notation.  By contrast, $H_\mu$ is independent of the auxiliary point $\lambdastar$, which is used only to characterize its domain and action.
\end{remark}

\begin{proof}[Proof of Lemma \ref{lem:S-construction-estimates}]
Write
\begin{equation*}
 \omega(k):=|k|^2+m,
 \qquad
 E_n(K):=\sum_{j=1}^n\omega(k_j)=|K|^2+nm,
\end{equation*}
and let $\mathscr C_{\rm fin}$ be the algebraic direct sum of the symmetric
smooth and compactly supported functions.  For $s>3/4$ the constant
\begin{equation*}
 c_s:=\int_{\RE^3}\omega(p)^{-2s}\,dp
\end{equation*}
is finite.

For the diagonal term, its real-valued multiplier in \eqref{Scircdn} has
absolute value bounded by $C_0+C_1E_n(K)^{1/2}$, with constants independent
of $n$.  Hence
\begin{equation*}
 \|S^{\circ,d}\psi\|
 \leq C\bigl(\|(H^\circ)^{1/2}\psi\|+\|\psi\|\bigr).
\end{equation*}

It remains to estimate the off-diagonal term uniformly in the boson number.
For $n\geq1$, rewrite \eqref{Scircodn} as
\begin{equation*}
 (S^{\circ,od}\psi)^{(n)}(K)
 =g^2\sum_{j=1}^n X_j(K),
 \qquad
 X_j(K):=\int_{\RE^3}
 \frac{\psi^{(n)}(\hat K_j,p)}{E_n(K)+\omega(p)-\lambdastar}\,dp.
\end{equation*}
The numbers $w_j(K):=\omega(k_j)/E_n(K)$ are positive and sum to one.
Thus a weighted Cauchy--Schwarz inequality yields
\begin{equation*}
 \left|\sum_{j=1}^nX_j(K)\right|^2
 \leq\sum_{j=1}^n\frac{E_n(K)}{\omega(k_j)}|X_j(K)|^2.
\end{equation*}
A second Cauchy--Schwarz inequality, now in $p$, gives
\begin{equation*}
 |X_j(K)|^2
 \leq c_s\int_{\RE^3}
 \frac{\omega(p)^{2s}|\psi^{(n)}(\hat K_j,p)|^2}
 {\bigl(E_n(K)+\omega(p)-\lambdastar\bigr)^2}\,dp.
\end{equation*}
Put $E_{\hat j}:=\sum_{\ell\neq j}\omega(k_\ell)$.  Writing
$K=(\hat K_j,k_j)$, so that $dK=d\hat K_j\,dk_j$ and
$E_n(K)=E_{\hat j}+\omega(k_j)$, combining the two preceding inequalities,
integrating over $K$, and applying Tonelli's theorem give
\begin{align*}
 \|(S^{\circ,od}\psi)^{(n)}\|^2
 &\leq g^4c_s\sum_{j=1}^n
 \int_{\RE^{3(n-1)}}d\hat K_j\int_{\RE^3}dp\,
 \omega(p)^{2s}|\psi^{(n)}(\hat K_j,p)|^2
 I_j(\hat K_j,p),
\end{align*}
where, for fixed $\hat K_j$ and $p$, only the variable $k_j$ is integrated
in
\begin{equation*}
 I_j(\hat K_j,p):=
 \int_{\RE^3}
 \frac{E_{\hat j}+\omega(k_j)}
 {\omega(k_j)\bigl(E_{\hat j}+\omega(k_j)+\omega(p)-\lambdastar\bigr)^2}
 \,dk_j.
\end{equation*}
Since $\omega(p)-\lambdastar>0$, relabelling the integration variable $k_j$
as $k$ gives
\begin{equation*}
 I_j(\hat K_j,p)
 \leq J_j(\hat K_j,p):=
 \int_{\RE^3}
 \frac{dk}{(|k|^2+m)
 \bigl(|k|^2+E_{\hat j}+m+\omega(p)-\lambdastar\bigr)}.
\end{equation*}
For the fixed variables $\hat K_j$ and $p$, set
\begin{equation*}
 A:=E_{\hat j}+m+\omega(p)-\lambdastar>m.
\end{equation*}
Passing to spherical coordinates and using the elementary partial fraction identity
\begin{equation*}
 \frac{r^2}{(r^2+m)(r^2+A)}
 =\frac{1}{A-m}
 \left(
 \frac{A}{r^2+A}
 -\frac{m}{r^2+m}
 \right),
\end{equation*}
we obtain
\begin{align*}
 J_j(\hat K_j,p)
 &=4\pi\int_0^\infty
 \frac{r^2}{(r^2+m)(r^2+A)}\,dr \\
 &=\frac{4\pi}{A-m}
 \left(
 A\frac{\pi}{2\sqrt A}
 -m\frac{\pi}{2\sqrt m}
 \right) \\
 &=\frac{2\pi^2(\sqrt A-\sqrt m)}
 {A-m}
 =\frac{2\pi^2}{\sqrt m+\sqrt A}.
\end{align*}
Here we used
\begin{equation*}
 \int_0^\infty\frac{dr}{r^2+a}=\frac{\pi}{2\sqrt a},
 \qquad a>0.
\end{equation*}
Since $m-\lambdastar>0$, it follows that
\begin{equation*}
 I_j(\hat K_j,p)
 \leq \frac{2\pi^2}
 {\sqrt m+\sqrt{E_{\hat j}+m+\omega(p)-\lambdastar}}
 \leq \frac{2\pi^2}{\sqrt{E_{\hat j}+\omega(p)}}.
\end{equation*}
Substituting this bound in the preceding norm
estimate, relabelling $p$ as the $j$th momentum variable, and using bosonic
symmetry, we obtain
\begin{align*}
 \|(S^{\circ,od}\psi)^{(n)}\|^2
 &\leq 2\pi^2g^4c_s
 \int_{\RE^{3n}}
 \frac{\sum_{j=1}^n\omega(k_j)^{2s}}
 {E_n(K)^{1/2}}|\psi^{(n)}(K)|^2\,dK \\
 &\leq 2\pi^2g^4c_s
 \int_{\RE^{3n}}E_n(K)^{2s-1/2}
 |\psi^{(n)}(K)|^2\,dK.
\end{align*}
Since $2s-1/2=2\alpha_s$, summation over $n$ gives
\begin{equation*}
 \|S^{\circ,od}\psi\|
 \leq (2\pi^2g^4c_s)^{1/2}
 \|(H^\circ)^{\alpha_s}\psi\|.
\end{equation*}
This proves the asserted extension of $S^{\circ,od}$.

The diagonal term is symmetric because its multiplier is real.  For the
off-diagonal term, Fubini's theorem gives on $\mathscr C_{\rm fin}$
\begin{equation*}
 \langle\phi,S^{\circ,od}\psi\rangle
 =g^2\sum_{n\geq1}\sum_{j=1}^n
 \int
 \frac{\overline{\phi^{(n)}(\hat K_j,k_j)}
 \psi^{(n)}(\hat K_j,p)}
 {E_{\hat j}+\omega(k_j)+\omega(p)-\lambdastar}
 \,d\hat K_j\,dk_j\,dp.
\end{equation*}
The denominator is invariant under $k_j\leftrightarrow p$, so the last
expression equals $\langle S^{\circ,od}\phi,\psi\rangle$.  The bounds just
proved and density extend this identity to $\H_b$.

Finally, for every $\beta\in(0,1)$ and every $\varepsilon>0$, the elementary
scalar estimate $x^{2\beta}\leq\varepsilon^2x^2+C_{\varepsilon,\beta}$ and
the spectral calculus imply, after adjusting the constants,
\begin{equation*}
 \|(H^\circ)^\beta\psi\|
 \leq\varepsilon\|H^\circ\psi\|+C'_{\varepsilon,\beta}\|\psi\|.
\end{equation*}
Applying this with $\beta=1/2$ and $\beta=\alpha_s<1$ to the two bounds above
proves \eqref{inf-small}.
\end{proof}
\begin{remark}\label{r:lampart-estimate}
The infinitesimal $H^\circ$-bound in \eqref{inf-small} can also be obtained
from \cite[Proposition~3.1 and the proof of Corollary~3.2]{Lampart2025}.  In
the notation of that paper one takes the constant form factor $F=g$, which
belongs to $\mathfrak{b}_q$ for every $q\in(3/2,2)$.  There the relevant estimate
is not a separate statement for the operator $S^\circ$:
it is used in the proof of the corollary and is embedded in the authors'
framework of interior-boundary conditions.  We retain the direct proof
above because it gives the diagonal and off-diagonal estimates in the
present notation and avoids introducing that additional framework.
\end{remark}

For later convenience, for any $z\in\varrho(H^{\circ})$, we define the operator (see Theorem \ref{Res-Lee})
\begin{equation}\label{Mz1}
\dom(M_z) := \dom(H^{\circ}).
\end{equation}
\begin{equation}\label{Mz2}
M_{z}\equiv M^{\circ}_{z}:=-H^{\circ}+S^{\circ}+a(G^{\circ}-G^{\circ}_{z})+z=-H^{\circ}+S^{\circ}+(z-\lambda_{\star})\,{G^{\circ}}^{*}G^{\circ}_{z}+z\,,
\end{equation}

\subsection{Decomposition in subspaces with fixed number of excitations\label{ss:decomposition}}
Here we exploit the conservation of the number of quanta by $H_{\mu}$. 

\subsubsection{Operators in subspaces with fixed number of bosons}
We start by introducing the operators in spaces with fixed numbers of bosons. In this discussion we use $\ell\in\NA_0$  as a dummy boson-number index.  The field Hamiltonian $H^{\circ,(\ell)}$  in the $\ell$-boson subspace was introduced in Eqs. \eqref{eq:free-sector-action} and \eqref{eq:free-sector-domain}. 
 
The restriction of the annihilation map $a$ from the
$(\ell+1)$-boson subspace to the $\ell$-boson subspace is denoted by $a^{(\ell)}$:
\begin{equation*}\label{eq:sector-annihilation1}
  a^{(\ell)} : \dom(H^{\circ,(\ell+1)}) \to L^2(\RE^{3\ell})
\end{equation*}
\begin{equation*}\label{eq:sector-annihilation2}
  (a^{(\ell)}\varphi)(K)
  :=g\sqrt{\ell+1}\int_{\RE^3}\varphi(K,p)\,dp;
\end{equation*}
$a^{(\ell)} \in \B(\dom(H^{\circ,(\ell+1)}),L^2(\RE^{3\ell}))$. The corresponding singular creation map is denoted
by $a^{*,(\ell+1)}$, formally $a^{*,(\ell+1)} = (a^{(\ell)})^*$, and its action is given by 
\begin{equation*}\label{eq:sector-creation}
  (a^{*,(\ell+1)}\xiq
  )(K)
  :=\frac{g}{\sqrt{\ell+1}}
  \sum_{j=1}^{\ell+1}\xiq
  (\hat K_j).
\end{equation*}
Although $a^{*,(\ell+1)}\xiq
$ need not belong to $L_b^2(\RE^{3(\ell+1)})$, its
composition with the free resolvent is bounded.

For
\begin{equation*}
  z\in\rho(H^{\circ,(\ell+1)})
  =\CO\setminus[(\ell+1)m,+\infty),
\end{equation*}
set
\begin{equation*}
  R_z^{\circ,(\ell+1)}:=(-H^{\circ,(\ell+1)}+z)^{-1}
  \quad \text{and}\quad 
  G_z^{\circ,(\ell+1)}
  :=R_z^{\circ,(\ell+1)}a^{*,(\ell+1)},
\end{equation*}
with $G^{\circ,(\ell+1)}:=G_{\lambdastar}^{\circ,(\ell+1)}$.
The explicit action of $G_z^{\circ,(\ell+1)}$ is
\begin{equation}\label{appcert:eq:G-lambda-explicit}
  (G_z^{\circ,(\ell+1)}\xiq
  )(K)
  =\frac{g}{\sqrt{\ell+1}}
  \frac{\displaystyle\sum_{j=1}^{\ell+1}\xiq
  (\hat K_j)}
  {-|K|^2-(\ell+1)m+z}.
\end{equation}
For every such $z$,
\begin{equation*}\label{appcert:eq:G-bounded-map}
  G_z^{\circ,(\ell+1)}
  \in\mathcal B\bigl(L_b^2(\RE^{3\ell}),L_b^2(\RE^{3(\ell+1)})\bigr).
\end{equation*}
Indeed, if $z=\lambda<(\ell+1)m$, then
\begin{equation*}
  \|G_\lambda^{\circ,(\ell+1)}\xiq
  \|_2^2
  \leq
  \frac{(\ell+1)\pi^2|g|^2}
  {\sqrt{(\ell+1)m-\lambda}}\,\|\xiq
  \|_2^2.
\end{equation*}
For non-real $z$, boundedness follows from the resolvent identity
\begin{equation*}
 G_z^{\circ,(\ell+1)}
 =\bigl(1+(\lambda-z)R_z^{\circ,(\ell+1)}\bigr)
 G_\lambda^{\circ,(\ell+1)},
 \qquad \lambda<(\ell+1)m.
\end{equation*}

We now write down  the  restrictions of the operator $S^\circ$ from Eqs. \eqref{Scirc} -- \eqref{Scircodn}.  The square root below is taken with the branch which is
positive on $(0,+\infty)$:
\begin{equation*}\label{eq:S-diagonal-sector}
\begin{aligned}
  (S^{\circ,d,(\ell)}\xiq
  )(K)
  :={}&\left[-\mu+2\pi^2g^2
  \left(\sqrt{m-\nu}
  -\sqrt{|K|^2+(\ell+1)m-\lambdastar}\right)\right]\xiq
  (K),
\end{aligned}
\end{equation*}
where for $\ell=0$ the variable $K$ is absent; for $\ell\geq1$, set
\begin{equation}\label{appcert:eq:offdiagonal-sector-definition}
  (S^{\circ,od,(\ell)}\xiq
  )(K)
  :=g^2\sum_{j=1}^{\ell}\int_{\RE^3}
  \frac{\xiq
  (\hat K_{j,p})}
  {|p|^2+|K|^2+(\ell+1)m-\lambdastar}\,dp,
\end{equation}
and set $S^{\circ,od,(0)}:=0$.  Finally,
\begin{equation*}\label{eq:S-sector-decomposition}
  S^{\circ,(\ell)}
  :=S^{\circ,d,(\ell)}+S^{\circ,od,(\ell)}.
\end{equation*}

To conclude we write the sector restrictions of the operator $M_z$ from Eqs. \eqref{Mz1} -- \eqref{Mz2}:
for $z\in\rho(H^{\circ,(\ell+1)})=\CO\setminus[(\ell+1)m,+\infty)$, 
\begin{equation*}\label{eq:M-sector-definition-construction2}
  \dom(M_z^{(\ell)})=\dom(H^{\circ,(\ell)}).
\end{equation*}
\begin{align}  M_z^{(\ell)}\equiv   M_z^{\circ,(\ell)}
  :=&-H^{\circ,(\ell )}+S^{\circ,(\ell)}+a^{(\ell)}\bigl(G^{\circ,(\ell+1)}-G_z^{\circ,(\ell+1)}\bigl)+z\label{eq:M-sector-definition-construction1}
 \\
  =&-H^{\circ,(\ell)}+S^{\circ,(\ell)}+(z-\lambdastar)\,(G^{\circ,(\ell+1)})^*G^{\circ,(\ell+1)}_{z}+z\,. \nonumber
\end{align}

It is convenient to isolate and write in explicit form  its diagonal and off-diagonal  part in the $(n-1)$ boson sector.  Define the maximal
multiplication operator $M_z^{d,(n-1)}$ by the multiplier
\begin{equation}\label{appcert:eq:diagonal-sector-definition}
  m_z^{(n-1)}(K)
  :=|K|^2+(n-1)m-z+\mu_g
  +2\pi^2g^2\sqrt{|K|^2+nm-z},\qquad K\in\RE^{3(n-1)}
\end{equation}
and the operator $M_z^{od,(n-1)}$ by
\begin{equation*}\label{Mod}
  (M_z^{od,(n-1)}\xiq
  )(K)
  :=g^2\sum_{j=1}^{n-1}\int_{\RE^3}
  \frac{\xiq
  (\hat K_{j,p})}
  {|p|^2+|K|^2+n m-z}\,dp, \qquad K\in\RE^{3(n-1)}.
\end{equation*}
We point out that $M_z^{od,(n-1)}$ is obtained from $S^{\circ,od,(n-1)}$ by replacing $\lambdastar$ with $z$.

Then
\begin{equation*}\label{eq:M-D-S-relation}
  M_z^{(n-1)}
  =-M_z^{d,(n-1)}+M_z^{od,(n-1)}.
\end{equation*}
\subsubsection{Lee model in the $n$-excitation subspace and spectral reduction}

For fixed
total excitation number, define
\begin{equation*}\label{eq:sector-spaces}
  \IS^{[0]}:=\CO\oplus0,
  \qquad
  \IS^{[n]}
  :=L_b^2(\RE^{3n})\oplus L_b^2(\RE^{3(n-1)}),
  \qquad n\geq1.
\end{equation*}

For later use, let
\begin{equation*}
  \mathcal H_+^{[n]}:=L_b^2(\RE^{3n}),
  \qquad
  \mathcal H_-^{[n]}:=L_b^2(\RE^{3(n-1)}),
  \qquad n\geq1,
\end{equation*}
so that $\IS^{[n]}=\mathcal H_+^{[n]}\oplus\mathcal H_-^{[n]}$.
The signs are only labels for the two entries of the block decomposition.

Then
\begin{equation*}\label{eq:full-sector-decomposition}
  \F_b\oplus\F_b
  =\bigoplus_{n=0}^{\infty}\IS^{[n]}.
  \end{equation*}
  The explicit domain and action of $H_\mu$ preserve every subspace $\IS^{[n]}$.  Consequently, such a decomposition allows us to write $H_{\mu}$ as the direct sum
  \begin{equation*}\label{Hmufacotrized}
H_\mu  = \bigoplus_{n=0}^\infty H^{[n]}_\mu    \,,\qquad
H^{[n]}_\mu: \dom(H^{[n]}_\mu) \subset \IS^{[n]} \to \IS^{[n]}
\end{equation*}
with domain
\begin{equation*}\label{eq:full-direct-sum-domain}
 \dom(H_\mu)
 =\left\{(\Psi^{[n]})_{n\geq0}\in\bigoplus_{n=0}^{\infty}\IS^{[n]}:\
 \Psi^{[n]}\in\dom(H_\mu^{[n]}),\quad
 \sum_{n=0}^{\infty}\|H_\mu^{[n]}\Psi^{[n]}\|^2<\infty\right\}.
\end{equation*}
where, for $n\geq1$,
\begin{equation}\label{eq:sector-domain-basepoint}
\begin{aligned}
  \dom(H_\mu^{[n]}):
  =\Bigl\{&(\varphi+G^{\circ,(n)}\xiq
  )\oplus \xiq
  :\
  &\varphi\in\dom(H^{\circ,(n)}),\quad
  \xiq
  \in\dom(H^{\circ,(n-1)})\Bigr\},
\end{aligned}
\end{equation}
\begin{equation}\label{eq:sector-action-basepoint}
  H_\mu^{[n]}\bigl((\varphi+G^{\circ,(n)}\xiq
  )\oplus \xiq
  \bigr)
  =\bigl(H^{\circ,(n)}\varphi
  +\lambda_\star G^{\circ,(n)}\xiq
  \bigr) \oplus\bigl((H^{\circ,(n-1)}-S^{\circ,(n-1)})\xiq
  +a^{(n-1)}\varphi\bigr);
\end{equation}
and for $n=0$, $H_\mu^{[0]}=0$ on $\IS^{[0]}=\CO\oplus0$.
 
 The resolvent identity implies that
$G_z^{\circ,(n)}-G_{\lambda_\star}^{\circ,(n)}$ maps
$\mathcal H_-^{[n]}$ into $\dom(H^{\circ,(n)})$.  Consequently, the domain \eqref{eq:sector-domain-basepoint} can be equivalently parametrized at any real $\lambda<nm$ as
\begin{equation}\label{appcert:eq:lambda-domain-param}
\begin{aligned}
  \dom(H_\mu^{[n]})
  =\Bigl\{&(\varphi+G_\lambda^{\circ,(n)}\xiq
  )\oplus \xiq
  :\
  &\varphi\in\dom(H^{\circ,(n)}),\quad
  \xiq
  \in\dom(M_\lambda^{(n-1)})\Bigr\}.
\end{aligned}
\end{equation}
In this parametrization,
\begin{equation}\label{appcert:eq:sector-action-parametrized}
(H_\mu^{[n]}-\lambda)
  \bigl((\varphi+G_\lambda^{\circ,(n)}\xiq
  )\oplus \xiq
  \bigr)=(H^{\circ,(n)}-\lambda)\varphi
  \oplus\bigl(a^{(n-1)}\varphi-M_\lambda^{(n-1)}\xiq
  \bigr).
\end{equation}
In particular,
\begin{equation*}\label{eq:lifted-M-action}
  (H_\mu^{[n]}-\lambda)
  (G_\lambda^{\circ,(n)}\xiq
  \oplus \xiq
  )
  =0\oplus(-M_\lambda^{(n-1)}\xiq
  ).
\end{equation*}
We conclude this section with the following corollary to Theorems \ref{teo}, \ref{Res-Lee}, and \ref{teo-sp}.
\begin{corollary}\label{appcert:cor:M-reduction}
Let $n\geq1$ and let $\lambda<nm$.  Then the following assertions hold.
\begin{enumerate}[label=\textup{(\roman*)}]
\item The bounded triangular operator
\begin{equation*}\label{appcert:eq:T-lambda-definition}
 \mathcal T_\lambda^{[n]}
 :=
 \begin{pmatrix}
  1&G_\lambda^{\circ,(n)}\\
  0&1
 \end{pmatrix}
\end{equation*}
on $\mathcal H_+^{[n]}\oplus\mathcal H_-^{[n]}$ restricts to a bounded
isomorphism
\begin{equation*}
 \mathcal T_\lambda^{[n]}
:
 \dom(H^{\circ,(n)})\oplus\dom(M_\lambda^{(n-1)})
 \longrightarrow\dom(H_\mu^{[n]})
\end{equation*}
when the domains are equipped with their graph norms.  Its inverse is the
restriction of
\begin{equation*}
\mathcal T_\lambda^{[n]-1}
 =
 \begin{pmatrix}
  1&-G_\lambda^{\circ,(n)}\\
  0&1
 \end{pmatrix}.
\end{equation*}
Moreover,
\begin{equation}\label{appcert:eq:M-factorization}
 \mathcal T_\lambda^{[n]*}
 (H_\mu^{[n]}-\lambda)
\mathcal T_\lambda^{[n]}
 =
 \begin{pmatrix}
  H^{\circ,(n)}-\lambda&0\\
  0&-M_\lambda^{(n-1)}
 \end{pmatrix}.
\end{equation}

\item\label{item:spcor1} One has
\begin{equation*}\label{appcert:eq:rho-spectrum-equivalence}
\begin{aligned}
 \lambda\in\rho(H_\mu^{[n]})
 &\quad\Longleftrightarrow\quad
 0\in\rho(M_\lambda^{(n-1)}),\\
 \lambda\in\sigma(H_\mu^{[n]})
 &\quad\Longleftrightarrow\quad
 0\in\sigma(M_\lambda^{(n-1)}).
\end{aligned}
\end{equation*}
More precisely,
\begin{equation*}\label{iff-sp-[n]}
 \lambda\in\sigma_{\natural}(H_\mu^{[n]})
 \quad\Longleftrightarrow\quad
 0\in\sigma_{\natural}(M_\lambda^{(n-1)}),
 \qquad \natural=p,c,d,\mathrm{ess}.
\end{equation*}
Both
residual spectra are empty.  If the equivalent resolvent conditions hold,
then
\begin{equation}\label{appcert:eq:inverse-factorization}
 (H_\mu^{[n]}-\lambda)^{-1}
 =
\mathcal T_\lambda^{[n]}
 \begin{pmatrix}
  (H^{\circ,(n)}-\lambda)^{-1}&0\\
  0&-(M_\lambda^{(n-1)})^{-1}
 \end{pmatrix}
\mathcal T_\lambda^{[n]*}
.\end{equation}
Equivalently, for every
$z\in\rho(H_\mu^{[n]})\cap\bigl(\CO\setminus[nm,+\infty)\bigr)$,
\begin{equation}\label{appcert:eq:krein-M-resolvent}
(-H_{\mu}^{[n]}+z)^{-1}
=
\begin{pmatrix}
 R^{\circ,(n)}_{z}
 +G^{\circ,(n)}_{z}(M^{(n-1)}_{z})^{-1}
  (G^{\circ,(n)}_{\bar z})^*
 &
 G^{\circ,(n)}_{z}(M^{(n-1)}_{z})^{-1}
 \\
 (M^{(n-1)}_{z})^{-1}(G^{\circ,(n)}_{\bar z})^*
 &
 (M^{(n-1)}_{z})^{-1}
\end{pmatrix}.
\end{equation}

\item\label{item:spcor2} The map
\begin{equation*}\label{appcert:eq:eigenspace-isomorphism}
 \ker(M_\lambda^{(n-1)})\ni \xiq
 \longmapsto G_\lambda^{\circ,(n)}\xiq
 \oplus \xiq
 \in\ker(H_\mu^{[n]}-\lambda)
\end{equation*}
is a linear isomorphism.  In particular,
\begin{equation*}\label{appcert:eq:multiplicity-equality}
 \dim\ker(H_\mu^{[n]}-\lambda)
 =\dim\ker(M_\lambda^{(n-1)}).
\end{equation*}

\item\label{item:spcor3} Ordinary and singular Weyl sequences can be
transferred in both directions.  If $\{\Psi_j\}$ is such a Weyl sequence for
$H_\mu^{[n]}$ at $\lambda$ and
\begin{equation*}
 (\varphi_j,\xiq
 _j):=\mathcal T_\lambda^{[n]-1}
 \Psi_j,
\end{equation*}
then $\varphi_j\to0$, $\liminf_j\|\xiq
_j\|>0$, and
$\{\xiq
_j/\|\xiq
_j\|\}$ is a Weyl sequence of the same kind for
$M_\lambda^{(n-1)}$ at zero.  Conversely, if $\{\xiq
_j\}$ is an ordinary or
singular Weyl sequence for $M_\lambda^{(n-1)}$ at zero, then
\begin{equation*}\label{appcert:eq:lifted-singular-weyl-sequence}
 \left\{
 \frac{G_\lambda^{\circ,(n)}\xiq
 _j\oplus \xiq
 _j}
 {\|G_\lambda^{\circ,(n)}\xiq
 _j\oplus \xiq
 _j\|}
 \right\}_j
\end{equation*}
is a Weyl sequence of the same kind for $H_\mu^{[n]}$ at $\lambda$.
\end{enumerate}
\end{corollary}

\begin{proof}
Apply Theorem~\ref{teo} with the identifications
\begin{equation*}
 \F_1^\circ=\mathcal H_+^{[n]},
 \qquad
 \F_2^\circ=\mathcal H_-^{[n]},
 \qquad
 H_1^\circ=H^{\circ,(n)},
 \qquad
 H_2^\circ=H^{\circ,(n-1)},
\end{equation*}
and take $a=a^{(n-1)}$ and
$S_2^\circ=S^{\circ,(n-1)}$.  Equations
\eqref{eq:sector-domain-basepoint} and
\eqref{eq:sector-action-basepoint} then show that the abstract operator
$H_{S_2^\circ}$ is precisely $H_\mu^{[n]}$, while \eqref{eq:M-sector-definition-construction1}
compared with Eq. \eqref{Mz}, with $\lambda_{1,\star} = \lambda_{\star}$, identifies the reduced operator $M_z$ with $M_z^{(n-1)}$.

Since
\begin{equation*}
 \lambda<nm=\inf\sigma(H^{\circ,(n)}),
\end{equation*}
one has $\lambda\in\rho(H_1^\circ)$, as required in
Theorem~\ref{teo-sp}. Item~\textup{(i)} follows from the proof of Theorem \ref{Res-Lee}.
The factorization~\eqref{appcert:eq:M-factorization} and the bounded
invertibility of $\mathcal T_\lambda^{[n]}$ and its adjoint also give
the claimed equivalence of graph norms. 
Items~\textup{(ii)}--\textup{(iv)}  follow directly from
Theorem~\ref{teo-sp},  except \eqref{appcert:eq:inverse-factorization} that is obtained by taking inverses in
\eqref{appcert:eq:M-factorization} and 
\eqref{appcert:eq:krein-M-resolvent} that follows from 
Theorem~\ref{Res-Lee}.
\end{proof}
  In what follows,  we use the renormalized parameter
\begin{equation}\label{eq:mu-g-definition}
  \mu_g:=\mu-2\pi^2g^2\sqrt{m-\nu}.
\end{equation}
When $\mu_g<m$, we also set
\begin{equation}
  \lambda_1:=m-\left(\sqrt{\pi^4g^4+m-\mu_g}-\pi^2g^2\right)^2 . \label{eq:lambda-one-definition}
\end{equation}
If $g\neq0$, the equivalent expression
\begin{equation*}
 \lambda_1=\mu_g+2\pi^4g^4
 \left(\sqrt{1+\frac{m-\mu_g}{\pi^4g^4}}-1\right)
\end{equation*}
is sometimes useful.  In all cases, $\mu_g\leq\lambda_1<m$, with equality
$\mu_g=\lambda_1$ precisely when $g=0$.  Theorem~\ref{thm:n1} below shows
that $\lambda_1$ is a simple isolated eigenvalue of $H_\mu^{[1]}$.

For real $\lambda<nm$, we shall also use
\begin{equation*}\label{appcert:eq:md}
  m_\lambda(K)
  :=m_\lambda^{(n-1)}(K)
  =|K|^2+(n-1)m-\lambda+\mu_g
  +2\pi^2g^2\sqrt{|K|^2+nm-\lambda}.
\end{equation*}
In particular,
\begin{equation}\label{eq:M-scalar-sector-zero}
  M_z^{(0)}
  =z-\mu_g-2\pi^2g^2\sqrt{m-z}.
\end{equation}
The equation $M_\lambda^{(n-1)}\xiq
=0$ is therefore equivalent to
\begin{equation*}\label{appcert:eq:M-equation}
  \bigl(M_\lambda^{d,(n-1)}-M_\lambda^{od,(n-1)}\bigr)\xiq
  =0.
\end{equation*}
For reference, the off-diagonal term has the quadratic-form representation
\begin{equation*}\label{appcert:eq:sod-form}
  \langle \xiq
  ,M_\lambda^{od,(n-1)}\xiq
  \rangle
  =(n-1)g^2\int_{\RE^{3(n-2)}}d\hat K
  \int_{\RE^3}dk\int_{\RE^3}dp \frac{\overline{\xiq
  (\hat K,k)}\xiq
  (\hat K,p)}
  {|\hat K|^2+|k|^2+|p|^2+nm-\lambda},
\end{equation*}
when $n\geq2$.  

\section[Essential spectrum]{Inclusions in the essential spectrum}
\label{sec:essential-spectrum}
 \subsection{The free part}

\begin{theorem}\label{thm:free-essential-spectrum}
For every $n\geq1$,
\begin{equation}\label{eq:free-essential-inclusion}
  [nm,+\infty)\subset\sigma_{\mathrm{ess}}(H_\mu^{[n]}).
\end{equation}
\end{theorem}

\begin{proof}
Fix $\lambda\geq nm$.  Choose $q_\lambda\in\mathbb R^3$ such that
\begin{equation*}
  n|q_\lambda|^2=\lambda-nm.
\end{equation*}
Let $\eta\in C_c^\infty(\mathbb R^3)$ satisfy $\|\eta\|_2=1$, and set
\begin{equation*}
  \eta_\ell(k):=\ell^{3/2}\eta\bigl(\ell(k-q_\lambda)\bigr),
\end{equation*}
\begin{equation*}
  \Phi_\ell^{(n)}(k_1,\ldots,k_n)
  :=\prod_{j=1}^n\eta_\ell(k_j).
\end{equation*}
Then $\Phi_\ell^{(n)}\in C_c^\infty(\mathbb R^{3n})\cap
L_b^2(\mathbb R^{3n})$ and $\|\Phi_\ell^{(n)}\|_2=1$.  Moreover,
\begin{equation}\label{eq:free-packets-weak}
  \Phi_\ell^{(n)}\rightharpoonup0.
\end{equation}
Indeed, the measure of $\supp\Phi_\ell^{(n)}$ tends to zero, and for every
$F\in L^2(\mathbb R^{3n})$ one has
\begin{equation*}
  |\langle F,\Phi_\ell^{(n)}\rangle|
  \leq
  \|F\|_{L^2(\supp\Phi_\ell^{(n)})}
  \|\Phi_\ell^{(n)}\|_2
  \longrightarrow0.
\end{equation*}

Put
\begin{equation*}
  \Psi_\ell:=\Phi_\ell^{(n)}\oplus0\in\mathcal H_+^{[n]}\oplus
  \mathcal H_-^{[n]}.
\end{equation*}
The fixed basepoint domain parametrization \eqref{eq:sector-domain-basepoint}, with
$\new{\xiq}
=0$, shows that $\Psi_\ell\in\dom(H_\mu^{[n]})$.  The action of the
Hamiltonian on these vectors is
\begin{equation*}
  (H_\mu^{[n]}-\lambda)\Psi_\ell
  =
  \bigl((H^{\circ,(n)}-\lambda)\Phi_\ell^{(n)}\bigr)
  \oplus a^{(n-1)}\Phi_\ell^{(n)}.
\end{equation*}
Since the support of each one-particle factor shrinks to $q_\lambda$,
\begin{equation*}
  \|(H^{\circ,(n)}-\lambda)\Phi_\ell^{(n)}\|_2
  \leq
  \sup_{K\in\supp\Phi_\ell^{(n)}}
  \left||K|^2+nm-\lambda\right|
  \longrightarrow0.
\end{equation*}
Furthermore,
\begin{equation*}
  a^{(n-1)}\Phi_\ell^{(n)}
  =g\sqrt n
  \left(\int_{\mathbb R^3}\eta_\ell(p)\,dp\right)
  \prod_{j=1}^{n-1}\eta_\ell(k_j),
\end{equation*}
and hence
\begin{equation*}
\begin{aligned}
  \|a^{(n-1)}\Phi_\ell^{(n)}\|_2
  &=|g|\sqrt n
  \left|\int_{\mathbb R^3}\eta_\ell(p)\,dp\right| \\
  &=|g|\sqrt n\,\ell^{-3/2}
  \left|\int_{\mathbb R^3}\eta(p)\,dp\right|
  \longrightarrow0.
\end{aligned}
\end{equation*}
Together with \eqref{eq:free-packets-weak}, this proves that
$\{\Psi_\ell\}$ is a singular Weyl sequence for $H_\mu^{[n]}$ at
$\lambda$.  Since $\lambda\geq nm$ was arbitrary,
\eqref{eq:free-essential-inclusion} follows.
\end{proof}

\subsection{Inclusions generated by lower sectors}
The main result of this subsection is the following theorem. 

\begin{theorem}\label{thm:lower-sector-inclusion}
Let $n\geq2$, let $1\leq r<n$, and let
\begin{equation}\label{eq:lower-sector-spectral-point}
  \zeta\in\sigma(H_\mu^{[r]}),
  \qquad
  \zeta<rm.
\end{equation}
Then
\begin{equation}\label{eq:lower-sector-essential-inclusion}
  [\zeta+(n-r)m,+\infty)
  \subset\sigma_{\mathrm{ess}}(H_\mu^{[n]}).
\end{equation}
\end{theorem}

We postpone the proof of the theorem to comment on its consequences. 
\begin{corollary}\label{cor:recursive-essential-threshold}
Set
\begin{equation*}
  E_r:=\inf\sigma(H_\mu^{[r]}),
  \qquad r\geq1,
  \qquad
  E_0:=0,
\end{equation*}
and, for $n\geq1$, define
\begin{equation}\label{eq:recursive-threshold}
  \tau_n
  :=\min_{0\leq r<n}
  \bigl\{E_r+(n-r)m\bigr\}.
\end{equation}
Then
\begin{equation}\label{eq:recursive-essential-inclusion}
  [\tau_n,+\infty)
  \subset\sigma_{\mathrm{ess}}(H_\mu^{[n]}).
\end{equation}
\end{corollary}

\begin{proof}
Theorem~\ref{thm:free-essential-spectrum} implies $E_r\leq rm$ for every
$r\geq1$.  If the minimum in \eqref{eq:recursive-threshold} is attained at
$r=0$, then \eqref{eq:recursive-essential-inclusion} is precisely
Theorem~\ref{thm:free-essential-spectrum}.  If it is attained at some
$r\geq1$ and $E_r<rm$, then $E_r\in\sigma(H_\mu^{[r]})$ and
Theorem~\ref{thm:lower-sector-inclusion} applies with $\zeta=E_r$.  Finally,
if $E_r=rm$, the corresponding value in
\eqref{eq:recursive-threshold} is $nm$, and the free inclusion applies again.
\end{proof}

\begin{remark}\label{rem:lower-sector-HVZ-position}
Theorem~\ref{thm:lower-sector-inclusion}, together with the free essential-spectrum inclusion of
Theorem~\ref{thm:free-essential-spectrum}, gives one inclusion in the expected sector-wise HVZ formula: after minimizing over the proper lower
sectors, one obtains the inclusion stated in
Corollary~\ref{cor:recursive-essential-threshold}, equation \eqref{eq:recursive-essential-inclusion}.  The full formula
also requires the reverse inclusion, namely the exclusion of essential spectrum
below the lowest threshold.  That converse is proved below for the first
two sectors in Corollaries~\ref{cor:essential-spectrum-sector-one}
and~\ref{cor:essential-spectrum-sector-two}.
Moreover, Theorem~\ref{ub-n} proves it for every $n\geq2$ when $\mu_g\geq m$,
or when $\mu_g<m$ and $\pi^2g^2/\sqrt{m-\lambda_1}\leq1$.
In these regimes $\tau_n=\thr_n$, with $\thr_n$ defined in
\eqref{eq:intro-BS-threshold}, and
$\sigma_{\mathrm{ess}}(H_\mu^{[n]})=[\tau_n,+\infty)$.
For $n\geq3$ outside these regimes, the reverse inclusion remains open.
\end{remark}

\subsubsection{Proof of Theorem \ref{thm:lower-sector-inclusion}}
The logic of the proof is to construct, in a fixed sector $n$,
singular Weyl sequences by combining an approximate spectral state of a proper
lower sector $r<n$ with $s:=n-r$ additional bosons that are moved to spatial
infinity.  The reduced operator $M_z^{(n-1)}$ in the $n$-th sector acts on $n-1$
bosonic momentum variables.  In the construction below these variables are
split as
\begin{equation*}
  X=(K,Q),
  \qquad
  K=(k_1,\ldots,k_{r-1})\in\mathbb R^{3(r-1)},
  \qquad
  Q=(q_1,\ldots,q_s)\in\mathbb R^{3s},
\end{equation*}
so that $n-1=(r-1)+s$.  The variables in $K$ are the reduced variables inherited
from the lower sector, whereas the variables in $Q$ will be called the
\emph{spectator variables}, or spectator momenta.  This terminology means that,
in the trial sequence, the corresponding additional bosons are translated far
away from the lower-sector configuration.  The related interaction terms then vanish
asymptotically, while their free energy survives in the limiting energy.  Notice that when
$r=1$, the block $K$ is absent and the lower reduced space is simply
$\mathbb C$.

To simplify the exposition, we first recall the dependence of the
lower-sector reduced operator $M_z^{(r-1)}$ on its spectral parameter $z$, and then prove two auxiliary lemmas.  We determine the
normalization and weak limit of the bosonic symmetrization of a lower-sector
state multiplied by translated spectator packets.  Finally, we show that every
off-diagonal term acting on a spectator variable disappears when the spectators
are sent to spatial infinity.

\begin{remark}\label{r:5.3}
By Remark \ref{r:3.9}, for any $r\in\mathbb N$
\begin{equation*}\label{eq:common-domain-M}
  \dom(M_z^{(r-1)})=\mathcal D_{r-1},
  \qquad z\in\CO\backslash[rm,+\infty),
\end{equation*}
with 
\begin{equation*}
  \mathcal D_{r-1}
  :=
  \left\{
  f\in L_b^2(\mathbb R^{3(r-1)}):
  (1+|K|^2)f(K)\in L_b^2(\mathbb R^{3(r-1)})
  \right\} 
\end{equation*}
and $\mathcal D_0=\mathbb C$.  Furthermore, if $I$ is a compact interval contained in
$(-\infty,rm)$,  there exists $C_I>0$ such that
\begin{equation*}\label{eq:M-local-continuity}
  \|(M_z^{(r-1)}-M_w^{(r-1)})f\|_2
  \leq C_I|z-w|\,\|f\|_2,
  \qquad z,w\in I,
\end{equation*}
for every $f\in\mathcal D_{r-1}$.  In particular, the difference
$M_z^{(r-1)}-M_w^{(r-1)}$ extends to a bounded operator on
$L_b^2(\mathbb R^{3(r-1)})$.
\end{remark}

We next discuss the effect of symmetrizing a product in which one group of
variables is translated in position space.
We use the convention $0!=1$. 
\begin{lemma}\label{lem:asymptotic-symmetrization}
Let $r\in\mathbb N$, $s\in\mathbb N$, and put $n:=r+s$.  Let
$f\in L_b^2(\mathbb R^{3(r-1)})$ and $h\in L_b^2(\mathbb R^{3s})$ satisfy
$\|f\|_2=\|h\|_2=1$.  For $y\in\mathbb R^3$, define
\begin{align*}
  h_y(q_1,\ldots,q_s)
  :&=e^{-iy\cdot(q_1+\cdots+q_s)}h(q_1,\ldots,q_s),\\ 
  \widetilde w_y(K,Q):&=f(K)h_y(Q),
\end{align*}
where $K\in\mathbb R^{3(r-1)}$ and $Q\in\mathbb R^{3s}$.  If $P_{n-1}$ denotes
the orthogonal projection onto $L_b^2(\mathbb R^{3(n-1)})$, then
\begin{equation}\label{eq:symmetrization-norm-limit}
  \|P_{n-1}\widetilde w_y\|_2^2
  \longrightarrow
  \frac{(r-1)!s!}{(n-1)!}
  =\binom{n-1}{r-1}^{-1}
  \qquad (|y|\to\infty).
\end{equation}
Moreover,
\begin{equation}\label{eq:symmetrization-weak-limit}
  \frac{P_{n-1}\widetilde w_y}{\|P_{n-1}\widetilde w_y\|_2}
  \rightharpoonup0
  \qquad (|y|\to\infty).
\end{equation}
\end{lemma}

\begin{proof}
Let $\mathfrak S_{n-1}$ be the permutation group and let $U_\pi$ be its
unitary representation on $L^2(\mathbb R^{3(n-1)})$.  Since $P_{n-1}$ is an
orthogonal projection,
\begin{equation}\label{eq:symmetrization-expansion}
  \|P_{n-1}\widetilde w_y\|_2^2
  =\langle\widetilde w_y,P_{n-1}\widetilde w_y\rangle
  =\frac1{(n-1)!}\sum_{\pi\in\mathfrak S_{n-1}}
  \langle\widetilde w_y,U_\pi\widetilde w_y\rangle.
\end{equation}
The function $\widetilde w_y$ is invariant under permutations of the first
$r-1$ variables and under permutations of the last $s$ variables.  Hence every
permutation in the subgroup $\mathfrak S_{r-1}\times\mathfrak S_s$ contributes
one to the sum in \eqref{eq:symmetrization-expansion}.

Consider now a permutation $\pi$ that does not preserve the two blocks.
The scalar product in \eqref{eq:symmetrization-expansion} is the Fourier
transform, evaluated at a vector whose norm tends to infinity with $|y|$, of
the function
\begin{equation*}
  \overline{f(K)h(Q)}\,(U_\pi(f\otimes h))(K,Q).
\end{equation*}
This function belongs to $L^1(\mathbb R^{3(n-1)})$ by the Cauchy--Schwarz
inequality.  The coefficient of the oscillating phase is nonzero precisely
because $\pi$ moves at least one variable from one block to the other.  The
Riemann--Lebesgue lemma therefore gives
\begin{equation*}
  \langle\widetilde w_y,U_\pi\widetilde w_y\rangle
  \longrightarrow0.
\end{equation*}
Exactly $(r-1)!s!$ permutations preserve the two blocks, and
\eqref{eq:symmetrization-norm-limit} follows.

To prove the weak convergence, let $F\in L^2(\mathbb R^{3(n-1)})$.  For every
$\pi\in\mathfrak S_{n-1}$, the scalar product
$\langle F,U_\pi\widetilde w_y\rangle$ is again the Fourier transform at a
frequency whose norm tends to infinity of an $L^1$ function; here the
$L^1$ property follows from the Cauchy--Schwarz inequality.  Hence
\begin{equation*}
  \langle F,P_{n-1}\widetilde w_y\rangle
  =\frac1{(n-1)!}\sum_{\pi\in\mathfrak S_{n-1}}
  \langle F,U_\pi\widetilde w_y\rangle
  \longrightarrow0.
\end{equation*}
The denominator in \eqref{eq:symmetrization-weak-limit} tends to the
strictly positive number $((r-1)!s!/(n-1)!)^{1/2}$, and the conclusion follows.
\end{proof}

For use in the proof below, we also introduce the operators before restriction to
the symmetric subspace.  For $n\geq2$ and $\lambda<nm$, on the full space
$L^2(\mathbb R^{3(n-1)})$, let
\begin{equation*}
  \widetilde M_\lambda^{(n-1)}
  :=-M_\lambda^{d,(n-1)}+\sum_{j=1}^{n-1} M_{\lambda,j}^{od,(n-1)},
\end{equation*}
where, for $X=(x_1,\ldots,x_{n-1})$, we set
\begin{equation*}
  X_{j,p}:=(x_1,\ldots,x_{j-1},p,x_{j+1},\ldots,x_{n-1})
\end{equation*}
and define
\begin{equation}\label{eq:single-off-diagonal-term}
  (M_{\lambda,j}^{od,(n-1)}F)(X)
  :=g^2\int_{\mathbb R^3}
  \frac{F(X_{j,p})}
  {|p|^2+|X|^2+nm-\lambda}\,dp.
\end{equation}
Thus the integration variable replaces the $j$-th coordinate at its original
position. Permutations merely relabel the summands, so
$\widetilde M_\lambda^{(n-1)}$ commutes with every permutation,
and its restriction to $L_b^2(\mathbb R^{3(n-1)})$ is $M_\lambda^{(n-1)}$.

\begin{lemma}\label{lem:spectator-decoupling}
Let $r\in\mathbb N$, $s\in\mathbb N$, and put $n:=r+s$.  Let
$\lambda<nm$, let $v\in L_b^2(\mathbb R^{3(r-1)})$, and let
$\eta\in C_c^\infty(\mathbb R^{3s})\cap L_b^2(\mathbb R^{3s})$.  Define
\begin{equation*}
  \widetilde w_y(K,Q)
  :=v(K)e^{-iy\cdot(q_1+\cdots+q_s)}\eta(Q).
\end{equation*}
Then, for every $a=1,\ldots,s$,
\begin{equation}\label{eq:spectator-decoupling}
  \|M_{\lambda,r-1+a}^{od,(n-1)}\widetilde w_y\|_2
  \longrightarrow0
  \qquad (|y|\to\infty).
\end{equation}
\end{lemma}

\begin{proof}
Fix $a\in\{1,\ldots,s\}$ and write
$\widehat Q_a=(q_1,\ldots,q_{a-1},q_{a+1},\ldots,q_s)$.  Formula
\eqref{eq:single-off-diagonal-term} gives
\begin{equation*}
\begin{aligned}
 &(M_{\lambda,r-1+a}^{od,(n-1)}\widetilde w_y)(K,Q) \\
 &\quad=g^2v(K)e^{-iy\cdot\sum_{b\neq a}q_b}
 \int_{\mathbb R^3}
 \frac{e^{-iy\cdot p}
 \eta(q_1,\ldots,q_{a-1},p,q_{a+1},\ldots,q_s)}
 {|p|^2+|K|^2+|Q|^2+nm-\lambda}\,dp.
\end{aligned}
\end{equation*}
For almost every $(K,Q)$, the integrand is in $L^1(dp)$ because $\eta$ is
compactly supported and $nm-\lambda>0$.  The Riemann--Lebesgue lemma
therefore gives pointwise convergence to zero.

It remains to find an $L^2$ majorant.  There exist compact sets
$C\subset\mathbb R^{3(s-1)}$ and $B\subset\mathbb R^3$, and a constant
$C_\eta>0$, such that
\begin{equation*}
  |\eta(q_1,\ldots,q_{a-1},p,q_{a+1},\ldots,q_s)|
  \leq C_\eta\mathbf1_C(\widehat Q_a)\mathbf1_B(p).
\end{equation*}
Consequently,
\begin{equation*}
  |M_{\lambda,r-1+a}^{od,(n-1)}\widetilde w_y(K,Q)|
  \leq
  C_{\eta,\lambda,g}
  \frac{|v(K)|\mathbf1_C(\widehat Q_a)}
  {|K|^2+|q_a|^2+nm-\lambda}.
\end{equation*}
The square of the right-hand side is integrable, since
\begin{equation*}
\begin{aligned}
 &\int_{\mathbb R^{3(r-1)}}|v(K)|^2\,dK
 \int_{\mathbb R^3}
 \frac{dq_a}{\bigl(|K|^2+|q_a|^2+nm-\lambda\bigr)^2} \\
 &\quad=\pi^2\int_{\mathbb R^{3(r-1)}}
 \frac{|v(K)|^2}
 {\sqrt{|K|^2+nm-\lambda}}\,dK \\
 &\quad\leq
 \frac{\pi^2}{\sqrt{nm-\lambda}}\|v\|_2^2.
\end{aligned}
\end{equation*}
The remaining variables range in the compact set $C$.  Dominated convergence
proves \eqref{eq:spectator-decoupling}.
\end{proof}

\begin{proof}[Proof of Theorem \ref{thm:lower-sector-inclusion}]
Put
\begin{equation*}\label{eq:cluster-index-decomposition}
  s:=n-r,
  \qquad
  n-1=(r-1)+s.
\end{equation*}
It is enough, in view of Theorem~\ref{thm:free-essential-spectrum}, to prove
that
\begin{equation*}
  [\zeta+sm,nm)
  \subset\sigma_{\mathrm{ess}}(H_\mu^{[n]}).
\end{equation*}
Fix $\lambda\in[\zeta+sm,nm)$ and write
\begin{equation*}\label{eq:rho-definition}
  \lambda=\zeta+sm+\rho,
  \qquad
  0\leq\rho<rm-\zeta.
\end{equation*}

Since $\zeta<rm$, Corollary~\ref{appcert:cor:M-reduction}\textup{(ii)} and
\eqref{eq:lower-sector-spectral-point} imply
\begin{equation*}
  0\in\sigma(M_\zeta^{(r-1)}).
\end{equation*}
The operator $M_\zeta^{(r-1)}$ is self-adjoint, so Weyl's criterion provides a
sequence $\{v_\ell\}\subset\dom(M_\zeta^{(r-1)})$ such that
\begin{equation*}\label{eq:lower-sector-weyl-sequence}
  \|v_\ell\|_2=1,
  \qquad
  \|M_\zeta^{(r-1)}v_\ell\|_2\longrightarrow0.
\end{equation*}
No weak convergence of $v_\ell$ is required.

Choose $q_\rho\in\mathbb R^3$ such that
\begin{equation*}
  s|q_\rho|^2=\rho.
\end{equation*}
Let $\chi\in C_c^\infty(\mathbb R^3)$ satisfy $\|\chi\|_2=1$, and define
\begin{equation*}
  \chi_\ell(q):=\ell^{3/2}\chi\bigl(\ell(q-q_\rho)\bigr),
\end{equation*}
\begin{equation*}\label{eq:spectator-packet}
  \eta_\ell(q_1,\ldots,q_s)
  :=\prod_{a=1}^s\chi_\ell(q_a).
\end{equation*}
Then $\eta_\ell\in C_c^\infty(\mathbb R^{3s})\cap
L_b^2(\mathbb R^{3s})$, $\|\eta_\ell\|_2=1$, and
\begin{equation*}\label{eq:spectator-shell}
  \delta_\ell
  :=\sup_{Q\in\supp\eta_\ell}
  \left||Q|^2-\rho\right|
  \longrightarrow0.
\end{equation*}
For $y\in\mathbb R^3$, set
\begin{equation*}
  \eta_{\ell,y}(Q)
  :=e^{-iy\cdot(q_1+\cdots+q_s)}\eta_\ell(Q),
\end{equation*}
\begin{equation*}\label{eq:unsymmetrized-cluster-vector}
  \widetilde w_{\ell,y}(K,Q)
  :=v_\ell(K)\eta_{\ell,y}(Q).
\end{equation*}
Since $v_\ell\in\mathcal D_{r-1}$ and $\eta_\ell$ is smooth and compactly
supported, $\widetilde w_{\ell,y}$ belongs to the natural domain
\begin{equation*}
  \widetilde{\mathcal D}_{n-1}
  :=\left\{F\in L^2(\mathbb R^{3(n-1)}):(1+|X|^2)F\in L^2(\mathbb R^{3(n-1)})\right\}
\end{equation*}
of $\widetilde M_\lambda^{(n-1)}$.  The projection $P_{n-1}$ maps this
domain into $\mathcal D_{n-1}$.

We first estimate the terms acting on the first $r-1$ variables.  For
$Q\in\mathbb R^{3s}$ define
\begin{equation*}\label{eq:shifted-spectral-parameter}
  \lambda_Q
  :=\lambda-sm-|Q|^2
  =\zeta-\bigl(|Q|^2-\rho\bigr).
\end{equation*}
A direct comparison of \eqref{appcert:eq:diagonal-sector-definition} and
\eqref{eq:single-off-diagonal-term}, with the positions of the $K$ and $Q$
variables kept fixed, gives the identity
\begin{equation}\label{eq:cluster-reconstruction-general}
\begin{aligned}
 &\left(-M_\lambda^{d,(n-1)}+
 \sum_{j=1}^{r-1}M_{\lambda,j}^{od,(n-1)}\right)
 \widetilde w_{\ell,y}(K,Q) \\
 &\quad=\eta_{\ell,y}(Q)
 \bigl(M_{\lambda_Q}^{(r-1)}v_\ell\bigr)(K).
\end{aligned}
\end{equation}
For $\ell$ sufficiently large, all $\lambda_Q$ with
$Q\in\supp\eta_\ell$ belong to a fixed compact interval
$I\Subset(-\infty,rm)$.  By Remark \ref{r:5.3}
\begin{equation*}
\begin{aligned}
  \|M_{\lambda_Q}^{(r-1)}v_\ell\|_2
  &\leq
  \|M_\zeta^{(r-1)}v_\ell\|_2
  +\|(M_{\lambda_Q}^{(r-1)}-M_\zeta^{(r-1)})v_\ell\|_2 \\
  &\leq
  \|M_\zeta^{(r-1)}v_\ell\|_2+C_I\delta_\ell.
\end{aligned}
\end{equation*}
Using \eqref{eq:cluster-reconstruction-general} and $\|\eta_\ell\|_2=1$,
one obtains
\begin{equation}\label{eq:cluster-reconstruction-error}
  \sup_{y\in\mathbb R^3}
  \left\|
  \left(-M_\lambda^{d,(n-1)}+
  \sum_{j=1}^{r-1}M_{\lambda,j}^{od,(n-1)}\right)
  \widetilde w_{\ell,y}
  \right\|_2
  \longrightarrow0.
\end{equation}
The identity is first verified on smooth compactly supported functions.  It
extends to $v_\ell\in\mathcal D_{r-1}$ by approximation in the weighted norm
$\|(1+|K|^2)\,\cdot\,\|_2$. The weighted estimates for the off-diagonal terms
apply also before symmetrization, and the remaining factors are smooth and
compactly supported. The dependence on $\lambda_Q$ is uniformly controlled
by Remark~\ref{r:5.3}.

For each fixed $\ell$, Lemma~\ref{lem:spectator-decoupling} gives
\begin{equation}\label{eq:spectator-error-fixed-ell}
  R_\ell(y)
  :=\sum_{a=1}^s
  \|M_{\lambda,r-1+a}^{od,(n-1)}\widetilde w_{\ell,y}\|_2
  \longrightarrow0
  \qquad (|y|\to\infty).
\end{equation}
Moreover, Lemma~\ref{lem:asymptotic-symmetrization} gives, for each fixed
$\ell$,
\begin{equation}\label{eq:symmetrized-product-limit}
  \|P_{n-1}\widetilde w_{\ell,y}\|_2^2
  \longrightarrow
  \frac{(r-1)!s!}{(n-1)!}>0,
\end{equation}
and
\begin{equation}\label{eq:normalized-product-weak}
  \frac{P_{n-1}\widetilde w_{\ell,y}}
  {\|P_{n-1}\widetilde w_{\ell,y}\|_2}
  \rightharpoonup0
  \qquad (|y|\to\infty).
\end{equation}

We now choose one sequence from the two-parameter family.  Let
$\{\varphi_j\}_{j\geq1}$ be a countable dense subset of
$L_b^2(\mathbb R^{3(n-1)})$.  For every $\ell$, choose $y_\ell\in\mathbb R^3$
so large that
\begin{equation*}\label{eq:diagonal-choice-size}
  |y_\ell|\geq\ell,
  \qquad
  \|P_{n-1}\widetilde w_{\ell,y_\ell}\|_2
  \geq\frac12\left(\frac{(r-1)!s!}{(n-1)!}\right)^{1/2},
\end{equation*}
\begin{equation}\label{eq:diagonal-choice-spectator}
  R_\ell(y_\ell)\leq\frac1\ell,
\end{equation}
and
\begin{equation}\label{eq:diagonal-choice-weak}
  \left|
  \left\langle
  \varphi_j,
  \frac{P_{n-1}\widetilde w_{\ell,y_\ell}}
  {\|P_{n-1}\widetilde w_{\ell,y_\ell}\|_2}
  \right\rangle
  \right|
  \leq\frac1\ell,
  \qquad j=1,\ldots,\ell.
\end{equation}
This is possible by \eqref{eq:spectator-error-fixed-ell},
\eqref{eq:symmetrized-product-limit}, and
\eqref{eq:normalized-product-weak}.  Define
\begin{equation}\label{eq:reduced-weyl-sequence}
  w_\ell
  :=\frac{P_{n-1}\widetilde w_{\ell,y_\ell}}
  {\|P_{n-1}\widetilde w_{\ell,y_\ell}\|_2}.
\end{equation}
Then $\|w_\ell\|_2=1$.  For each fixed $j$,
\eqref{eq:diagonal-choice-weak} gives
$\langle\varphi_j,w_\ell\rangle\to0$.  Since the sequence is bounded and
$\{\varphi_j\}$ is dense,
\begin{equation}\label{eq:reduced-weyl-weak}
  w_\ell\rightharpoonup0.
\end{equation}

The operator $\widetilde M_\lambda^{(n-1)}$ commutes with $P_{n-1}$, and the
restriction of $\widetilde M_\lambda^{(n-1)}$ to the symmetric subspace is
$M_\lambda^{(n-1)}$.  Therefore,
\begin{equation*}
\begin{aligned}
  \|M_\lambda^{(n-1)}w_\ell\|_2
  &\leq
  2\left(\frac{(n-1)!}{(r-1)!s!}\right)^{1/2}
  \|\widetilde M_\lambda^{(n-1)}
  \widetilde w_{\ell,y_\ell}\|_2 \\
  &\leq
  2\left(\frac{(n-1)!}{(r-1)!s!}\right)^{1/2}
  \left(
  \left\|
  \left(-M_\lambda^{d,(n-1)}+
  \sum_{j=1}^{r-1}M_{\lambda,j}^{od,(n-1)}\right)
  \widetilde w_{\ell,y_\ell}
  \right\|_2
  +R_\ell(y_\ell)
  \right).
\end{aligned}
\end{equation*}
The first term tends to zero by
\eqref{eq:cluster-reconstruction-error}, and the second one tends to zero by
\eqref{eq:diagonal-choice-spectator}.  Hence
\begin{equation}\label{eq:reduced-weyl-error}
  M_\lambda^{(n-1)}w_\ell\longrightarrow0.
\end{equation}
Equations \eqref{eq:reduced-weyl-sequence},
\eqref{eq:reduced-weyl-weak}, and
\eqref{eq:reduced-weyl-error} show that $\{w_\ell\}$ is a singular Weyl
sequence for $M_\lambda^{(n-1)}$ at zero.  Corollary
\ref{appcert:cor:M-reduction}\textup{(iv)} then gives
\begin{equation*}
  \lambda\in\sigma_{\mathrm{ess}}(H_\mu^{[n]}).
\end{equation*}
This proves the inclusion up to $nm$.  Theorem
\ref{thm:free-essential-spectrum} supplies the interval $[nm,+\infty)$, and
\eqref{eq:lower-sector-essential-inclusion} follows.
\end{proof}

\subsection{Consequences for the first two sectors}
The first two sectors allow more precise results and may generate lower
essential thresholds in higher sectors. The existence of eigenvalues below
threshold in these sectors is studied in Section~\ref{sec:point-discrete-spectrum}.
\begin{corollary}\label{cor:essential-spectrum-sector-one}
One has
\begin{equation}\label{eq:essential-spectrum-sector-one}
  \sigma_{\mathrm{ess}}(H_\mu^{[1]})=[m,+\infty).
\end{equation}
If $\mu_g<m$, then for every $n\geq2$,
\begin{equation}\label{eq:one-sector-branch}
  [\lambda_1+(n-1)m,+\infty)
  \subset\sigma_{\mathrm{ess}}(H_\mu^{[n]}).
\end{equation}
\end{corollary}
\begin{proof}
The inclusion $[m,+\infty)\subset\sigma_{\mathrm{ess}}(H_\mu^{[1]})$
follows from Theorem~\ref{thm:free-essential-spectrum}.  For every
$\lambda<m$, the reduced operator $M_\lambda^{(0)}$ acts on the
one-dimensional space $\CO$.  Corollary~\ref{appcert:cor:M-reduction}\textup{(ii)} excludes
essential spectrum below $m$, proving
\eqref{eq:essential-spectrum-sector-one}.  The description of the spectrum
below $m$ is given by Theorem~\ref{thm:n1} in
Section~\ref{sec:point-discrete-spectrum}.  If $\mu_g<m$, the inclusion
\eqref{eq:one-sector-branch} follows from
Theorem~\ref{thm:lower-sector-inclusion} with $r=1$ and
$\zeta=\lambda_1$, together with the free inclusion.
\end{proof}

\begin{corollary}\label{cor:essential-spectrum-sector-two}
Let
\begin{equation*}\label{eq:two-sector-threshold}
  T_2:=m+E_1
  =
  \begin{cases}
    2m,&\mu_g\geq m,\\
    m+\lambda_1&\mu_g<m.
  \end{cases}
\end{equation*}
Then
\begin{equation}\label{eq:essential-spectrum-sector-two}
  \sigma_{\mathrm{ess}}(H_\mu^{[2]})=[T_2,+\infty).
\end{equation}
In particular, every spectral point of $H_\mu^{[2]}$ below $T_2$ is a
discrete eigenvalue of finite multiplicity.
\end{corollary}

\begin{proof}
The inclusion
\begin{equation*}
  [T_2,+\infty)\subset\sigma_{\mathrm{ess}}(H_\mu^{[2]})
\end{equation*}
follows from Theorem~\ref{thm:free-essential-spectrum} and, when
$\mu_g<m$, from \eqref{eq:one-sector-branch}.

Fix $\lambda<T_2$.  Then $\lambda<2m$ and the diagonal multiplier
$m_\lambda^{(1)}$ is strictly positive.  More precisely, there is
$c_\lambda>0$ such that
\begin{equation}\label{eq:two-sector-diagonal-lower-bound}
  m_\lambda^{(1)}(k)\geq c_\lambda(1+|k|^2),
  \qquad k\in\mathbb R^3.
\end{equation}
For $\mu_g\geq m$, this follows directly from the formula for
$m_\lambda^{(1)}$.  If $\mu_g<m$, put
$a=\sqrt{m-\lambda_1}$ and $b=\pi^2g^2$.  With
$x=\sqrt{|k|^2+2m-\lambda}$ one has
\begin{equation*}
  m_\lambda^{(1)}(k)
  =x^2+2bx+\mu_g-m
  =(x-a)(x+a+2b).
\end{equation*}
Since $\lambda<m+\lambda_1$, one has
$x\geq\sqrt{2m-\lambda}>a$, and
\eqref{eq:two-sector-diagonal-lower-bound} follows.

We now show that the off-diagonal term is relatively compact with respect
to $M_\lambda^{d,(1)}$.  The operator
\begin{equation*}
  M_\lambda^{od,(1)}(M_\lambda^{d,(1)}-i)^{-1}
\end{equation*}
has integral kernel
\begin{equation*}\label{eq:two-sector-relative-compact-kernel}
  \mathcal K_\lambda(k,p)
  =
  \frac{g^2}
  {\bigl(|k|^2+|p|^2+2m-\lambda\bigr)
   \bigl(m_\lambda^{(1)}(p)-i\bigr)}.
\end{equation*}
Since $m_\lambda^{(1)}(p)$ is real and has quadratic growth, there is a
constant $C_\lambda>0$ such that
\begin{equation*}
  \frac1{|m_\lambda^{(1)}(p)-i|}
  \leq\frac{C_\lambda}{1+|p|^2}.
\end{equation*}
Consequently,
\begin{equation*}
  |\mathcal K_\lambda(k,p)|
  \leq
  \frac{C_\lambda}
  {(1+|p|^2)(1+|k|^2+|p|^2)}.
\end{equation*}
The square of the right-hand side is integrable on $\mathbb R^6$.  Indeed,
\begin{equation*}
\begin{aligned}
 &\int_{\mathbb R^3}\int_{\mathbb R^3}
 \frac{dk\,dp}
 {(1+|p|^2)^2(1+|k|^2+|p|^2)^2} \\
 &\quad=\pi^2\int_{\mathbb R^3}
 (1+|p|^2)^{-5/2}\,dp<\infty.
\end{aligned}
\end{equation*}
Thus $M_\lambda^{od,(1)}(M_\lambda^{d,(1)}-i)^{-1}$ is
Hilbert--Schmidt, and $M_\lambda^{od,(1)}$ is relatively compact with
respect to $M_\lambda^{d,(1)}$.
By Weyl's theorem,
\begin{equation*}
  \sigma_{\mathrm{ess}}(M_\lambda^{(1)})
  =\sigma_{\mathrm{ess}}(-M_\lambda^{d,(1)}).
\end{equation*}
The lower bound \eqref{eq:two-sector-diagonal-lower-bound} gives
$0\notin\sigma_{\mathrm{ess}}(-M_\lambda^{d,(1)})$, and hence
\begin{equation*}
  0\notin\sigma_{\mathrm{ess}}(M_\lambda^{(1)}).
\end{equation*}
Corollary~\ref{appcert:cor:M-reduction}\textup{(ii)} therefore implies
$
  \lambda\notin\sigma_{\mathrm{ess}}(H_\mu^{[2]}).
$
This proves the reverse inclusion in
\eqref{eq:essential-spectrum-sector-two}.

Finally, a spectral point below the essential spectrum of a self-adjoint operator
is an isolated eigenvalue of finite multiplicity.  This proves the last assertion.
\end{proof}
Theorem~\ref{n=2eigenvalues} proves, under an explicit coupling condition,
the existence of an eigenvalue $\lambda_2<T_2=m+\lambda_1$. Thus
Theorem~\ref{thm:lower-sector-inclusion}, applied with $r=2$, yields
$[\lambda_2+(n-2)m,+\infty)\subset\sigma_{\mathrm{ess}}(H_\mu^{[n]})$
for every $n\geq3$, with an endpoint strictly below $\lambda_1+(n-1)m$.
%

\section{Point and discrete spectrum}
\label{sec:point-discrete-spectrum}

We first give the point spectrum in the zero- and one-excitation sectors,
where the reduced spaces are finite-dimensional and the spectral equations can
be solved explicitly.  The Birman--Schwinger analysis for the sectors $n\geq2$
is developed afterwards.

\subsection{The vacuum sector}
\label{subsec:vacuum-sector} This is a trivial case. By definition, $\IS^{[0]}=\CO\oplus0$ and $H_\mu^{[0]}$ is the zero
operator on this one-dimensional space. Hence,
\begin{equation*}
  \sigma(H_\mu^{[0]})=\{0\}.
\end{equation*}
The eigenvalue $0$ is simple, with normalized eigenvector $1\oplus0$.
Consequently, $0$ is always an eigenvalue of the full Hamiltonian $H_\mu$,
with eigenvector $\Omega^\circ\oplus0$.

\subsection{The sector \texorpdfstring{$n=1$}{n=1}}
\label{subsec:one-excitation-sector}

The one-excitation sector is the first nontrivial sector of the model.  Its
reduced space is one-dimensional, and hence the reduced operator is a scalar.
This makes it possible to determine all spectral points below the free threshold
$m$ explicitly. As recalled in the introduction, all the results in this subsection are well known (see \cite{FLL21} and references therein), but they are collected here both for completeness and for future reference. 

\begin{theorem}\label{thm:n1}
The following assertions hold.
\begin{enumerate}[label=\textup{(\roman*)}]
\item If $\mu_g\geq m$, then
\begin{equation*}
  \sigma(H_\mu^{[1]})\cap(-\infty,m)=\varnothing.
\end{equation*}
\item If $\mu_g<m$, then
\begin{equation*}
  \sigma(H_\mu^{[1]})\cap(-\infty,m)=\{\lambda_1\}.
\end{equation*}
The eigenvalue $\lambda_1$ is simple and discrete, and its eigenspace is
spanned by
\begin{equation*}
  G_{\lambda_1}^{\circ,(1)}1\oplus1.
\end{equation*}
If $\nu=\mu$, then $\lambda_1=\mu$.
\end{enumerate}
\end{theorem}

\begin{proof}
For $n=1$, the lower space is $L_b^2(\RE^0)=\CO$.  By
\eqref{eq:M-scalar-sector-zero}, for every $\lambda<m$ the reduced operator
is multiplication by the scalar
\begin{equation*}
  M_\lambda^{(0)}
  =\lambda-\mu_g-2\pi^2g^2\sqrt{m-\lambda}.
\end{equation*}
By Corollary~\ref{appcert:cor:M-reduction}, a number $\lambda<m$ belongs
to $\sigma(H_\mu^{[1]})$ if and only if $M_\lambda^{(0)}=0$.

Set
\begin{equation*}
  x:=\sqrt{m-\lambda}>0.
\end{equation*}
Then $M_\lambda^{(0)}=0$ is equivalent to
\begin{equation}\label{eq:n1-scalar}
  x^2+2\pi^2g^2x+\mu_g-m=0.
\end{equation}
If $\mu_g\geq m$, equation \eqref{eq:n1-scalar} has no positive solution.
When $\mu_g=m$, its nonnegative root $x=0$ corresponds to the threshold
$\lambda=m$ and not to an isolated  spectral point.

If $\mu_g<m$, equation \eqref{eq:n1-scalar} has the unique positive root
\begin{equation*}
  x_1=\sqrt{\pi^4g^4+m-\mu_g}-\pi^2g^2
  =\sqrt{m - \lambda_1}.
\end{equation*}
It follows that
\begin{equation*}
  \lambda_1=m-x_1^2.
\end{equation*}
The complete spectral equivalence in
Corollary~\ref{appcert:cor:M-reduction} shows that there are no other
spectral points below $m$.  Since $M_{\lambda_1}^{(0)}$ acts on a
one-dimensional space, its zero eigenvalue is simple and discrete.  Then Corollary~\ref{appcert:cor:M-reduction}$\textrm{\ref{item:spcor2}}$ gives the
stated eigenvector of $H_\mu^{[1]}$.

Finally, assume $\nu=\mu$.  Since $\nu<m$, one has $\mu<m$, and
$x=\sqrt{m-\mu}$ solves \eqref{eq:n1-scalar}.  By uniqueness of the positive
root, $x_1=\sqrt{m-\mu}$ and hence $\lambda_1=\mu$.
\end{proof}

\begin{remark}
When $\nu=\mu$, Theorem~\ref{thm:n1} reproduces the corresponding
one-excitation result in \cite{Raj}.
\end{remark}

\subsubsection{Resonances in the one-excitation sector}

For completeness, we also record the poles obtained by continuing the
one-sector resolvent through the cut $[m,+\infty)$.  They are most naturally
described on the two-sheeted Riemann surface of the function
$z\mapsto\sqrt{m-z}$.

\begin{definition}[Resonant pole, cf.\ \cite{RSIV}]\label{d:respole}
Let $H_0$ and $H$ be self-adjoint operators in a Hilbert space $\mathscr H$.
Suppose that there is a dense set $\mathscr D\subset\mathscr H$ such that,
for every $\psi\in\mathscr D$, the scalar resolvent matrix elements
\begin{equation*}
  z\longmapsto\langle\psi,(H-z)^{-1}\psi\rangle,
  \qquad
  z\longmapsto\langle\psi,(H_0-z)^{-1}\psi\rangle,
\end{equation*}
have analytic continuations from the upper half-plane across a portion of
the continuous spectrum to a second sheet.  If the continued free matrix
element is analytic at $z_0$, while for some $\psi\in\mathscr D$ the
continued interacting matrix element has a pole at $z_0$, then $z_0$ is
called a resonant pole of $H$ relative to $H_0$.
\end{definition}

\begin{proposition}\label{prop:resonances-sector-one}
Assume $g\neq0$ and put $b:=\pi^2g^2$.  The poles of the meromorphic
continuation of $(M_z^{(0)})^{-1}$ are determined by
\begin{equation}\label{eq:one-sector-resonance-roots}
  X_\pm=-b\pm\sqrt{b^2+m-\mu_g},
  \qquad
  z_\pm=m-X_\pm^2.
\end{equation}
More precisely:
\begin{enumerate}[label=\textup{(\roman*)}]
\item If $\mu_g>m+b^2$, there are two nonreal resonant poles
\begin{equation*}
  \zeta_\pm
  =\mu_g-2b^2
  \pm2ib\sqrt{\mu_g-m-b^2}.
\end{equation*}
\item If $m<\mu_g<m+b^2$, there are two distinct real resonant poles below
$m$,
\begin{equation*}
  z_\pm
  =\mu_g-2b^2
  \pm2b\sqrt{b^2+m-\mu_g}.
\end{equation*}
If $\mu_g=m+b^2$, the two roots coalesce and give a real pole of
multiplicity two at $m-b^2$.
\item If $\mu_g=m$, one root lies on the second sheet and gives the real
resonant pole $m-4b^2$, while the other root reaches the threshold $m$.
The latter is a threshold singularity and is not an isolated eigenvalue.
\item If $\mu_g<m$, the root $X_+$ lies on the physical sheet and gives the
eigenvalue $z_+=\lambda_1$, whereas $X_-$ lies on the second sheet and gives
the real resonant pole
\begin{equation*}
  z_-
  =\mu_g-2b^2-2b\sqrt{b^2+m-\mu_g}<m.
\end{equation*}
\end{enumerate}
\end{proposition}

\begin{proof}
In Definition~\ref{d:respole} we take
$H_0=H^{\circ,(1)}\oplus\mu$ on $L^2(\RE^3)\oplus\CO$ and, for example,
the dense set
\begin{equation*}
  \mathscr D
  =\operatorname{span}\{P(k)e^{-|k|^2/2}:P\text{ is a polynomial}\}
  \oplus\CO.
\end{equation*}
The physical sheet is determined by the branch
$\operatorname{Re}\sqrt{m-z}>0$; the second sheet is obtained by changing
the sign of this square root.
The scalar function controlling the one-sector resolvent is
\begin{equation*}
  M_z^{(0)}=z-\mu_g-2b\sqrt{m-z}.
\end{equation*}
The free resolvent matrix elements admit analytic continuation through
$[m,+\infty)$ on the indicated dense set, and the same is true of the
matrix elements of $G_z^{\circ,(1)}$ and $G_{\bar z}^{\circ,(1)*}$
against these test vectors.  By the resolvent formula
\eqref{appcert:eq:krein-M-resolvent}, the additional poles are therefore the
zeros of the continued scalar function $M_z^{(0)}$.  There is no cancellation:
for instance, the lower--lower matrix element of the resolvent contains
$(M_z^{(0)})^{-1}$ directly.

Set $X=\sqrt{m-z}$.  On the physical sheet, $X>0$ for real $z<m$; after
continuation to the second sheet the sign is reversed, so that $X<0$ for real
$z<m$.  Since $z=m-X^2$, the equation $M_z^{(0)}=0$ becomes
\begin{equation*}
  X^2+2bX+\mu_g-m=0,
\end{equation*}
whose roots are exactly those in
\eqref{eq:one-sector-resonance-roots}.  If $\mu_g>m+b^2$, the square root in
\eqref{eq:one-sector-resonance-roots} is purely imaginary and both roots have
negative real part, which gives the two poles in item~\textup{(i)}.  If
$m<\mu_g<m+b^2$, the square root is real and strictly smaller than $b$;
hence both roots are negative and lie on the second sheet, giving
item~\textup{(ii)}.  At equality the quadratic has a double root $X=-b$.
If $\mu_g=m$, the roots are $0$ and $-2b$, which gives item~\textup{(iii)}.
The threshold root $X=0$ is not an eigenvalue: by
\eqref{appcert:eq:G-lambda-explicit}, the formal lifted component at
$\lambda=m$ is proportional to $|k|^{-2}$ and therefore does not
belong to $L^2(\RE^3)$ near $k=0$.
Finally, if $\mu_g<m$, one has $X_+>0$ and $X_-<0$.  The first root is the
physical-sheet eigenvalue already found in Theorem~\ref{thm:n1}, while the
second gives item~\textup{(iv)}.
\end{proof}

\subsection[The sectors n at least 2]{The sectors \texorpdfstring{$n\geq2$}{n at least 2}}
\label{subsec:point-discrete-general-sectors}
We take $n\ge 2$ and use the notation
\begin{equation*}
L_{b}^{2,\alpha}(\RE^{3(n-1)})
\equiv L_{b}^{2}\bigl(\RE^{3(n-1)},(|K|^{2}+1)^{\alpha}dK\bigr).
\end{equation*}
We notice that for fixed $n$, the identifications
\begin{equation*}
 \mathcal D_{n-1}=\H_b^{(n-1)}
 =L_b^{2,2}(\RE^{3(n-1)})
\end{equation*}
hold as sets, with equivalent norms.
We denote by $\|\cdot\|_{\alpha}$ the norm in
$L_{b}^{2,\alpha}(\RE^{3(n-1)})$ and set
$\|\cdot\|\equiv\|\cdot\|_{0}$. For brevity we omit the index $(n-1)$
from functions and operators acting in these spaces:
\begin{equation*}
M^{d}_{\lambda}:=M_{\lambda}^{d,(n-1)}
\,,\qquad
    M_{\lambda}^{od}:=M_{\lambda}^{od,(n-1)}
    \,.
\end{equation*}
More explicitly,
$M^d_\lambda$ is the multiplication by the function 
\begin{equation*}
{m}_\lambda(K)\equiv{m}_\lambda(|K|)=p_{2}\big(\sqrt{|K|^{2}+nm-\lambda}\,\big)\,,\qquad
p_{2}(x):=x^{2}+2\pi^{2}g^{2}x+\mu_{g}-m\,.
\end{equation*}
Subsequently, in the case  $\mu_{g}< m$, we use the following convenient writing of the two roots $x_{\pm}$ of $p_{2}$:
\be\label{roots}
x_{+}=\sqrt{m-\lambda_{1}}\,,\qquad x_{-}=-\big(2\pi^{2}g^{2}+\sqrt{m-\lambda_{1}}\,\big)\,. 
\ee
where $\lambda_1$ denotes the simple eigenvalue for the $n=1$ case given in Eq. \eqref{eq:lambda-one-definition}.
One has
\begin{equation}\label{1}
|K|^{2}+p_{2}(\sqrt{nm-\lambda}\,)\le{m}_\lambda(K)\le (1+2\pi^{2}g^{2})(|K|^{2}+nm-\lambda) +\pi^{2}g^{2}+\mu_{g}-m\,.
\end{equation}

We define 
\begin{equation*}\label{eq:intro-two-threshold2}
  \thr_n:=
  \begin{cases}
    nm,&\mu_g\geq m,\\
    (n-1)m+\lambda_1&\mu_g<m,
  \end{cases}\qquad n\geq 1. 
\end{equation*}

Within this section we always assume $\lambda < \thr_n$, which implies $\lambda<nm$. 

\begin{remark}
For $\mu_g<m$, the number $\thr_n$ is the threshold generated by the
one-excitation sector; for $\mu_g\geq m$, it is the free threshold.  If
$n\geq3$, one cannot in general claim that $\thr_n$ is the bottom of the full essential
spectrum.  Indeed, if $\zeta\in\sigma(H_\mu^{[2]})$ and $\zeta<2m$, then
Theorem \ref{thm:lower-sector-inclusion} for $r=2$ 
gives
$
  [\zeta+(n-2)m,+\infty)
  \subset\sigma_{\mathrm{ess}}(H_\mu^{[n]}),
$
and its left endpoint may lie below $\thr_n$. In the absence regimes of
Theorem~\ref{ub-n}, however, $\thr_n$ is the bottom of the essential spectrum.
\end{remark}
There holds 
\begin{equation}\label{2}
p_{2}(\sqrt{nm-\lambda}\,)\ge 
\begin{cases}
\thr_n-\lambda+\mu_{g}-m &\text{if $\mu_{g}\ge m$}\\
\thr_n-\lambda&\text{if $\mu_{g}< m$}
\end{cases}
\end{equation}
For $\mu_{g}\ge m$, the lower bound \eqref{2} follows immediately from $p_{2}(x)\geq x^{2}+\mu_{g}-m$, for $x>0$; for $\mu_{g}< m$ it follows from the identity 
\begin{equation*}
p_{2}(\sqrt{nm -\lambda}) = \thr_n -\lambda + 2\pi^2g^2\big(\sqrt{\thr_n-\lambda +m-\lambda_1} - \sqrt{m-\lambda_1}\big) \qquad \mu_{g}< m.
\end{equation*}
By the lower bounds \eqref{1} and \eqref{2} there follows 
\begin{equation*}
{m}_\lambda(K) \geq |K|^{2}+p_{2}(\sqrt{nm-\lambda}\,)
\ge 
\begin{cases}
|K|^2+\thr_n-\lambda+\mu_{g}-m &\text{if $\mu_{g}\ge m$}\\
|K|^2+\thr_n-\lambda&\text{if $\mu_{g}< m$}
\end{cases}
\end{equation*}
We write the latter inequality as 
\begin{equation}\label{mK_lowerbound}
{m}_\lambda(K) \geq |K|^2+\thr_n-\lambda+\delta_g
\end{equation}
with 
\begin{equation}\label{deltag}
\delta_g :=
  \begin{cases}
 \mu_g - m&\mu_g\geq m,\\
    0 &\mu_g<m,
  \end{cases}. 
\end{equation}

\begin{lemma}\label{BS} 
Let $\lambda<nm$ and $\beta>1$. If
\begin{equation*}
  \frac{\beta-2}{2}\leq\alpha\leq\frac{\beta+2}{2},
\end{equation*}
then
\begin{equation*}
  M^{od}_{\lambda}\in
  \B\bigl(L^{2,\alpha}_{b}(\RE^{3(n-1)}),
  L^{2,\alpha-\beta}_{b}(\RE^{3(n-1)})\bigr).
\end{equation*}
\end{lemma}
\begin{proof} Set
$4\theta=\beta-2\alpha+2$. The restriction on $\alpha$ is equivalent to
$0\leq\theta\leq1$. By the Cauchy--Schwarz inequality, 
\begin{align*}
&|\langle \phi,M^{od}_{\lambda}\varphi\rangle_{\alpha-\beta}|^{2}
\leq C_\lambda \left(\int_{\RE^{3n}}  \sum_{j=1}^{n-1}
\frac{  |\phi(K)|\,|\varphi(\hat K_{j,p})|\,(\,|K|^{2}+1)^{\alpha-\beta}}{ |p|^2+
| K|^{2}+1  }\, dpdK
  \right)^{2}
  \\
\leq &C_\lambda\left(\sum_{j=1}^{n-1}\int_{\RE^{3n}}  
\frac{  |\phi(K)|^{2}\, \,(|K|^{2}+1)^{\alpha-\beta}(|p|^{2}+|\hat K_{j}|^{2}+1)^{-\alpha}}{ (\,|p|^2+|K|^{2}+1)^{2\theta}  }\,dpdK  \right)\times
\\
&\times\left(\sum_{j=1}^{n-1}\int_{\RE^{3n}}  
\frac{  |\varphi(\hat K_{j,p})|^{2}(|\hat K_{j,p}|^{2}+1)^{\alpha}(|k_{j}|^{2}+|\hat K_{j}|^{2}+1)^{\alpha-\beta}}{(|p|^{2}+|k_{j}|^2+|\hat K_{j}|^{2}+1)^{2(1-\theta)}}
  \, dpdK\right)
  \\
\leq &C_\lambda\|\phi\|_{\alpha-\beta}^{2}\|\varphi\|^{2}_{\alpha}\int_{0}^{+\infty}  
\frac{  \varrho^{2} d\varrho}{ (\varrho^{2}+1)^{2\theta+\alpha}}
 \int_{0}^{+\infty}    
\frac{ r^{2} dr}{ (r^2+1)^{2(1-\theta)+\beta-\alpha}}  
  \,.
\end{align*}
With the above choice of $\theta$, both exponents in the last
line equal $1+\beta/2>3/2$, and hence both radial integrals converge.
\end{proof}

By \eqref{1} and \eqref{mK_lowerbound} there follows that for any $\lambda <\thr_n$ there exists $C_\lambda >0$ such that 
\begin{equation*}
\frac{1}{C_\lambda}(|K|^{2}+1)\le{m}_\lambda(K)\le C_\lambda(|K|^{2}+1).
\end{equation*}
Hence, for every $\alpha,q\in\RE$ and $\lambda<\thr_n$,
\begin{equation}\label{Mlambda-boundedness}
  (M_{\lambda}^{d})^{-q/2}
  \in\B\bigl(L_{b}^{2,\alpha}(\RE^{3(n-1)}),
  L_{b}^{2,\alpha+q}(\RE^{3(n-1)})\bigr).
\end{equation}

Let $\gamma\in(0,1)$, set $\beta=2-\gamma$, and assume that
$0\leq s\leq2$ and
\begin{equation*}
  -\frac{\gamma}{2}\leq\alpha\leq2-\frac{\gamma}{2}.
\end{equation*}
Lemma~\ref{BS} then gives
\begin{equation}\label{imp}
(M_{\lambda}^{d})^{-1+\frac{s}2}M^{od}_{\lambda}(M_{\lambda}^{d})^{-\frac{s}2}
\in \B\bigl(L_{b}^{2,\alpha-s}(\RE^{3(n-1)}),
L_{b}^{2,\alpha-s+\gamma}(\RE^{3(n-1)})\bigr),
\qquad \lambda<\thr_n.
\end{equation}
In particular, taking $s=1$ and $\alpha=1$, for all
$\lambda<\thr_n$ one has
\begin{equation}\label{bl}
B_{\lambda}:
=(M_{\lambda}^{d})^{-\frac{1}2}M^{od}_{\lambda}(M_{\lambda}^{d})^{-\frac{1}2} 
\in \B\bigl(L_{b}^{2}(\RE^{3(n-1)}),
L_{b}^{2,\gamma}(\RE^{3(n-1)})\bigr)
\subset\B(L_{b}^{2}(\RE^{3(n-1)})).
\end{equation} 
For fixed $n$, we write $B_\lambda\equiv B_{n,\lambda}$, in the notation
of the Introduction.
\begin{equation}\label{bl1}
B_{\lambda}\phi(K)=g^2 {m}^{-\frac{1}2}_\lambda(K)\sum_{j=1}^{n-1}   \int_{\RE^3} \frac{{m}^{-\frac{1}2}_\lambda(\hat K_{j,p})\phi(\hat K_{j,p}) }{ |p|^2+| K|^{2}+nm-\lambda  }
  \, dp\,.
\end{equation} 
\begin{lemma}\label{symm} 
$B_{\lambda}$ is a bounded self-adjoint nonnegative operator in
$L_{b}^{2}(\RE^{3(n-1)})$.
\end{lemma}
\begin{proof} Let $\phi,\varphi\in
L_{b}^{2,1}(\RE^{3(n-1)})$.  Then
$(M_{\lambda}^{d})^{-1/2}\phi$ and
$(M_{\lambda}^{d})^{-1/2}\varphi$ belong to
$L_{b}^{2,2}(\RE^{3(n-1)})$.  The explicit formula for
$M_{\lambda}^{od}$ gives its symmetry on this domain, and therefore
\begin{align*}
&\left\langle
(M_{\lambda}^{d})^{-\frac12}M_{\lambda}^{od}
(M_{\lambda}^{d})^{-\frac12}\phi,\varphi\right\rangle \\
&\qquad=
\left\langle\phi,
(M_{\lambda}^{d})^{-\frac12}M_{\lambda}^{od}
(M_{\lambda}^{d})^{-\frac12}\varphi\right\rangle.
\end{align*}
Since $L_{b}^{2,1}$ is dense and $B_\lambda$ is bounded by
\eqref{bl}, the identity extends to all of $L_b^2$.  Thus $B_\lambda$ is
self-adjoint.

To prove nonnegativity, use
\begin{equation*}
  \frac1{x}=\int_0^{+\infty}e^{-tx}\,dt,
  \qquad x>0.
\end{equation*}
For $\psi\in L_b^{2,1}(\RE^{3(n-1)})$, set
$u=(M_\lambda^d)^{-1/2}\psi$.  The explicit formula for
$M_\lambda^{od}$, the above representation and Fubini's theorem give
\begin{equation*}
\begin{aligned}
  \langle\psi,B_\lambda\psi\rangle
  & =\langle u,M_\lambda^{od}u\rangle\\
  & =g^2\sum_{j=1}^{n-1}\int_0^{+\infty}
  \int_{\RE^{3(n-2)}}
  e^{-t(|\hat K_j|^2+nm-\lambda)}
  \left|\int_{\RE^3}e^{-t|k_j|^2}
  u(\hat K_j,k_j)\,dk_j\right|^2
  d\hat K_j\,dt
  \geq0.
\end{aligned}
\end{equation*}
The case $n=2$ is understood with the variable $\hat K_j$ absent.  By the
density of $L_b^{2,1}$ and the boundedness of $B_\lambda$, the inequality
extends to every $\psi\in L_b^2$.
\end{proof}
\begin{lemma}\label{sp}
\begin{equation}\label{6.9.1}
1\notin\sigma(B_{\lambda})\quad\Rightarrow\quad
\lambda\notin \sigma(H^{[n]}_{\mu})\cap(-\infty,\thr_n),
\end{equation}
\begin{equation}\label{6.9.2}
1\notin\sigma_{\mathrm{ess}}(B_{\lambda})\quad\Rightarrow\quad
\lambda\notin \sigma_{\mathrm{ess}}(H^{[n]}_{\mu})\cap(-\infty,\thr_n),
\end{equation}
\begin{equation}\label{6.9.3}
1\notin\sigma_{c}(B_{\lambda})\quad\Rightarrow\quad\lambda\notin \sigma_{c}(H^{[n]}_{\mu})\cap(-\infty,\thr_n) \,,
\end{equation}
\begin{equation}\label{6.9.4}
1\in\sigma_{p}(B_{\lambda})\quad\Leftrightarrow\quad
\lambda\in \sigma_{p}(H^{[n]}_{\mu})\cap(-\infty,\thr_n).
\end{equation}
\begin{equation}\label{6.9.5}
1\in\sigma_{d}(B_{\lambda})\quad\Rightarrow\quad\lambda\in \sigma_{d}(H^{[n]}_{\mu})\cap(-\infty,\thr_n) \,.
\end{equation}
\end{lemma}
\begin{proof} We use the characterization of the spectrum of $H^{[n]}_{\mu}$ below $nm$ given by Corollary \ref{appcert:cor:M-reduction}\ref{item:spcor1}.

Take $\lambda <\thr_n$ and 
 let  $\{\psi_{\ell}\}\subset L_{b}^{2,2}(\RE^{3(n-1)})$ be a Weyl sequence for $(0,M^{d}_{\lambda}-M_{\lambda}^{od})$, i.e., $\|\psi_{\ell}\|=1$ and $\|(M^{d}_{\lambda}-M_{\lambda}^{od})\psi_{\ell}\|\to 0$. By Eq. \eqref{Mlambda-boundedness}, 
  $\phi_{\ell}:=(M_{\lambda}^{d})^{\frac{1}2}\psi_{\ell}\in 
L_{b}^{2,1}(\RE^{3(n-1)})\subset L_{b}^{2}(\RE^{3(n-1)})$ and one has 
\begin{equation}\label{rel}
(1-B_{\lambda})\phi_{\ell}=(M_{\lambda}^{d})^{-\frac{1}2}
(M^{d}_{\lambda}-M_{\lambda}^{od})(M_{\lambda}^{d})^{-\frac{1}2}\phi_{\ell}
=(M_{\lambda}^{d})^{-\frac{1}2}(M^{d}_{\lambda}-M_{\lambda}^{od})\psi_{\ell}\,.
\end{equation} 
Hence
\begin{equation*}
\|(1-B_{\lambda})\phi_{\ell}\|\le \|(M_{\lambda}^{d})^{-\frac12}\|\,\|(M^{d}_{\lambda}-M_{\lambda}^{od})\psi_{\ell}\|\to 0\,.
\end{equation*}
Set
\begin{equation*}
  c_\lambda:=\inf_K m_\lambda(K)>0,
\end{equation*}
where positivity follows from \eqref{mK_lowerbound}. Then
\begin{equation*}
  \|\phi_\ell\|
  =\|(M_\lambda^d)^{1/2}\psi_\ell\|
  \geq\sqrt{c_\lambda}\,\|\psi_\ell\|
  =\sqrt{c_\lambda}.
\end{equation*}
Consequently, with
$\widetilde\phi_\ell:=\phi_\ell/\|\phi_\ell\|$, the sequence
$\{\widetilde\phi_\ell\}$ is a Weyl sequence for $(1,B_\lambda)$, which proves \eqref{6.9.1}.

To prove \eqref{6.9.2}, 
 assume in addition that
$\psi_\ell\rightharpoonup0$. For every
$\varphi\in L_b^{2,1}(\RE^{3(n-1)})$ one has
\begin{equation*}
  \langle\varphi,\phi_\ell\rangle
  =\langle(M_\lambda^d)^{1/2}\varphi,\psi_\ell\rangle
  \longrightarrow0.
\end{equation*}
Since $L_b^{2,1}$ is dense, the normalized sequence also satisfies
$\widetilde\phi_\ell\rightharpoonup0$.

As regards the point spectrum, see Eq. \eqref{6.9.4}, Eq. 
\eqref{rel} shows that if
$\psi_0\in L_b^{2,2}(\RE^{3(n-1)})$ and
$(M_\lambda^d-M_\lambda^{od})\psi_0=0$, then
$\phi_0=(M_\lambda^d)^{1/2}\psi_0$ satisfies
$\phi_0=B_\lambda\phi_0$. Conversely, suppose that
$\phi_0\in L_b^2(\RE^{3(n-1)})$ and
$\phi_0=B_\lambda\phi_0$. Choose $\gamma=1/2$ in \eqref{imp}.
Taking successively $(s,\alpha)=(1,1)$ and $(s,\alpha)=(1,3/2)$ gives
\begin{equation*}
  \phi_0\in L_b^{2,1/2}(\RE^{3(n-1)})
  \quad\text{and then}\quad
  \phi_0\in L_b^{2,1}(\RE^{3(n-1)}).
\end{equation*}
Thus
\begin{equation*}
  \psi_0:=(M_\lambda^d)^{-1/2}\phi_0
  \in L_b^{2,2}(\RE^{3(n-1)}),
\end{equation*}
and \eqref{rel} yields
$(M_\lambda^d-M_\lambda^{od})\psi_0=0$.

To prove Eq. \eqref{6.9.3}, assume that $\lambda\in \sigma_{c}(H^{[n]}_{\mu})\cap(-\infty,\thr_n)$. By Eq. \eqref{6.9.2} it should be $1\in\sigma_{\mathrm{ess}}(B_{\lambda})$. Since, by Eq.\eqref{6.9.4}, $1$ cannot be an eigenvalue of $B_{\lambda}$ there follows $1\in\sigma_{\mathrm{c}}(B_{\lambda})$. 

To conclude, $1\in\sigma_{d}(B_{\lambda})$ implies, by Eqs. \eqref{6.9.2} and \eqref{6.9.4},  $\lambda \notin \sigma_{\mathrm{ess}}(H^{[n]}_{\mu})\cap(-\infty,\thr_n)$ and $\lambda\in \sigma_{p}(H^{[n]}_{\mu})\cap(-\infty,\thr_n)$, hence $\lambda\in \sigma_{d}(H^{[n]}_{\mu})\cap(-\infty,\thr_n)$ which proves Eq. \eqref{6.9.5}.

\end{proof}

\begin{theorem}\label{ub-n}
Let $n\geq2$. If
\begin{equation*}
\mu_g<m\qquad\text{and}\qquad
\frac{\pi^2g^2}{\sqrt{m-\lambda_1}}\leq1,
\end{equation*}
or if $\mu_g\geq m$, then
\begin{equation*}
\sigma(H_\mu^{[n]})
=\sigma_{\mathrm{ess}}(H_\mu^{[n]})
=[\thr_n,+\infty).
\end{equation*}
\end{theorem}
\begin{proof}
Fix $\lambda<\thr_n$. The case $g=0$ is immediate since $B_\lambda=0$;
suppose henceforth that $g\neq0$.

Define 
\begin{equation*}
\gamma_g:=
  \begin{cases}
 0&\mu_g\geq m,\\
 m-\lambda_1&\mu_g<m,
  \end{cases} 
\end{equation*}
With this definition $nm = \thr_n +\gamma_g$. Set
$c:=nm-\lambda=\thr_n-\lambda+\gamma_g>0$.
We first take $\phi$ smooth, compactly supported and symmetric; the resulting
form estimate extends to $L_b^{2,1}$ by density.

Setting $\hat K := (k_1,\dots,k_{n-2})$ and, in the $n=2$ case,  $\phi(\hat K,\cdot)\equiv \phi$,  by $\frac1{x}=\int_{0}^{+\infty}e^{-tx}dt$ one has
\begin{equation*}
\begin{aligned}
\langle\phi,M_\lambda^{od}\phi\rangle 
= & (n-1) g^2 \int_{\RE^{3(n-2)}}\int_{\RE^{3}}  \int_{\RE^{3}} \frac{\overline{\phi}(\hat K,k) \phi(\hat K,p)\,dp\,dk\,d\hat K}{|\hat K|^2+|k|^2+|p|^2+\thr_n -\lambda+ \gamma_g}   
\\
= & (n-1) g^2 \int_0^{+\infty} \left(\int_{\RE^{3(n-2)}} e^{-t(|\hat K|^2+\thr_n -\lambda+ \gamma_g)} \left| \int_{\RE^{3}} e^{-t |k|^2} \phi(\hat K,k)\,dk  \right|^2 d\hat K\right) dt
\end{aligned}
\end{equation*}
which shows that $\langle\phi,M_\lambda^{od}\phi\rangle$ is real valued and nonnegative, so that 
\begin{equation*}
\langle\phi,M_\lambda^{od}\phi\rangle 
\leq  (n-1) g^2 \int_{\RE^{3(n-2)}}\int_{\RE^{3}}  \int_{\RE^{3}} \frac{|\phi(\hat K,k)|\, | \phi(\hat K,p)|\,dp\,dk\,d\hat K}{|\hat K|^2+|k|^2+|p|^2+\thr_n -\lambda+ \gamma_g}   .
\end{equation*}
In the integral we use the inequality 
\begin{equation*}
|\phi(\hat K,k)|\, |\phi(\hat K,p)|
\leq\frac12\left(
\frac{|k|^2}{|p|^2}|\phi(\hat K,k)|^2
+\frac{|p|^2}{|k|^2}|\phi(\hat K,p)|^2\right)
\end{equation*}
and get two terms that give exactly the same contribution after exchanging $k$ and $p$:
\begin{equation*}
\begin{aligned}
\langle\phi,M_\lambda^{od}\phi\rangle 
\leq&  (n-1) g^2 \int_{\RE^{3(n-2)}}\int_{\RE^{3}}  \int_{\RE^{3}}
 \frac{|k|^2 |\phi(\hat K,k)|^2\,dp\,dk\,d\hat K}{|p|^2(|\hat K|^2+|k|^2+|p|^2+\thr_n -\lambda+ \gamma_g)}    \\ 
 =& \sum_{j=1}^{n-1} g^2 \int_{\RE^{3(n-1)}} \int_{\RE^{3}}
 \frac{|k_j|^2 |\phi(K)|^2\,dp\,d K}{|p|^2(| K|^2+|p|^2+\thr_n -\lambda+ \gamma_g)}  \\
=&2\pi^2g^2\int_{\RE^{3(n-1)}}
\frac{|K|^2}{\sqrt{|K|^2+\thr_n -\lambda+ \gamma_g}}\,|\phi(K)|^2\,dK,
\end{aligned}
\end{equation*}
where we used $\sum_{j=1}^{n-1}|k_j|^2=|K|^2$ and, by integration in spherical coordinates,
\begin{equation*}
\int_{\RE^3}\frac{dp}{|p|^2(|p|^2+|K|^2+c)}
=4\pi\int_0^{+\infty}\frac{dr}{r^2+|K|^2+c}
=\frac{2\pi^2}{\sqrt{|K|^2+c}}.
\end{equation*}
By Lemma~\ref{symm}, $B_\lambda\geq0$. Hence, by \eqref{bl}, one gets 
\begin{equation*}
\sup\sigma(B_{\lambda})=\|B_{\lambda}\|
=\sup_{\|\phi\|=1}\langle(M_{\lambda}^{d})^{-1/2}\phi,M_\lambda^{od}(M_{\lambda}^{d})^{-1/2}\phi\rangle.
\end{equation*}
Recalling  that ${m}_\lambda(K)\equiv{m}_\lambda(|K|)=p_{2}\big(\sqrt{|K|^{2}+nm-\lambda}\,\big)$ there follows 
\begin{equation*}
\|B_{\lambda}\| \leq 2\pi^2g^2 \sup_{r>0}\,\frac{r^{2}}{p_{2}\big(\sqrt{r^2+nm-\lambda}\,\big)\sqrt{r^{2} +\thr_n -\lambda+ \gamma_g}}.
\end{equation*}

First suppose that $\mu_g<m$. Put $a:=\sqrt{m-\lambda_1}$ and
$b:=\pi^2g^2$. Then $c>a^2$ and, by \eqref{roots},
$p_2(x)=(x-a)(x+a+2b)$. With $x=\sqrt{r^2+c}$, the preceding bound gives
\begin{align*}
\|B_\lambda\|
&\leq 2b\sup_{x\geq\sqrt c}
\frac{x^2-c}{x(x-a)(x+a+2b)}\\
&\leq 2b\sup_{x\geq\sqrt c}
\frac{x+a}{x(x+a+2b)}.
\end{align*}
The last function is strictly decreasing, since
\begin{equation*}
\frac{d}{dx}\frac{x+a}{x(x+a+2b)}
=-\frac{(x+a)^2+2ab}{x^2(x+a+2b)^2}<0.
\end{equation*}
Consequently,
\begin{equation}\label{eq:BS-small-gap-bound}
\|B_\lambda\|
\leq\frac{2b(\sqrt c+a)}{\sqrt c(\sqrt c+a+2b)}
<\frac{2b}{a+b}.
\end{equation}
Here the strict inequality follows from $\sqrt c>a$. Thus $a\geq b$
gives $\|B_\lambda\|<1$, including the endpoint $a=b$.

If  $\mu_g\geq m$ then  $\gamma_g =  0 $, $\thr_n=nm$ and there holds $p_{2}(x)\geq x(x+2\pi^{2}g^{2})$. Hence,
\begin{equation}\label{eq:BS-large-gap-bound}
\begin{aligned}
\|B_{\lambda}\| \leq& 2\pi^2g^2 \sup_{r>0}\,\frac{r^{2}}{p_{2}\big(\sqrt{r^2+nm-\lambda}\,\big)\sqrt{r^{2} +\thr_n -\lambda+ \gamma_g}} \\ 
 \leq & 2\pi^2g^2 \sup_{r>0}\,\frac{r^{2}}{(r^2+nm-\lambda) (\sqrt{r^{2} +nm -\lambda} + 2\pi^2g^2)}\\ 
 \leq & \frac{ 2\pi^2g^2}{(\sqrt{nm -\lambda} + 2\pi^2g^2)}
<1.
\end{aligned}
\end{equation}
This bound also holds for $g=0$ and includes the case $\mu_g=m$.

In any case, Lemma~\ref{sp} yields
$\sigma(H_\mu^{[n]})\cap(-\infty,\thr_n)=\varnothing$.
The reverse inclusion
$[\thr_n,+\infty)\subset\sigma_{\mathrm{ess}}(H_\mu^{[n]})$
follows from Theorem~\ref{thm:free-essential-spectrum} when
$\mu_g\geq m$, and from
Corollary~\ref{cor:essential-spectrum-sector-one} when $\mu_g<m$.
This proves the stated equality.

In these regimes the expected HVZ formula also follows. Indeed,
Theorem~\ref{thm:n1} and the equality just proved give $E_r=rm$ for
$r\geq0$ when $\mu_g\geq m$. When $\mu_g<m$ and $b\leq a$, they give
$E_0=0$ and $E_r=\lambda_1+(r-1)m$ for $r\geq1$. Hence, in both cases,
\begin{equation*}
\tau_n=\min_{0\leq r<n}\{E_r+(n-r)m\}=\thr_n,
\qquad n\geq2,
\end{equation*}
and $\sigma_{\mathrm{ess}}(H_\mu^{[n]})=[\tau_n,+\infty)$.
\end{proof}

The next Lemmata are ingredients for the proof of the successive Theorem.
\begin{lemma}\label{mon} The function $(-\infty, \thr_n)\ni\lambda\mapsto \|B_{\lambda}\|$ is continuous and non decreasing. Moreover, if $\lim_{\lambda\nearrow \thr_n}\|B_{\lambda}\|>1$ then $\|B_{\lambda}\|=1$ for some $\lambda<\thr_n$.
\end{lemma}
\begin{proof} Let $I\Subset(-\infty,\thr_n)$.  In a complex
neighbourhood $U$ of $I$, the branch of $m_z^{-1/2}$ occurring in
\eqref{bl1} is uniformly defined, and the weighted estimates used above are
locally uniform in $z$.  Thus $z\mapsto B_z$ is locally bounded in
$\B(L_b^2(\RE^{3(n-1)}))$.  Formula \eqref{bl1} shows that
$z\mapsto\langle\phi,B_z\varphi\rangle$ is analytic for every
$\phi,\varphi\in L_b^2$.  Weak analyticity together with local boundedness
therefore gives operator-norm analyticity.  In particular,
$\lambda\mapsto\|B_\lambda\|$ is continuous on
$(-\infty,\thr_n)$.
By Lemma~\ref{symm}, $B_\lambda\geq0$.  Since both
$\lambda\mapsto m_{\lambda}^{1/2}(K)$ and
$\lambda\mapsto |p|^{2}+|K|^{2}+nm-\lambda$ are decreasing on
$(-\infty,\thr_n)$, formula \eqref{bl1} gives, whenever
$\lambda'\leq\lambda''$,
\begin{equation*}
\|B_{\lambda'}\|=\sup_{\|\phi\|=1}\langle \phi, B_{\lambda'}\phi\rangle\le 
\sup_{\|\phi\|=1}\langle |\phi|, B_{\lambda'}|\phi|\rangle\le
\sup_{\|\phi\|=1}\langle |\phi|, B_{\lambda''}|\phi|\rangle\le\|B_{\lambda''}\|\,.
\end{equation*}

The bounds \eqref{eq:BS-small-gap-bound} and \eqref{eq:BS-large-gap-bound}
hold without the small-coupling hypothesis and tend to zero as
$\lambda\to-\infty$, since $c=nm-\lambda\to+\infty$.
Thus $\|B_\lambda\|\to0$, and the last assertion follows by continuity.

\end{proof}
\begin{lemma}\label{n-sp}
For all $\lambda<\thr_n$,
\begin{equation*}
  \|B_{\lambda}\|\in\sigma(B_{\lambda}).
\end{equation*}
If $g\neq0$ and $\|B_{\lambda}\|\in\sigma_{p}(B_{\lambda})$, then this
eigenvalue is simple.  If $n=2$ and $g\neq0$, then
$\|B_{\lambda}\|$ is always a simple isolated eigenvalue.
\end{lemma}
\begin{proof}
By Lemma~\ref{symm}, $B_\lambda$ is bounded, self-adjoint and
nonnegative.  Therefore,
\begin{equation*}
  \|B_\lambda\|=\max\sigma(B_\lambda)\in\sigma(B_\lambda).
\end{equation*}

For $n>2$, $B_\lambda$ itself is not positivity improving. Put $N=n-1$,
let $A\subset\RE^3$ be a ball, and set
\begin{equation*}
  \phi(K)=\prod_{i=1}^{N}\boldsymbol{1}_A(k_i).
\end{equation*}
This is a nonzero nonnegative symmetric function. If at least two components
of $K$ lie outside $A$, replacing only one component cannot bring $K$ into
the support of $\phi$, and hence $(B_\lambda\phi)(K)=0$.

Nevertheless, first extend \eqref{bl1} to the full space
$L^2(\RE^{3N})$, replacing the $j$-th variable at its original position.
The individual summands are bounded by the same weighted estimates, and
their sum leaves the symmetric subspace invariant. On the full space write
\begin{equation*}
  B_\lambda=\sum_{j=1}^{N}T_j,
\end{equation*}
where $T_j$ is the $j$-th summand in this extension of \eqref{bl1}. If $g\neq0$, each $T_j$
has a strictly positive kernel in the variable it replaces. In the expansion
of $B_\lambda^N$ there is, for example, the term
$T_1T_2\cdots T_N$, which replaces all $N$ variables and therefore has a
strictly positive integral kernel on
$\RE^{3N}\times\RE^{3N}$. All the remaining terms are nonnegative. Thus $B_\lambda^{\,n-1}$ is positivity improving, also upon restriction
to the symmetric subspace.

If $\|B_\lambda\|$ is an eigenvalue, then
$\|B_\lambda\|^{n-1}$ is the maximal eigenvalue of
$B_\lambda^{\,n-1}$ and is simple by
\cite[Corollary to Theorem 5.2, Chapter V]{Scha}. Since $B_\lambda\geq0$,
the spectral calculus gives
\begin{equation*}
 \ker\bigl(B_\lambda^{\,n-1}-\|B_\lambda\|^{n-1}\bigr)
 =
 \ker\bigl(B_\lambda-\|B_\lambda\|\bigr),
\end{equation*}
so $\|B_\lambda\|$ is simple as an eigenvalue of $B_\lambda$ as well.
For $n=2$, $B_\lambda$ is an integral operator with kernel
\begin{equation}\label{kernel}
{\mathcal K}_{\lambda}(k,p)=g^{2}\,
\frac{m^{-1/2}_{\lambda}(k)\,m^{-1/2}_{\lambda}(p)}
{|k|^{2}+|p|^{2}+2m-\lambda}\,.
\end{equation}
By \eqref{kernel},
\begin{equation*} 
{\mathcal K}_{\lambda}(k,p)\leq C_\lambda \frac{1}{(|k|^{2}+1)^{1/2}(|p|^{2}+1)^{1/2}(|k|^{2}+|p|^{2}+1)}\le \frac{C_\lambda}{(|k|^{2}+1)(|p|^{2}+1)}
\end{equation*} 
and so ${\mathcal K}_{\lambda}\in L^{2}(\RE^{6})$. Therefore, in the case $n=2$,  $B_{\lambda}$ is Hilbert-Schmidt and hence compact. If $g\neq0$, its kernel is strictly positive; Jentzsch's theorem, see \cite[Theorem 6.6, Chapter V]{Scha},  implies that $\|B_{\lambda}\|$ is a simple isolated eigenvalue.
\end{proof}
\begin{remark}\label{rem:caveat}
For $n>2$, the lemma does not assert that $\|B_\lambda\|$ is an
eigenvalue. Even when it is an eigenvalue, simplicity alone does not imply
that it is isolated. Thus Lemma \ref{mon} 
alone is not sufficient to show 
discrete spectrum in the higher sectors.
\end{remark}
\begin{theorem}\label{n=2eigenvalues} Let $n=2$. If $\mu_g<m$ and
\begin{equation*}
  \frac{\pi^{2}g^{2}}{\sqrt{m-\lambda_{1}}}
  >\frac{\pi}{4-\pi},
\end{equation*}
then $H_{\mu}^{[2]}$ has a simple isolated eigenvalue
$\lambda_{2}<m+\lambda_{1}$.
\end{theorem}
\begin{proof} Take
$\lambda<\thr_2=m+\lambda_1<2m$. We use the trial function 
 \begin{align*}
 &\phi_{\lambda}(k):=m_{\lambda}^{1/2}(k)(|k|^{2}+m+\lambda_{1}-\lambda)^{-1}(|k|^{2}+2m-\lambda)^{-1/2}
\\
&=p_{2}\big(\sqrt{|k|^{2}+2m-\lambda}\,\big)^{1/2}(|k|^{2}+m+\lambda_{1}-\lambda)^{-1}(|k|^{2}+2m-\lambda)^{-1/2} \,.
 \end{align*} 
 Recall that  $p_{2}(x)=(x-x_{+})(x-x_{-})$, where the $x_{\pm}$'s are given in \eqref{roots}; in particular, we will  use the identity $x_{+}^2=m-\lambda_{1}$. By  
 \begin{equation*}
 (|k|^{2}+m+\lambda_{1}-\lambda) = \big(\sqrt{|k|^{2}+2m-\lambda} -x_+\big)\big(\sqrt{|k|^{2}+2m-\lambda} +x_+\big),
 \end{equation*}
 one has,
 \begin{align*}
 \phi_{\lambda}(k)=&
 \frac1{\big(\sqrt{|k|^{2}+2m-\lambda}-x_{+}\big)^{1/2}}\ 
\frac{\big(\sqrt{|k|^{2}+2m-\lambda}-x_{-}\big)^{1/2}}{(\sqrt{|k|^{2}+2m-\lambda}+x_{+})\sqrt{|k|^{2}+2m-\lambda}}
\,.
 \end{align*}
More explicitly, since $x_-=-x_+-2\pi^2g^2$ and $2m-\lambda>x_+^2$,
\begin{equation*}
\begin{aligned}
|\phi_\lambda(k)|^2
&=\frac{1+2\pi^2g^2/(\sqrt{|k|^2+2m-\lambda}+x_+)}
{(|k|^2+m+\lambda_1-\lambda)(|k|^2+2m-\lambda)}\\
&\leq\frac{1+\pi^2g^2/x_+}{|k|^2(|k|^2+x_+^2)}.
\end{aligned}
\end{equation*}
The right-hand side is independent of $\lambda$ and integrable over
$\RE^3$, and hence supplies the majorant needed below.
 
By the dominated convergence theorem and by the identity
\begin{equation*}
r^{2}=((r^{2}+x_{+}^{2})^{1/2}-x_{+})((r^{2}+x_{+}^{2})^{1/2}+x_{+})\,,
\end{equation*}
\begin{align*}
&\lim_{\lambda\nearrow m+\lambda_{1} } \|\phi_{\lambda}\|^{2}=
\lim_{\lambda\nearrow m+\lambda_{1} } \int_{\RE^{3}}\frac{m_{\lambda}(k)\,dk}{
(|k|^{2}+m+\lambda_{1}-\lambda)^{2}(|k|^{2}+2m-\lambda)}
\\
=&
4\pi\lim_{\lambda\nearrow m+\lambda_{1} } \int_{0}^{+\infty}\frac{p_{2}(\sqrt{r^{2}+2m-\lambda}\,)\,r^{2}\, dr}{
(r^{2}+m+\lambda_{1}-\lambda)^{2}(r^{2}+2m-\lambda)}
=
4\pi \int_{0}^{+\infty}\frac{p_{2}(\sqrt{r^{2}+m-\lambda_{1}}\,)\, dr}{
r^{2}(r^{2}+m-\lambda_{1})}
\\
=&4\pi \int_{0}^{+\infty}\frac{((r^{2}+x_{+}^{2})^{1/2}-x_{-})\, dr}{
((r^{2}+x_{+}^{2})^{1/2}+x_{+})(r^{2}+x_{+}^{2})}
=
4\pi\left(\frac{\pi}{2\,x_{+}}-\left(\frac{\pi}2-1\right)\,\frac{x_{+}+x_{-}}{x^{2}_{+}}\right)
\\
=&\frac{2\pi^{2}}{\sqrt{m-\lambda_{1}}}\left(1+2\left(1-\frac2{\pi}\right)\,\,\frac{\pi^{2}g^{2}}{\sqrt{m-\lambda_{1}}}\right)
\,.
\end{align*}
Furthermore, by the monotone convergence theorem
\begin{align*}
\lim_{\lambda\nearrow m+\lambda_{1} }&\langle \phi_{\lambda},B_{\lambda}\phi_{\lambda}\rangle
\\
= g^{2}
\lim_{\lambda\nearrow m+\lambda_{1} }&\int_{\RE^{3}}\int_{\RE^{3}}  
\frac{ (|k|^{2}+m+\lambda_{1}-\lambda)^{-1}(|p|^{2}+m+\lambda_{1}-\lambda)^{-1} \,dp\,dk}{ (|k|^{2}+2m-\lambda)^{1/2} (|p|^{2}+2m-\lambda)^{1/2}
(|p|^{2}+|k|^{2}+2m-\lambda)  }  
  \\
  = g^{2}&
\int_{\RE^{3}}\int_{\RE^{3}}  
\frac{ |k|^{-2}|p|^{-2} \,dp\,dk}{
(|k|^{2}+m-\lambda_{1})^{1/2}
(|p|^{2}+m-\lambda_{1})^{1/2}
(|p|^{2}+|k|^{2}+m-\lambda_{1})}
\\
  = \frac{16\,\pi^{2} g^{2}}{m-\lambda_{1}}&\int_{(0,+\infty)^{2}} 
 \frac{  drd\varrho}{ (r^{2}+1)^{1/2}(\varrho^{2}+1)^{1/2}(r^{2}+\varrho^{2}+1)  }
=\frac{2\pi^{4} g^{2}}{m-\lambda_{1}}\,.
 \end{align*} 
It was this simple result that dictated the choice of the trial function $\phi_{\lambda}$. 
Therefore, 
 \begin{align}
\lim_{\lambda\nearrow m+\lambda_{1} }\|B_{\lambda}\|\ge &\lim_{\lambda\nearrow m+\lambda_{1} }\frac{\langle \phi_{\lambda},B_{\lambda}\phi_{\lambda}\rangle}{\| \phi_{\lambda}\|^{2}}=
\frac{\lim_{\lambda\nearrow m+\lambda_{1} }\langle \phi_{\lambda},B_{\lambda}\phi_{\lambda}\rangle}{\lim_{\lambda\nearrow m+\lambda_{1} }\| \phi_{\lambda}\|^{2}}\nonumber
  \\
 =& \frac{\pi^{2}g^{2}}{\sqrt{m-\lambda_{1}}}\left(1+2\left(1-\frac2{\pi}\right)\,\,\frac{\pi^{2}g^{2}}{\sqrt{m-\lambda_{1}}}\right)^{-1}\label{nb1}
 \,.
 \end{align}
Under the hypothesis of the theorem, the right-hand side
of \eqref{nb1} is greater than one. 
Lemma \ref{mon} 
therefore gives
some $\lambda_2<m+\lambda_1$ such that $\|B_{\lambda_2}\|=1$.
By Lemma~\ref{n-sp}, this is a simple isolated eigenvalue of
$B_{\lambda_2}$, and Lemma~\ref{sp} gives a
simple discrete eigenvalue of $H_\mu^{[2]}$.
\end{proof}
\begin{remark}\label{rem:Ibrogimov-comparison}
A related model is studied in \cite{Ibrogimov2018}, and we give here a brief comparison between the respective frameworks and results. 
In the present Lee model, the total excitation number is conserved and $H_\mu^{[2]}$ is the restriction of the Hamiltonian
to the invariant space
\begin{equation*}
  L_b^2(\RE^6)\oplus L_b^2(\RE^3),
\end{equation*}
corresponding to two field bosons in one internal state of the source and one
field boson in the other.  In the regular spin--boson model studied in \cite{Ibrogimov2018}, by contrast,
the interaction does not conserve photon number.  Truncation at two photons
means replacing the full Fock space by the direct sum of its zero-, one-,
and two-photon subspaces.  This whole truncated space is retained.
Under Assumption~(A) of \cite{Ibrogimov2018}, Theorem~1(i) gives
\begin{equation*}
  \sigma_{\mathrm{ess}}(H_\alpha)
  =[m+E(\alpha),+\infty),
\end{equation*}
whereas Theorem~1(ii) proves that the discrete spectrum of $H_\alpha$ is
finite for every coupling.  The latter assertion permits the discrete
spectrum to be empty and it supplies no quantitative coupling criterion for
binding.  Theorem~2 concerns the weak-coupling expansion of the essential
threshold and is not an existence theorem. 

So, while treating different models, the present and Ibrogimov results are in some sense complementary.  Ibrogimov proves finiteness of
the entire discrete spectrum at arbitrary coupling, while the present
analysis gives explicit absence criteria in every sector and an existence
criterion in the two-excitation sector.  We finally notice that his argument on finiteness of the discrete spectrum does not apply directly to the present 
singular renormalized Lee model, being based on the presence of a square
integrable form factor.
\end{remark}

\begin{remark}[Ground states]\label{rem:ground-state}
Let $E_n:=\inf\sigma(H_\mu^{[n]})$, with $E_0=0$, and let
$E_{\mathrm{gs}}:=\inf\sigma(H_\mu)$. Since $H_\mu$ is self-adjoint
and bounded from below, the direct-sum decomposition gives
\begin{equation*}
 E_{\mathrm{gs}}=\inf_{n\geq0}E_n
 =\inf_{\substack{\Psi\in\dom(H_\mu)\\\|\Psi\|=1}}
   \langle\Psi,H_\mu\Psi\rangle\leq0.
\end{equation*}
A ground state is a normalized vector attaining this infimum,
equivalently an eigenvector of $H_\mu$ at $E_{\mathrm{gs}}$.

The weighted estimate in the proof of Theorem~\ref{ub-n} also yields
the rough lower bound
\begin{equation}\label{eq:ground-state-sector-lower-bound}
 E_n\geq\min\{nm,(n-1)m+\mu_g\},\qquad n\geq1.
\end{equation}
Indeed, for $n\geq2$ and
$\lambda<\min\{nm,(n-1)m+\mu_g\}$, put $c=nm-\lambda>0$.
Subtracting the off-diagonal estimate from the diagonal term gives
\begin{equation*}
\begin{aligned}
 \langle\phi,-M_\lambda^{(n-1)}\phi\rangle
 &\geq\int_{\RE^{3(n-1)}}
 \left(|K|^2+(n-1)m-\lambda+\mu_g
       +\frac{2\pi^2g^2c}{\sqrt{|K|^2+c}}\right)|\phi(K)|^2\,dK\\
 &\geq\bigl((n-1)m-\lambda+\mu_g\bigr)\|\phi\|^2>0
 \qquad(\phi\neq0).
\end{aligned}
\end{equation*}
The spectral reduction in Corollary~\ref{appcert:cor:M-reduction}
therefore excludes every such $\lambda$ from the spectrum.
For $n=1$, the bound follows from Theorem~\ref{thm:n1}.
Consequently, $E_n\to+\infty$ as $n\to\infty$, and
$E_{\mathrm{gs}}=\min_{n\geq0}E_n$: only finitely many sectors can
contribute to the ground-state energy. Its attainment by an eigenvector
requires that $E_n$ be an eigenvalue in at least one minimizing sector.

In the absence regimes proved in Theorem~\ref{ub-n}, the answer is explicit.
If $\mu_g\geq m$, then $E_n=nm$ for all $n\geq0$, so the vacuum
$\Omega^\circ\oplus0$ is the unique ground state, up to a phase.
If $\mu_g<m$ and $\pi^2g^2/\sqrt{m-\lambda_1}\leq1$, then
\begin{equation*}
 E_n=\lambda_1+(n-1)m\quad(n\geq1),\qquad
 E_{\mathrm{gs}}=\min\{0,\lambda_1\}.
\end{equation*}
The ground state is the vacuum if $\lambda_1>0$ and the simple
one-excitation eigenstate if $\lambda_1<0$. If $\lambda_1=0$, the
ground-state eigenspace is two-dimensional, spanned by these two states.
In fact, \eqref{eq:ground-state-sector-lower-bound} also identifies the
vacuum as the unique ground state whenever $\mu_g\geq0$, without a
weak-coupling condition: $E_n>0$ for $n\geq2$, and
Theorem~\ref{thm:n1} gives $E_1>0$ (at $\mu_g=0$ we use $g\neq0$).
Under the hypothesis of Theorem~\ref{n=2eigenvalues}, one may choose
$\lambda_2=E_2$. Indeed, take the first value of $\lambda$ at which
$\|B_{2,\lambda}\|=1$. Below it, Lemma~\ref{sp} excludes spectrum,
while at it Lemma~\ref{n-sp} gives a simple isolated eigenvalue.
Thus the two-excitation sector has a simple ground state below
$T_2=m+\lambda_1$. On the other hand, this inequality alone does not allow to compare $\lambda_2$
with $0$, $\lambda_1$, or the higher-sector infima, and therefore does
not allow to identify it as a ground state of the full Hamiltonian.
For general parameters, the existence and the excitation sector of a
ground state of $H_\mu$ remain to be determined. The missing reverse
HVZ inclusion is relevant here: if $N\geq1$ is a minimizing sector,
then
\begin{equation*}
 \tau_N=\min_{0\leq r<N}\{E_r+(N-r)m\}
 \geq E_{\mathrm{gs}}+m.
\end{equation*}
The full identity $\sigma_{\mathrm{ess}}(H_\mu^{[N]})
=[\tau_N,+\infty)$ would therefore make
$E_N=E_{\mathrm{gs}}$ a discrete eigenvalue.
This identity is available for $N=1,2$ and in the absence regimes,
but is not yet established for arbitrary higher sectors.
\end{remark}

\appendix
\section{Approximation by regularized Hamiltonians}
\label{app:reg}
This appendix gives a characterization of the
Hamiltonian constructed intrinsically in Sections~\ref{s:prelude} as a norm resolvent limit of a family of regularized Hamiltonians. 
Notice that none of the spectral arguments and results of the preceding Sections depends on
the approximation proved here.
\subsection{The abstract Lee model}
We refer to the results and notation in Subsection~\ref{s:abstractlee};
for our purposes, it is enough to consider the case
$\F_1^\circ=\F_2^\circ=\F^\circ$ and
$H_1^\circ=H_2^\circ=H^\circ$.
We use the notation
\begin{equation*}
 \lambdastar:=\lambda_{1,\star},
 \qquad R^\circ:=(-H^\circ+\lambdastar)^{-1}.
\end{equation*}
For any $\Lambda>0$, we assume that
$$
a_{\Lambda}:\dom(a_\Lambda)\subseteq\F^{\circ}\to\F^{\circ}
$$ 
is closable,  such that $a_{\Lambda}\in\B(\H^{\circ,1/2},\F^{\circ})$ and
\begin{equation}\label{teo_i}
\begin{bmatrix}0&a_{\Lambda}^{*}\\
a_{\Lambda}&0\end{bmatrix}\quad\text{is}\quad \begin{bmatrix}H^{\circ}&0\\
0&H^{\circ}\end{bmatrix}\text{-small.}
\end{equation}

By the Rellich--Kato theorem, the operators
\begin{equation*}
 H_\Lambda:=\begin{bmatrix}H^\circ&a_\Lambda^*\\
 a_\Lambda&H^\circ\end{bmatrix},
 \qquad \dom(H_\Lambda)=\H^\circ\oplus\H^\circ,
\end{equation*}
are self-adjoint and bounded from below.

\begin{theorem}\label{teo-conv-Lee} Let $S^{\circ}\equiv S_{2}^{\circ}$, $a$ and $H_{S^{\circ}}\equiv H_{S_{2}^{\circ}}$ be as in Theorem \ref{teo}; let $a_{\Lambda}$ as above and suppose that
\begin{equation}\label{h1}
\lim_{\Lambda \nearrow\infty}\|a_{\Lambda}-a\|_{\H^{\circ} ,\F^{\circ}}=0
\end{equation}
and that there exists a symmetric operator $E_{\Lambda}^{\circ}\in\B(\F^{\circ})$ such that
\begin{equation}\label{h3}
\lim_{\Lambda \nearrow\infty}\|E^{\circ}_{\Lambda}-a_{\Lambda}R^{\circ}a_{\Lambda}^{*}-S^{\circ}\|_{\H^{\circ} ,\F^{\circ}}=0\,.
\end{equation}
Then, as $\Lambda \nearrow\infty$,  the bounded-from-below self-adjoint operators
$$
H_{\Lambda}-E_{\Lambda}
:\H^{\circ}\oplus\H^{\circ}\subseteq\F^{\circ}\oplus\F^{\circ}\to\F^{\circ}\oplus\F^{\circ}\,,
$$
$$
(H_{\Lambda}-E_{\Lambda})(\psi_{1}\oplus\psi_{2}):=(H^{\circ}\psi_{1}+a_{\Lambda}^{*}\psi_{2})\oplus((H^{\circ}-E_{\Lambda}^{\circ})\psi_{2}+a_{\Lambda}\psi_{1})
$$
converge in norm resolvent sense to $H_{S^{\circ}}$.
\end{theorem}
\begin{proof}
We use the space $\H^\circ=\dom(H^\circ)$ with its free graph norm
\begin{equation*}
 \|\psi\|_{\H^\circ}
 :=\bigl(\|\psi\|^2+\|H^\circ\psi\|^2\bigr)^{1/2}.
\end{equation*}
Choose a real number
\begin{equation*}
 z<\min\{\lambdastar,\inf\sigma(H^\circ),
                    \inf\sigma(H_{S^\circ})\}.
\end{equation*}
Such a choice is possible by Theorem~\ref{teo}.
For real $w\in\varrho(H^\circ)$, set
\begin{equation*}
 R_w^\circ:=(-H^\circ+w)^{-1},\qquad
 G_{\Lambda,w}^\circ:=(a_\Lambda R_w^\circ)^*,\qquad
 G_w^\circ:=(aR_w^\circ)^*.
\end{equation*}
Since $R_w^\circ\in\B(\F^\circ,\H^\circ)$, hypothesis~\eqref{h1}
implies
\begin{equation}\label{eq:approx-G-convergence}
 G_{\Lambda,w}^\circ\longrightarrow G_w^\circ
 \quad\text{in }\B(\F^\circ),\qquad w=z,\lambdastar.
\end{equation}
Put
\begin{equation*}
 D_\Lambda:=E_\Lambda^\circ-a_\Lambda R_{\lambdastar}^\circ
 a_\Lambda^*-S^\circ.
\end{equation*}
By~\eqref{h3}, $D_\Lambda\to0$ in $\B(\H^\circ,\F^\circ)$.

The reduced operator for the regularized Hamiltonian is
\begin{equation*}
 M_{\Lambda,z}:=-H^\circ+z+E_\Lambda^\circ-a_\Lambda R_z^\circ a_\Lambda^*,
 \qquad\dom(M_{\Lambda,z})=\H^\circ.
\end{equation*}
The resolvent identity gives, on $\H^\circ$,
\begin{align*}
 a_\Lambda(R_{\lambdastar}^\circ-R_z^\circ)a_\Lambda^*
 &=(z-\lambdastar)
   (G_{\Lambda,\lambdastar}^\circ)^*G_{\Lambda,z}^\circ,\\
 M_{\Lambda,z}
 &=-H^\circ+z+S^\circ+D_\Lambda
   +(z-\lambdastar)
    (G_{\Lambda,\lambdastar}^\circ)^*G_{\Lambda,z}^\circ.
\end{align*}
On the same domain, the limiting reduced operator~\eqref{Mz} is
\begin{equation*}
 M_z=-H^\circ+z+S^\circ
       +(z-\lambdastar)(G_{\lambdastar}^\circ)^*G_z^\circ.
\end{equation*}
Consequently, by~\eqref{eq:approx-G-convergence} and~\eqref{h3},
\begin{equation}\label{eq:approx-M-graph-convergence}
 \|M_{\Lambda,z}-M_z\|_{\H^\circ,\F^\circ}\longrightarrow0.
\end{equation}

The factorization~\eqref{eq:abstract-Lee-factorization} and the choice
of $z$ show that $M_z:\H^\circ\to\F^\circ$ is bijective.
This map is bounded with respect to the free graph norm, since
$S^\circ$ is $H^\circ$-bounded and the remaining correction is
bounded on $\F^\circ$. The bounded inverse theorem therefore gives
\begin{equation*}
 M_z^{-1}\in\B(\F^\circ,\H^\circ).
\end{equation*}
Thus~\eqref{eq:approx-M-graph-convergence} implies
\begin{equation*}
 Q_\Lambda:=(M_{\Lambda,z}-M_z)M_z^{-1}\longrightarrow0
 \quad\text{in }\B(\F^\circ).
\end{equation*}
For all sufficiently large $\Lambda$, $1+Q_\Lambda$ is invertible
by a Neumann series. Since $M_{\Lambda,z}=(1+Q_\Lambda)M_z$ on
$\H^\circ$, it follows that
\begin{equation}\label{eq:approx-M-inverse-convergence}
 M_{\Lambda,z}^{-1}=M_z^{-1}(1+Q_\Lambda)^{-1}
 \longrightarrow M_z^{-1}
 \quad\text{in }\B(\F^\circ,\H^\circ).
\end{equation}

Write
\begin{equation*}
 E_\Lambda=\begin{bmatrix}0&0\\0&E_\Lambda^\circ\end{bmatrix},
 \qquad
 \mathcal T_{\Lambda,z}:=
 \begin{bmatrix}1&G_{\Lambda,z}^\circ\\0&1\end{bmatrix}.
\end{equation*}
The triangular operator $\mathcal T_{\Lambda,z}$ is boundedly
invertible on $\F^\circ\oplus\F^\circ$.
For each fixed $\Lambda$, assumption~\eqref{teo_i} implies that
$a_\Lambda^*$ maps $\H^\circ$ into $\F^\circ$.
Hence $G_{\Lambda,z}^\circ=R_z^\circ a_\Lambda^*$ on $\H^\circ$
maps $\H^\circ$ into itself, and $\mathcal T_{\Lambda,z}$ maps
$\H^\circ\oplus\H^\circ$ bijectively onto itself.
Direct block multiplication on this domain gives
\begin{equation*}
 \mathcal T_{\Lambda,z}^*
 (-H_\Lambda+E_\Lambda+z)\mathcal T_{\Lambda,z}
 =\begin{bmatrix}-H^\circ+z&0\\0&M_{\Lambda,z}\end{bmatrix}.
\end{equation*}
For sufficiently large $\Lambda$, both diagonal entries are
invertible, and therefore $z\in\varrho(H_\Lambda-E_\Lambda)$ with
\begin{equation*}
 (-H_\Lambda+E_\Lambda+z)^{-1}
 =\mathcal T_{\Lambda,z}
 \begin{bmatrix}R_z^\circ&0\\0&M_{\Lambda,z}^{-1}\end{bmatrix}
 \mathcal T_{\Lambda,z}^*.
\end{equation*}
By~\eqref{eq:approx-G-convergence} and~\eqref{eq:approx-M-inverse-convergence},
the bounded factors on the right converge in norm. Their product
therefore converges to
\begin{equation*}
 \begin{bmatrix}1&G_z^\circ\\0&1\end{bmatrix}
 \begin{bmatrix}R_z^\circ&0\\0&M_z^{-1}\end{bmatrix}
 \begin{bmatrix}1&0\\(G_z^\circ)^*&1\end{bmatrix}
 =(-H_{S^\circ}+z)^{-1},
\end{equation*}
where the equality follows from Theorem~\ref{Res-Lee}.
Convergence at this common real resolvent point proves the asserted
norm resolvent convergence.
\end{proof}

\subsection{The regularized Lee Hamiltonian: removal of the ultraviolet cutoff} Here we apply the preceding scheme to the Lee Hamiltonian with an ultraviolet cutoff.
\par
For any $\Lambda >0$, let  $\chi_\Lambda \in L^2(\RE^3)$  be  a cut-off function and define the  operator 
\begin{equation*}
a_\Lambda: \dom(a_\Lambda)  \subset \F_b \to \F_b\,,   
\end{equation*}
\begin{equation*}
\dom(a_\Lambda):= \left\{\psi\in \F_b : \sum_{n=0}^\infty (n+1)\int_{\RE^{3n}} dK \left| \int_{\RE^3} dp \,  \overline{\chi_\Lambda}(p) \psi^{(n+1)}(K,p)\right|^2<\infty\right\},
\end{equation*}
\begin{equation*}
(a_\Lambda \psi)^{(0)} := g\int_{\RE^3} \overline{\chi_\Lambda}(p) \psi^{(1)}(p) \, dp\, ,  
\end{equation*}
and
\begin{equation*}
(a_\Lambda \psi)^{(n)}(K) := g\sqrt{n+1}\int_{\RE^3} \overline{\chi_\Lambda}(p) \psi^{(n+1)}(K,p) \, dp\,, \qquad K\in \RE^{3n}\,, \ n\ge 1\,. 
\end{equation*}
The operator $a_\Lambda$ is unbounded, densely defined and closed  on $\F_b$, see, e.g., \cite[Section 5.7]{Arai}. 

The adjoint of $a_\Lambda$ is  
\begin{equation*}
a_\Lambda^*:\dom(a_\Lambda^*)\subset  \F_b \to \F_b 
\end{equation*}
with 
\begin{equation*}
\dom(a_\Lambda^*) = \left\{\psi\in \F_b : \sum_{n=1}^\infty \frac{1}{n}\int_{\RE^{3n}} dK \left| \sum_{j=1}^{n} \chi_\Lambda(k_j) \psi^{(n-1)}(\hat K_j) \right|^2<\infty\right\},
\end{equation*}

\begin{equation*}
(a_\Lambda^* \psi)^{(0)}= 0,
\end{equation*}
and 
\begin{equation*}
(a_\Lambda^* \psi)^{(n)}(K) = \frac{g}{\sqrt{n}} \sum_{j=1}^{n} \chi_\Lambda(k_j) \psi^{(n-1)}(\hat K_j) \qquad K\in \RE^{3n}, n\in\NA.
\end{equation*}
Likewise, $a_\Lambda^*$ is unbounded, densely defined, and closed on $\F_b$.

\begin{remark}\label{Hyp_i}
We remark that, by \cite[Theorem 5.17 and Remark 5.8]{Arai}, there hold
\begin{equation*}
\H_{b}^{1/2} \subset \dom(a_\Lambda) = \dom(a_\Lambda^*);\quad  a_\Lambda \in \B(\H_{b}^{1/2},\F_b)  ;\quad  a_\Lambda^* \in \B(\H_{b}^{1/2},\F_b).   
\end{equation*}
Additionally, see, e.g., \cite[Corollary 5.10]{Arai}, both $a_\Lambda$ and $a_\Lambda^*$ are $H^{\circ}$-small; hence
\begin{equation*}
\text{the operator} \quad 
\begin{bmatrix}0&a_{\Lambda}^{*}\\
a_{\Lambda}&0\end{bmatrix}\quad\text{is}\quad \begin{bmatrix}H^{\circ}&0\\
0&H^{\circ}\end{bmatrix}\text{-small}
\end{equation*}
and  the assumption  \eqref{teo_i}  of Theorem \ref{teo-conv-Lee} is satisfied. 
\end{remark}

By  the Rellich-Kato theorem, for any $\Lambda >0$ and $\mu > 0 $ the Lee Hamiltonian with cutoff $\Lambda$ 
\begin{equation*}
H_{\Lambda,\mu}:\H^{\circ}\oplus\H^{\circ}\subseteq\F^{\circ}\oplus\F^{\circ}\to\F^{\circ}\oplus\F^{\circ}\,,
\end{equation*}
\begin{equation*}
H_{\Lambda,\mu}(\psi_{1}\oplus\psi_{2}):=(H^{\circ}\psi_{1}+a_{\Lambda}^{*}\psi_{2})\oplus((H^{\circ}+\mu)\psi_{2}+a_{\Lambda}\psi_{1})
\end{equation*}
is self-adjoint and  bounded-from-below. 
 
Now,  we notice that $a_\Lambda R_z^\circ a_\Lambda^* \in \B(\F_b)$ for all $z\in \rho(H^{\circ})$ and define the renormalization constant
$$
    \mathscr E^{\circ}_{\Lambda} := \langle\Omega^{\circ},a_\Lambda R^\circ_\nu a^*_\Lambda\Omega^{\circ}\rangle \qquad \nu\in\RE\cap\rho(H^{\circ})\,.
$$
We use the same symbol $\mathscr E^{\circ}_{\Lambda}$ to denote the corresponding multiplication operator and then we define 
$$
E_{\Lambda}:=\begin{bmatrix}0&0\\0&\mathscr E^{\circ}_{\Lambda}
\end{bmatrix}.
$$
Since $\nu<m$, $R_\nu^\circ$ is negative on the one-boson subspace and
$\mathscr E_\Lambda^\circ\leq0$.  Thus subtracting $E_\Lambda$ adds the
positive counterterm $-\mathscr E_\Lambda^\circ$ to the second component.
The parameters $\mu$ and $\nu$ are kept fixed in the limit (equivalently, so
is the renormalized gap $\mu_g$).  Although $E_{\Lambda}$, like $H_{\mu}$, depends on the parameter $\nu\in\RE\cap\rho(H^{\circ})$, we omit such dependence from the notation. The renormalizing operator $E_{\Lambda}$ enters in the following theorem.
 
\begin{theorem}\label{th:convergence} Suppose that, for some $s\in(3/4,1]$, there holds 
$$
\lim_{\Lambda \nearrow\infty}\|\chi_\Lambda - 1\|_{L^2(\RE^3,(|k|^2 +1)^{-2s}dk)}=0\,.
$$ Then 
\begin{equation}\label{limitRLM_1}
\lim_{\Lambda \nearrow\infty}(H_{\Lambda,\mu}-E_{\Lambda})=H_{\mu}\quad\text{in norm resolvent sense.}
\end{equation}
\end{theorem}
\begin{proof}
Hypothesis \eqref{teo_i} is provided by
Remark \ref{Hyp_i}. We verify hypotheses \eqref{h1}--\eqref{h3} of
Theorem \ref{teo-conv-Lee}.  By the definition of $H_{\Lambda,\mu}$, the operator
$E_\Lambda^\circ$ in \eqref{h3} is here multiplication by
$\mathscr E_\Lambda^\circ-\mu$.

For \eqref{h1}, apply \cite[Proposition 3.3]{Lampart2025} with form-factor
difference $F=g(\chi_\Lambda-1)$ and positive resolvent parameter
$-\lambdastar$.  Taking into account the equivalence of the weights
$|k|^2+m$ and $|k|^2+1$, that estimate gives
\begin{equation}\label{eq:cutoff-annihilator-convergence}
 \|(a_\Lambda-a)R_{\lambdastar}^\circ\|_{\B(\F^\circ)}
 \leq
 C_{g,m,s,\lambdastar}
 \|\chi_\Lambda-1\|_{L^2(\RE^3,(|k|^2+1)^{-2s}dk)}.
\end{equation}
The right-hand side tends to zero.  Since
$R_{\lambdastar}^\circ:\F^\circ\to\H^\circ$ is a bounded bijection with
bounded inverse, \eqref{eq:cutoff-annihilator-convergence} is equivalent to
\eqref{h1}.  

Finally, \eqref{h3} is, in the present notation,
\begin{equation*}
 \lim_{\Lambda\nearrow\infty}
 \left\|
 \mathscr E_\Lambda^\circ
 -a_\Lambda R^\circ a_\Lambda^*
 -(\mu+S^\circ)
 \right\|_{\H_b,\F_b}
 =0.
\end{equation*}
This is the normal ordered finite part convergence proved in
\cite{Lampart2025}: the identification follows from the normal ordering
formula in \cite[Lemma 3.9]{Lampart2025}, while the required graph norm
continuity is the estimate in \cite[Proposition 3.1]{Lampart2025} used in
the proof of \cite[Proposition 3.5, immediately before Eq.~(3.6)]{Lampart2025}.
Thus all hypotheses of Theorem \ref{teo-conv-Lee} hold, and
\eqref{limitRLM_1} follows.
\end{proof}

\section*{Acknowledgments.}
\noindent The authors acknowledge the support of the Next Generation EU - Prin 2022 project "Singular Interactions and Effective Models in Mathematical Physics- 2022CHELC7".\\ C.C., D.N. and A.P. also acknowledge the support  of the INdAM-GNFM. 

\end{document}